\documentclass{article}
\usepackage[a4paper, total={6in, 10in}]{geometry}

\usepackage{authblk}
\usepackage{setspace}
\usepackage[utf8]{inputenc}
\usepackage[T1]{fontenc}  
\usepackage{hyperref}      
\usepackage{url}           
\usepackage{booktabs}       
\usepackage{amsmath}
\usepackage{amsfonts}
\usepackage{amsthm}
\usepackage{amssymb}
\usepackage{bbm}
\usepackage{tikz}
\usepackage{nicefrac}      
\usepackage{microtype}       
\usepackage{lipsum}
\usepackage{fancyhdr}       
\usepackage{graphicx}   
\usepackage{algorithm}
\usepackage[noend]{algpseudocode}
\usepackage{mathtools}
\usepackage{natbib}

\usepackage{amsmath}
\usepackage{cleveref}
\usepackage[shortlabels]{enumitem}

\graphicspath{{figures/}}  
\usepackage{mathtools}

\newtheorem{theorem}{Theorem}[section]
\newtheorem{lemma}{Lemma}[section]
\newtheorem{corollary}{Corollary}[section]
\newtheorem{proposition}{Proposition}[section]

\newtheorem{assumption}{Assumption}

\theoremstyle{definition} 
\newtheorem{definition}{Definition}
\newtheorem{example}{Example}
\newtheorem{remark}{Remark}

\title{Flexible Modeling of Information Diffusion on Networks\\ with Statistical Guarantees}

\author{
  Alexander Kagan,  Elizaveta Levina, Ji Zhu \\
  Department of Statistics\\
University of Michigan\\
  \texttt{\{amkagan, elevina, jizhu\}@umich.edu} \\

}
\date{}

\begin{document}
\maketitle
\begin{abstract}
Modeling information spread through a network is one of the key problems of network analysis, with applications in a wide array of areas such as marketing and public health. Most approaches assume that the spread is governed by some probabilistic diffusion model, often parameterized by the strength of connections between network members (edge weights), highlighting the need for methods that can accurately estimate them. Multiple prior works suggest such estimators for particular diffusion models; however, most of them lack a rigorous statistical analysis that would establish the asymptotic properties of the estimator and allow for uncertainty quantification. In this paper, we develop a likelihood-based approach to estimate edge weights from the observed information diffusion paths under the proposed General Linear Threshold (GLT) model, a broad class of discrete-time information diffusion models that includes both the well-known linear threshold (LT) and independent cascade (IC) models. We first derive necessary and sufficient conditions that make the edge weights identifiable under this model. Then, we derive a finite sample error bound for the estimator and demonstrate that it is asymptotically normal under mild conditions. We conclude by studying the GLT model in the context of the Influence Maximization (IM) problem, that is, the task of selecting a subset of $k$ nodes to start the diffusion, so that the average information spread is maximized. We derive conditions that ensure the IM problem can be greedily solved under the proposed GLT model with standard optimality guarantees, and establish the dependency between the accuracy of the GLT weight estimates and the quality of the IM problem solutions. 
% For the special case of the standard LT model, we also present a much faster expectation-maximization approach for weight estimation.
Extensive experiments on synthetic and real-world networks demonstrate that the flexibility of the proposed class of GLT models, coupled with the proposed estimation and inference framework for its parameters, can significantly improve estimation of spread from a given subset of nodes, prediction of node activation, and the quality of the IM problem solutions.
\end{abstract}

\section{Introduction}
    \label{ch:introduction}

The emergence of large-scale online social networks has led to the appearance of rich datasets which include not only connections between users and user features,  but also the paths of information propagation (also called information diffusion) between users. The term ``information'' here is interpreted broadly and can refer to anything that spreads from node to node, be it a news item, a tweet, or a virus. 
% Sometimes, these can also be interconnected -- \cite{covid_twitter} showed that the spread of Covid-19 through people's social networks can be well predicted by their interactions around Covid-related posts on Twitter. 
Information propagation paths, also known as {\it propagation traces}, are especially valuable in modeling information spread, since they provide direct data on the influence users have on their network neighbors.
%This data has found numerous applications in network analysis. 
For example, \cite{fakenews} used propagation traces for fake news detection, while \cite{Saito} used them to estimate the probabilities of information transmission between users, i.e., edge weights, assuming that the propagation follows the Independent Cascade (IC) model \citep{Goldenberg:2001}, a simple and arguably the most popular information diffusion model. 
% As a more general result, with the help of propagation data, \cite{InferringGraphs2015} proposed a sparse recovery framework for weight estimation under the General Cascade model (extension of the IC model) and assuming that the underlying graph is not observed. 
\cite{goyal2011databased} used this data to assign influence credits to users and subsequently solve the influence maximization (IM) problem \citep{richardson, Kempe}, that is, a task of identifying a fixed-size subset of users to ``seed'' the propagation that would result in the highest number of nodes eventually reached. While all these papers use the propagation traces to solve important problems related to influence propagation, the solutions are usually obtained assuming a very specific diffusion model and with no further uncertainty quantification. This calls for developing a general statistical framework enabling flexible modeling, estimation, and uncertainty quantification of a network diffusion process based on these traces, which is the main contribution of this paper.

We introduce a natural but surprisingly rich class of models we call General Linear Threshold (GLT) diffusion models, which includes both the popular Linear Threshold (LT) \citep{granovetter1978threshold} and Independent Cascade (IC) models.  Unlike other flexible generalizations, it is not over-parametrized and thus allows for robust estimation of its parameters from propagation traces and uncertainty quantification for various downstream tasks. Similarly to the LT model, the GLT model assumes that each edge is assigned a deterministic weight, and a node activates when the sum of incoming weights from its currently active parents crosses a node-specific random threshold.  Unlike the LT model, which assumes that all the thresholds are uniformly distributed on $[0,1]$, the GLT model allows the threshold distribution to vary between nodes. This seemingly straightforward  generalization turns out to bring a lot of new flexibility to the model, since it allows for heterogeneity in how readily different nodes accept new information;  some may be much more easily influenced than others.  The popular alternative to the LT, the IC model, which we show corresponds to the GLT model with exponentially distributed thresholds (Proposition \ref{ic_is_glt_propos}), is also not able to allow for this heterogeneity, leaving the GLT model the only option that can.

% The work by \cite{goyal2011databased} supports this claim by showing that solutions of the influence maximization problem obtained under a random assignment of diffusion edge weights are significantly inferior compared to the solutions using (even inaccurately) pre-estimated weights.

 Importantly, the proposed class of GLT models comes with a convenient likelihood form,  allowing for efficient parameter estimation via constrained convex optimization. Under mild regularity conditions, we  establish a finite-sample error bound for this estimator which guarantees consistency, and derive its asymptotic distribution, allowing for construction of asymptotic confidence intervals for the GLT edge weights and their smooth transformations. In section \ref{consist_section}, we demonstrate several applications of this result, including uncertainty quantification for predicted node activation probabilities, testing for the difference in parents' influence on a child node, and solving the Robust Influence Maximization problem, an extension of the standard IM problem where edge weights are not available but are known to lie in a given confidence interval.

 While there have been several papers \citep{learnin_influence_kempe, learnability_influence} establishing Probably Approximately Correct (PAC) learnability guarantees for nodes' activation probabilities under the IC, LT, and several similar models, there has been very little work establishing theoretical guarantees for the diffusion model parameters. The only papers we are aware of that address this are  \cite{learning_diffusion_in_cont_time}, which derived sufficient identifiability conditions and established consistency for the parameters of several continuous-time diffusion models, and \cite{InferringGraphs2015}, which introduced a General Cascade model (which we show can be viewed as a variant of GLT as well) and focused on inferring the unobserved graph structure from the observed cascades.  While they established a bound implying $\sqrt{n}$-consistency for their estimator, they acknowledged their framework cannot be extended to the LT model and leave this for future work.   To the best of our knowledge, this paper is the first to provide a theoretical guarantee for parameter estimates from propagation traces for any threshold-based diffusion model. We conclude our theoretical analysis with a study of the GLT model in the context of the IM problem in  Section \ref{IM_section}, showing that, under the GLT model, the IM problem can be solved using the natural greedy strategy with standard optimality guarantees if all threshold distributions are concave, and that, for a class of graphs, the error rate of the IM solution mirrors the error rate of weight estimates.  %Then, in Proposition \ref{spread_error_bound_propos}, we prove that for the family of directed bipartite graphs, the worst-case discrepancy between the spread of the IM solution under the true and estimated GLT models is proportional to the $\ell_1$ error of the weight estimates. This implies that the finite sample error bound established for the weight estimates can be directly propagated to the error rate of an IM solution under the estimated model. 

The rest of this manuscript is organized as follows.  In Section \ref{ch:methods1}, we fix notation and define important concepts related to information propagation.  In Section \ref{gltm_introduce_section}, we introduce the General Linear Threshold model and show how it relates to other diffusion models. In Section \ref{ch:methods2}, we establish identifiability conditions for the GLT, derive a likelihood approach to weight estimation, and establish conditions for the estimator's consistency and asymptotic normality. We further extend the estimation procedure to the case of partially observed propagation traces and unknown threshold distributions.
 Section \ref{IM_section} presents a short study of the IM problem under the GLT model. Finally, Section \ref{ch:experiments} presents experiments on synthetic and real-world networks showing how the flexibility of the GLT model, coupled with the proposed weight estimation and uncertainty quantification procedure, can significantly improve performance in various downstream tasks, including the IM problem, prediction of node activation probabilities, and spread estimation. 
% Section \ref{ch:conclusion} presents a discussion of results to date,
% , and Section \ref{ch:future_work} outlines plans for future projects.   

\section{A statistical framework for information diffusion models}
    \label{ch:methods1}
	In this section, we set up a general framework for modeling discrete-time influence propagation on a network.  We present all models from the point of view of statistical models dependent on parametric distributions, which may differ from other standard treatments of such models.  We start by setting up notation.

Let $G=(V, E)$ be a graph where $V$ is the set of nodes and $E \subseteq V\times V$ is the set of edges. Unless otherwise stated, we assume throughout this manuscript that $G$ is a simple directed graph, that is, a graph with no self-loops and no repeated edges.    For a directed edge $u \rightarrow v$, we refer to $u$ as the parent node and $v$ as the child node.
We denote the sets of parent and children nodes  of a node $v$ as, respectively,  
$$P(v) = \{u: (u, v) \in E\}\quad \text{ and }\quad  C(v) = \{u: (v, u) \in E\}. $$
Similarly, for any set of nodes $S \subset V$, we write 
$$P(S) := \bigcup_{v\in S} P(v)  \, \setminus \, S \quad \text{ and }\quad  C(S) := \bigcup_{v\in S} C(v) \setminus S.
$$
In words, $P(S)$ and $C(S)$ consist of nodes outside of $S$ which have at least one child or parent in $S$, respectively.  Since we will often additionally assume that each edge is associated with a weight or a transmission probability, it is convenient to fix some edge order for this correspondence to be uniquely defined. For that, we assume that the nodes are enumerated from $1$ to $|V|$ and the edges are sorted lexicographically, by child node first and then by parent node.  For example, a triangle with each edge going in both directions corresponds to the ordered set $E = [(2, 1), (3, 1), (1, 2), (3, 2), (1, 3), (2, 3)]$.

A {\it diffusion model} associated with $G$ is a probabilistic model that governs the spread of information through the graph, usually with some edge- or node-dependent transmission parameters $\boldsymbol{\theta}$. Our primary goal is estimation and inference for these parameters from the observed information diffusion paths.  In general, the discrete-time information diffusion path starts with a given non-empty set $A_0 \subset V$ of initially activated (influenced) nodes, also known as the {\it seed set}. Then, at each time step $t = 1, 2\ldots$, the currently active nodes $A_{t-1}$ try to activate their network children while possibly deactivating themselves, making the set of all active nodes change to $A_t$. The process stops when no node can change its activation status or when the maximal time horizon is reached. In this work, we focus on the (quite general) subtype of discrete-time diffusion processes which satisfying the following assumptions.
\begin{enumerate}
    \item Information propagation is {\it progressive}, which means that once a node is activated, it remains activated.  This implies that $A_{t-1} \subset A_{t}$ for every time step $t$. We define the disjoint sets of {\it newly activated nodes} at times $t\ge 0$ as $D_t:= A_t\setminus A_{t-1}$ with the convention that $A_{-1} = \emptyset$  (no nodes are active before time $0$). 
    \item A child node $v$ can be activated at a time $t \ge 1$ only if it has a newly activated parent, that is, if $D_{t-1} \cap P(v)$ is not empty. Intuitively, this means that if a parent's influence was not enough to activate a child as soon as the parent become active, the only way for the child to be activated later is if another one of its parents becomes active.  
\end{enumerate}
Assumption 2 implies that propagation stops at time $T = \arg\min_{t \ge 0} \{t:  D_{t + 1} = \emptyset\}$, that is, the first time when no new nodes were activated.  The entire diffusion process can then be encoded by the set sequence $\mathcal{D}:= (D_0, \ldots, D_T)$, which we refer to as the {\it information diffusion path} or the {\it propagation trace}. We write $A(\mathcal{D})$  to denote the set of all nodes activated along the entire path.  

 The assumptions imply that not every sequence of node subsets can be a feasible propagation trace.  The following definition formalizes what makes a propagation trace feasible.    We denote the set of all feasible traces on $G$ by $\mathcal{F}(G)$.
 \begin{definition}[Feasible trace]\label{feas_trace_def}
 We say that a set sequence $\mathcal{D} = (D_0, \ldots, D_T)$ %with $D_t \subseteq V$ 
 is a feasible propagation trace if
 \begin{enumerate}
     \item $D_0 \neq \emptyset$; 
     \item All sets $D_t, \ t=0,\ldots, T$ are disjoint; 
     \item For every $t=1,\ldots, T$,  if $v \in D_t$ then $v \in C(D_{t-1}) \setminus A_{t-1}$, i.e., each node newly activated at time $t$ has at least one parent that was newly activated at time $t-1$.   
 \end{enumerate}
\end{definition}

Now, we are ready to state the formal definition of the diffusion model on the graph, which can be thought of as a collection of rules governing the transmission of information from parents to children.   These rules are not necessarily  Markovian and can depend on the entire previous propagation history.  

\begin{definition}[Diffusion model]\label{difmodel_def} 
A diffusion model on $G = (V, E)$ with a possibly graph-dependent parameter space $\Theta$ is a collection $\mathcal{M}_G (\Theta) = \{M_{G, \boldsymbol{\theta}}, \boldsymbol{\theta} \in \Theta\}$ where $M_{G, \boldsymbol{\theta}}$ is a mapping from a feasible trace $\mathcal{D}_t = (D_0, \ldots, D_{t-1}) \in \mathcal{F}(G)$ to a  probability distribution on feasible sets of newly activated nodes, 
\begin{equation}\label{diff_model_def_equation}
   % \mathcal{D}_t=(D_0, \ldots, D_{t-1}) \in \mathcal{F}(G) \quad \stackrel{M_{G, \boldsymbol{\theta}}}{\longmapsto} \quad 
  \mathbb{P}_{\boldsymbol{\theta}}^{t}(D_{t} = S \ | \mathcal{D}_t), \ S \subset C(D_{t-1}) \setminus A_{t-1}.
\end{equation} 
We omit the subscript $G$ in $M_{G, \boldsymbol{\theta}}$ and $\mathcal{M}_{G} (\Theta)$ whenever it is clear from the context.  
\end{definition}

%\begin{remark}\label{independent_activ_remark}
   
    The most general form of Definition \ref{difmodel_def} requires specifying the activation probability for each possible subset of $C(D_{t-1}) \setminus A_{t-1}$, which is impractical to work with.  A standard simplifying assumption is that node activations are independent conditionally on the preceding propagation history, implying that the right-hand side in \eqref{diff_model_def_equation} can be decomposed as follows:
     $$\mathbb{P}_{\boldsymbol{\theta}}^{t}(D_{t} = S \ | \ \mathcal{D}) = \prod_{v\in S} \mathbb{P}_{\boldsymbol{\theta}}^{t}(v \in D_t \ | \ \mathcal{D})\prod_{v\in C(D_{t-1}) \setminus (A_{t-1} \sqcup S)}\Bigl(1 -\mathbb{P}_{\boldsymbol{\theta}}^{t}(v\in D_t \ | \ \mathcal{D})\Bigr).
    $$
    With this additional assumption, defining a diffusion model is equivalent to defining the activation probability of a node given an arbitrary feasible diffusion history.
%\end{remark}
\begin{remark}\label{dif_model_on_arbitrary_graph_remark}
Sometimes, it may be convenient to refer to a class of  diffusion models  without a reference to a specific graph $G$.  % as we did in Definition \ref{difmodel_def}, and rather consider it for each graph in the set $\mathcal{G}$ of all possible simple directed graphs simultaneously. 
We  define a {\it diffusion model class} as the collection 
$\mathcal{M} = \{\mathcal{M}_G(\Theta_G): G \in \mathcal{G}\}$, where $\mathcal{G}$ is the set of simple directed graphs.  
%where we assume that there is some natural correspondence between the graphs $G\in\mathcal{G}$ and the associated parameter spaces $\Theta_G$. This correspondence exists for all popular diffusion models as those defined in Examples \ref{ltm_def}, \ref{icm_def}, and \ref{trigger_def}. For example, for the IC model, it will be $G = (V, E) \mapsto \Theta_G = [0, 1]^{|E|}$.
\end{remark}

Next, we introduce the three arguably most popular diffusion model classes --  the Linear Threshold (LT), the Independent Cascade (IC), and the Triggering model -- and show how they fit into our general framework.  We note that by construction, node activation events are independent conditionally on propagation history for all of these model classes.   
%Conceptually, under the LT model, all of a node's activated parents influence it additively; the IC model gives each active parent node one individual chance to influence its children, and under the Triggering model, each node initially chooses a ``triggering'' subset of its parents according to some node-dependent distribution and activates as soon as at least of the parents in this set is activated. It is easy to check by definition that the node-activation events are independent in all these models. So, per Remark \ref{independent_activ_remark}, these models are fully defined by the activation probability of a node conditional on the preceding propagation history. We state these probabilities at the end of each example.

\begin{example}[Linear Threshold (LT) Model]\label{ltm_def}
Assume each edge in the graph $G=(V, E)$ is assigned a weight parameter $b_{u,v}$, and we arrange the weights into a vector $\boldsymbol{\theta}$.  The parameter space for the LT model is given by 
\begin{equation}\label{param_space_ltm}
\Theta_{LT} = \Bigl\{\boldsymbol{\theta} \in \mathbb{R}^{|E|} \text{ s.t. for all } v\in V,\ \sum_{u\in P(v)} b_{u,v} \le 1 \ \text{ and } \ b_{u,v} \ge 0 \ \text{ for all } \ (u, v) \in E\Bigr\}.
\end{equation}
Each node $v\in V$ gets a random activation threshold $U_v$ sampled i.i.d.\ from the $\operatorname{Unif}[0,1]$ distribution at the outset.  At every time step $t\ge 1$, a non-active node $v$ becomes activated if the sum of the  edge weights from all its previously activated parents exceeds its threshold $U_v$, that is, %$v \in D_t$ if 
%$\sum_{u \in P(v) \cap A_{t-1}} b_{u, v} \geq U_v.$ 
%So, the transition probability can be shown to equal
\begin{align*}
\mathbb{P}^t_{\boldsymbol{\theta}}(v\in D_t\ |\ D_0, \ldots, D_{t-1})  
% &= \mathbb{P}\Bigl(U_v \le \sum_{u\in A_{t-1} \cap P(v)} b_{u,v} \ | \ U_v > \sum_{u\in A_{t-2} \cap P(v)} b_{u, v}\Bigr) \\
    &={\sum_{u \in D_{t-1} \cap P(v)} b_{u, v} \over 1- \sum_{u \in A_{t-2} \cap P(v)} b_{u, v}}.  
\end{align*}
This expression gives the probability of exceeding the threshold at time $t$ given that it was not yet exceeded at time $t-1$.  
\end{example}

\begin{example}[Independent Cascade (IC) Model]\label{icm_def}
Assume that each edge $(u, v)\in E$ is associated with a propagation probability $p_{u, v} \in [0, 1]$, arranged into a  vector $\boldsymbol{\theta} \in \Theta_{IC} = [0, 1]^{|E|}$.
At every time step $t\ge 1$, each newly active node $u \in D_{t-1}$ independently tries to activate all of its not yet active children $v \in C(u) \setminus A_{t-1}$  with probability $p_{u, v}$, that is,   %That is, $v\in D_{t}$ if at least one $Z_{u,v} \sim \operatorname{Bernoulli}(p_{u,v})$, for $u\in P(v) \cap D_{t-1}$, takes on the value $1$. %Then, the transition probability equals:
$$\mathbb{P}^t_{\boldsymbol{\theta}}(v\in D_t\ |\ D_0, \ldots, D_{t-1})  = 1 - \prod_{D_{t-1} \cap P(v)} (1 - p_{u, v}).
$$
\end{example}

 \begin{example}[Triggering Model] \label{trigger_def}
     At the outset, each node $v$ independently chooses a random triggering set $\Gamma_v$ according to some distribution over subsets of its parents, with probability of a node  $v$ choosing $S \subset P(v)$ as its triggering set denoted by $p_{v, S}$. With $\mathcal{P}^m$ denoting the set of all discrete distributions on $m$ points, the parameter space is then 
     $\Theta_{TR} = \{\mathcal{P}^{|P(v)|}, \ v\in V\}.
     $
     An inactive node $v$ becomes active at time $t \ge 1$ if its triggering set $\Gamma_v$ contains a node in $D_{t-1}$, with probability of activation given by 
$$\mathbb{P}^t_{\boldsymbol{\theta}}(v\in D_t\ |\ D_0, \ldots, D_{t-1})  = {\sum_{S\subset P(v)\cap D_{t-1}} p_{v, S}\over \sum_{S\subset P(v)\setminus A_{t-2}} p_{v, S}}.
$$
 \end{example}

As shown by \citep{Kempe}, the IC and LT model classes are special cases of the Triggering model class. Conveniently, within our parametric framework, we can formalize this in a very general form.  
\begin{definition}[Diffusion model subclass]\label{diff_model_subclass_def}
    Consider two diffusion model classes $\mathcal{M} = \{\mathcal{M}_G(\Theta_G): G\in \mathcal{G}\}$ and $\tilde{\mathcal{M}} =\{\tilde{\mathcal{M}}_G(\tilde{\Theta}_G): G\in \mathcal{G}\}$. We say that $\mathcal{M}$ is a {\it subclass} of $\tilde{\mathcal{M}}$ if for any $G\in\mathcal{G}$ and any instance $M_{G, \boldsymbol{\theta}}\in \mathcal{M}_{G}(\Theta_G)$, there exists an instance $\tilde{M}_{G, \boldsymbol{\theta}} \in \tilde{\mathcal{M}}_G(\tilde{\Theta}_G)$, such that they  coincide on each feasible trace $\mathcal{F}(G)$. If $\tilde{\mathcal{M}}$ and $\mathcal{M}$ are subclasses of each other, we say that the two model classes are {\it equivalent}.
\end{definition}
 %Unfortunately, each instance of the parametric family $\mathcal{M}_G(\Theta)$ is not a data-generating distribution but a collection of transition kernels. This fact does not allow us to immediately apply the classical statistical learning framework that aims to choose the ``best'' distribution $\mathbb{P}_{\hat{\boldsymbol{\theta}}}$ from a parametric family $\{\mathbb{P}_{\boldsymbol{\theta}}, \boldsymbol{\theta}\in \Theta\}$, for example, by finding an MLE $\hat{\boldsymbol{\theta}}$ using samples from the true data distribution $\mathbb{P}_{\boldsymbol{\theta}^*}$.

Note that Definition \ref{difmodel_def} does not specify the distribution that generates the seed set $D_0$; to describe the full data-generating distribution, we need a distribution $\mathbb{P}^0$ over the subsets of $V$ from which the seed set is generated.  We will assume that  $\mathbb{P}^0$ does not depend on $\boldsymbol{\theta}$, and $\mathbb{P}^0(\emptyset) = 0$.  We refer to the pair $(\mathcal{M}_{G}(\Theta), \mathbb{P}^0)$ as a {\em seeded diffusion model} with the seed distribution $\mathbb{P}^0$.
The seeded diffusion model corresponds to a distribution on all feasible propagation traces:
\begin{equation}\label{trace_distrib_def_eq}
\mathbb{P}_{\boldsymbol{\theta}}(\mathcal{D}):=\mathbb{P}^{0}(D_0)\prod_{t=1}^{T}\mathbb{P}_{\boldsymbol{\theta}}^t(D_{t}|D_0, \ldots, D_{t-1})\mathbb{P}_{\boldsymbol{\theta}}^{T+1}(\emptyset| \mathcal{D}), \qquad \mathcal{D} \in \mathcal{F}(G).
\end{equation}
In turn, the trace distribution also uniquely defines the seeded diffusion model.

% Indeed, for any feasible trace  $(D_0, \ldots, D_{T-1})$ and $S \subset  C(D_{T-1}) \setminus A_{T-1}$,  we can write %for each $\boldsymbol{\theta}\in \Theta$:
% \begin{equation}\label{equiv_dif_model_and_trace_distrib}
%     M_{G, \boldsymbol{\theta}}(D_0, \ldots, D_{T-1})(S) = {\sum_{\mathcal{D}' \in \mathcal{F}(G)} \mathbb{P}_{\boldsymbol{\theta}} (\mathcal{D}') \mathbf{1}(D_T'=S; D_t' = D_t, \ \forall t = 0, \ldots, {T-1})\over \sum_{\mathcal{D}' \in \mathcal{F}(G)} \mathbb{P}_{\boldsymbol{\theta}} (\mathcal{D}') \mathbf{1}(D_t' = D_t, \ \forall t = 0, \ldots, T-1)},
% \end{equation}
% where $\mathbf{1}(\cdot)$ is the indicator function. Therefore, there is one-to-one correspondence between the seeded diffusion model and a trace-generating distribution.  \liza{To be honest, I struggle to understand how they are different from each other. I do understand that they are constructed to be mathematically different objects, but why do we need both?  It seems to me that (4) is unnecessary, and it does not appear to ever be referenced again.  }
Assuming one observes a collection of traces from a seeded diffusion model $(\mathcal{M}_{G}(\Theta), \mathbb{P}^0)$,  the established equivalence with the trace-generating distribution allows directly using the standard statistical parameter estimation techniques, such as the maximum likelihood estimation (MLE).
Before proceeding to estimation, a natural question statistician would address is when the family $\{\mathbb{P}_{\boldsymbol{\theta}}: \boldsymbol{\theta}\in\Theta\}$ of trace distributions is identifiable, that is,  $\mathbb{P}_{\boldsymbol{\theta}_1} = \mathbb{P}_{\boldsymbol{\theta}_2}$ implies   $\boldsymbol{\theta}_1 = \boldsymbol{\theta}_2$. While specific conditions will depend on the form of the distribution, it is clear that all nodes in the graph must be reachable with a positive probability; otherwise, parameters associated with these nodes, such as the incoming edge weights, have no influence on the trace distribution and thus cannot be identified.  The necessary and sufficient identifiability conditions for our proposed class of models will be stated formally in Theorem \ref{identif_theorem}.

% \noindent We further propose the General Linear Threshold model, another diffusion model of which the IC and LT models are subclasses.

% \begin{remark}
%    Notice that for the LT model, the vector of parameters $\boldsymbol{\theta}$ from definition \ref{difmodel_def} corresponds to $\{b_{u,v}: (u,v)\in E\}$ and 
%    $$M_{LT}(D_0, \ldots, D_{t-1})(S) = \prod_{v\in S} B_{\boldsymbol{\theta}_v}(D_{t-1}) \prod_{v\in C(A_{t-1}) \setminus S} (1 - B_{\boldsymbol{\theta}_v}(A_{t-1})),
%    $$
%    where we denoted 

%    On the other hand, for the IC model $\boldsymbol{\theta} = \{p_{u,v}: (u,v)\in E\}$ and 
%    $$M_{IC}(D_0, \ldots, D_{t-1})(S) = \prod_{u\in D_{t-1}}\left[\prod_{v\in S} p_{u,v} \prod_{v\in C(A_{t-1}) \setminus S}  (1 - p_{u,v})\right].
%    $$
% \end{remark}

\section{The General Linear Threshold model}\label{gltm_introduce_section}

In this section, we introduce the General Linear Threshold (GLT) model and establish its relationship to the IC, LT, and Triggering models.  One may ask why there is a need for a new model, when there are already several, and in particular the triggering model seems quite general, encompassing both IC and LT models.  The problem with the triggering model is its number of unknown parameters, $\sum_{v\in V} \left( 2^{|P(v)|} - 1\right)$, which is not feasible to fit in most cases. On the other extreme, the LT and IC models have only $|E|$ parameters, but are often  insufficiently flexible in practice. In particular, both these models  assume that all nodes behave identically when receiving equal amounts of influence from their neighbors, which does not allow for any node heterogeneity. In many contexts, some individuals will need a lot more influence than others to become ``activated", and this can be estimated from data on propagation traces. The proposed GLT model allows us to account for differences in users' willingness to accept new information while not significantly increasing the number of unknown parameters.  

\begin{definition}[General Linear Threshold (GLT) model]\label{gltm_def}
Assume that each node $v\in V$ has a random threshold $U_v$ drawn independently from a distribution supported on $[0, h_v]$ with $h_v \le \infty$, with cumulative distribution function (cdf) $F_v$. Further assume that edges have weights $b_{u,v} \ge 0$, such that the in-degree of each node $v\in V$ satisfies
$\sum_{u \in P(v)} b_{u, v} \leq h_v$.
% :=\inf\{t: F_v(t) = 1\}.
As in the LT model, a node $v$ activates at time $t\ge 1$ if $\sum_{u \in P(v) \cap A_{t-1}} b_{u,v} \ge U_v$.
\end{definition} 

If we treat all threshold distributions as fixed, the parameter space of the GLT model is the set of all possible edge weights satisfying the model constraints:
\begin{equation}\label{glt_param_space}
 \Theta = \left\{\boldsymbol{\theta} \in \mathbb{R}^{|E|}  \text{ s.t.\ for all } v\in V, \ \|\boldsymbol{\theta}_v\|_1 \le h_v \text{ and } \boldsymbol{\theta}_v \ge 0\right\},
\end{equation}
where the sub-vector of $\boldsymbol{\theta}$ corresponding to the parent edges of $v$ is denoted by $\boldsymbol{\theta}_v = \{b_{u, v}: u\in P(v)\}$. %, the $\ell_1$-norm of a vector by $\|\cdot\|_1$, and the inequality between a number and a vector means element-wise comparison.
If we also model the threshold distributions $F_v, v\in V$ in some way, there may be additional parameters.   We will discuss this option in Section \ref{est_threshold_params};  for now, we will assume they are fixed.
%from some parametric or non-parametric family $\mathcal{P}_v$, that is, define the threshold distribution space as $\Phi = \{F_v \in \mathcal{P}_v,\  v \in V\}$ and the corresponding full parameter space as the Cartesian product of $\Theta$ in \eqref{glt_param_space} and $\Phi$: \begin{equation}\label{glt_param_space_with_thresholds} \Theta^{(full)} = \Theta \mathop{\otimes} \Phi. \end{equation}
%Since we only mention the full parameter space in the context of model fitting (), we assume that, unless otherwise stated

Under the GLT model, the activation probability of a node $v$ given a feasible history $(D_0, \ldots, D_{t-1})$ equals the probability that the sum of weights from $A_{t-1}$ is higher than $U_v$ conditional on the event that it had not exceeded $U_v$ at time $t-2$: 
\begin{equation}\label{trans_prob_glt}
\mathbb{P}^t_{\boldsymbol{\theta}}(v\in D_t\ |\ D_0, \ldots, D_{t-1}) = {F_v\Bigl(B_v(A_{t-1};{\boldsymbol{\theta}_v})\Bigr) - F_v\Bigl(B_v(A_{t-2};{\boldsymbol{\theta}_v})\Bigr) \over 1 - F_v\Bigl(B_v(A_{t-2};{\boldsymbol{\theta}_v})\Bigr)},
\end{equation}
where 
\begin{equation}\label{weight_indeg_from_set}
    B_v(S; \boldsymbol{\theta}_v) = \sum_{S \cap P(v)} b_{u,v}
\end{equation}
is the influence node $v$ with parent edge weights $\boldsymbol{\theta}_v$ receives from a node set $S$.   
Plugging the cdf $F_v(x) = x$ of the uniform distribution into \eqref{trans_prob_glt} gives the activation probability under the LT model stated in Example \ref{ltm_def}, confirming the LT model class is equivalent to the GLT model class with uniformly distributed thresholds.

The following proposition establishes the less obvious relationship between the IC and the GLT models. 
\begin{proposition}\label{ic_is_glt_propos}
   The class of IC models is equivalent to the class of GLT models with all node thresholds distributed as $U_v\sim \operatorname{Exponential(1)}$, that is, with $F_v(x) = 1-e^{-x}$.
\end{proposition}
\begin{proof}
  %  Consider an arbitrary graph $G = (V, E)$. As exponential distribution has unbounded support, the parameter space of the GLT model on $G$ obtains the form $[0, +\infty]^{|E|}$, while for the IC model it is $[0, 1]^{|E|}$. 
  By Definition \ref{diff_model_subclass_def}, it is enough to construct a bijective mapping between each IC edge probability $p_{u, v} \in [0, 1],$ and the GLT edge weight $b_{u, v} \in [0, \infty]$ for each $(u, v)\in E$.  %, so that the corresponding transition kernels of the two models coincide. 
   Let $b_{u, v}:=  -\log(1-p_{u,v})$.  Then, if we model node threshold distributions as exponential with parameter 1, by the memoryless property of the exponential  distribution the activation probability of a node $v\in V$ given a feasible history $(D_0, \ldots, D_{t-1})\in\mathcal{F}(G)$ can be written under the GLT model as
    \begin{align*}
        \mathbb{P}^t_{\boldsymbol{\theta}}(v\in D_t |\ D_0, \ldots, D_{t-1}) &= \mathbb{P}\Bigl(U_v \le B_v(A_{t-1}; \boldsymbol{\theta}_v) \mid \ U_v > B_v(A_{t-2}; \boldsymbol{\theta}_v)\Bigr) \\
        &=\mathbb{P}\Bigl(U_v \le B_v(D_{t-1}; \boldsymbol{\theta}_v)\Bigr) = 1 - \prod_{u\in D_{t-1}\cap P(v)} (1-p_{u, v}),
    \end{align*}
    which coincides with the transition kernel of the IC model.
\end{proof}
%With Proposition \ref{ic_is_glt_propos}, we can conclude that the 

We have now established that both IC and LT models are subclasses of the GLT model with identically distributed node thresholds. As we argued, a more 
 interesting scenario is allowing this distribution to vary from node to node. A natural approach to make the number of GLT parameters manageable in this case is to choose threshold distributions from a parametric family. For example, if $U_v \sim \operatorname{Beta}(\alpha_v, \beta_v)$, the model has only $|E| + 2|V|$ parameters. If we further assume that the network can be partitioned into communities and nodes within one community follow the same distribution (see Figure \ref{fig:community_beta}), we can further reduce the number of parameters. The following proposition demonstrates that even if each node has an individual set of $r$ parameters but it is not too large compared to the average node in-degree, the GLT has negligibly more parameters than the IC and LT models and many fewer than the Triggering model.  % with the total of $\sum_{v\inV}\left(2^{|P(v)|} - 1\right)$ parameters. 
\begin{proposition}\label{propos_trig_more_params}
   Consider a directed graph $G = (V, E)$ with the average node in-degree $d$. %Then for any $r \ge 0$, it holds $|E| + r|V| = \Bigl(1 + {r\over d}\Bigr)|E|$. 
   Then for any $0 \le r \le 2^d - d - 1$, it holds that 
   $$|E| + r|V| \le \sum_{v\in V}\left(2^{|P(v)|} - 1\right).$$
   In particular, the inequality is strict with $r=1$ and $d > 2$ and with $r=2$ and $d>2.45$.
\end{proposition}
\begin{proof}
   % The first relationship follows from $|E| = d|V|$ and the second requires additional application of 
   Noting that  $|E| = d|V|$ and applying Jensen's inequality to the convex function $f(x) = 2^x$, we have 
    $${|E| + r|V| \over \sum_{v\in V}\left(2^{|P(v)|} - 1\right)} \le {(d + r)|V| \over  (2^{d} -1)|V|}  %= {d + r \over 2^d - 1} 
    \le 1.
    $$
\end{proof}
  \begin{figure}[t!]
     \centering
     \includegraphics[width=9.cm, height=6cm]{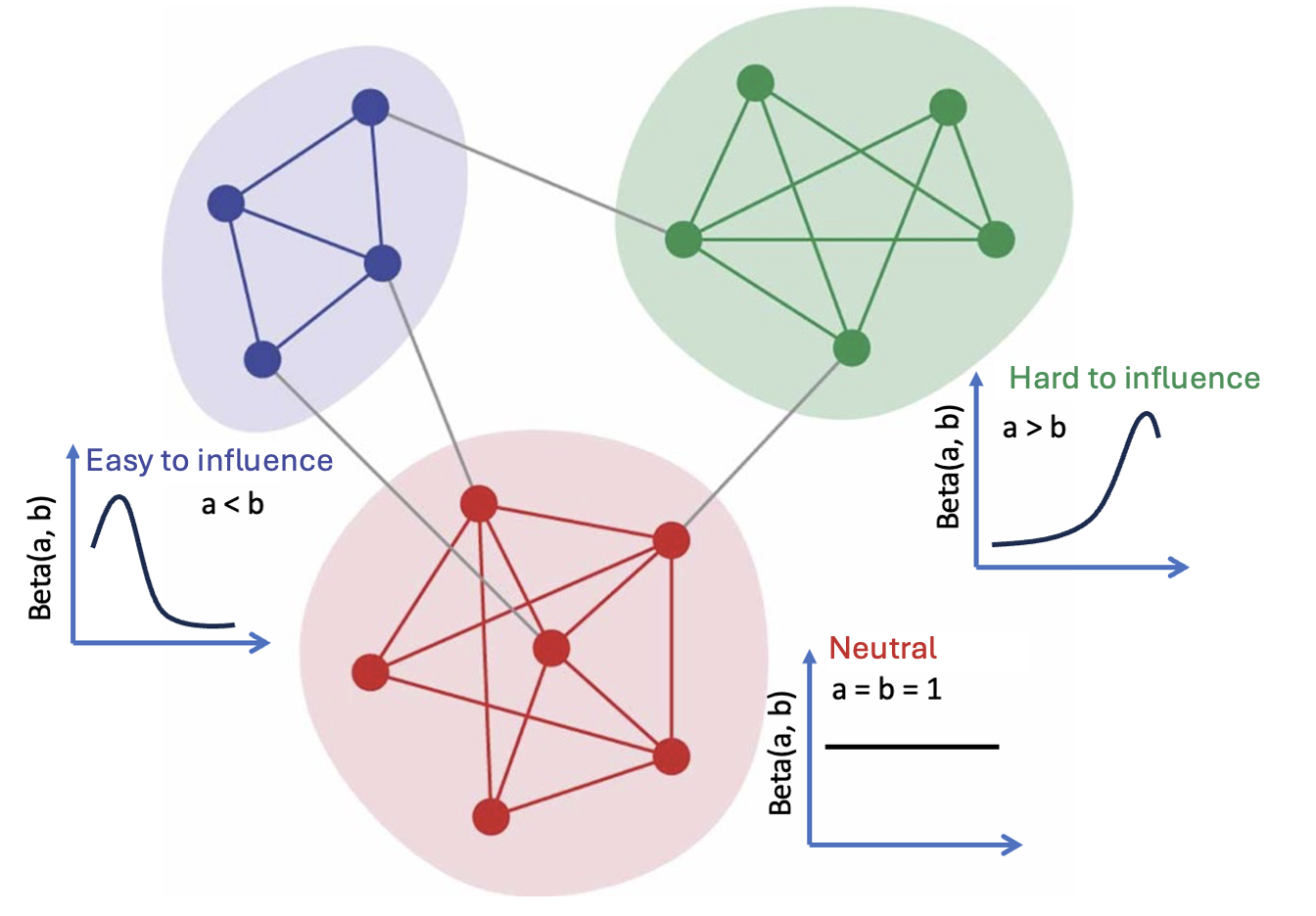}
     \caption{An example network with three communities with different levels of receptiveness to new information, all modeled with the Beta distribution.} 
           \label{fig:community_beta}
 \end{figure}
 
The next proposition shows that, despite a much smaller number of parameters, the GLT model is not a subclass of the Triggering model. Unsurprisingly, the Triggering model is also not a subclass of the GLT. The proof can be found in Section \ref{gltm_introduce_section_proofs} of the Appendix. 
  
\begin{proposition}\label{glt_vs_trig}
 The GLT model is not a subclass of the Triggering model and vice versa. 
 \end{proposition}

  Figure \ref{fig:model_relation} summarizes the relationships between all the diffusion models discussed. % As in Definition \ref{diff_model_subclass_def}, we represent one diffusion model as contained in another model if for each instance of the ``smaller'' model, there is an instance of the ``larger'' one such that their induced trace distributions coincide. 
  We can think of the GLT and the Triggering model as two alternative flexible generalizations, both encompassing the popular IC and LT models, with GLT being preferable for statistical estimation and inference due to its better balance between flexibility and the number of parameters.
  
  \begin{figure}[tbh!]
    \centering
    \includegraphics[scale=0.4]{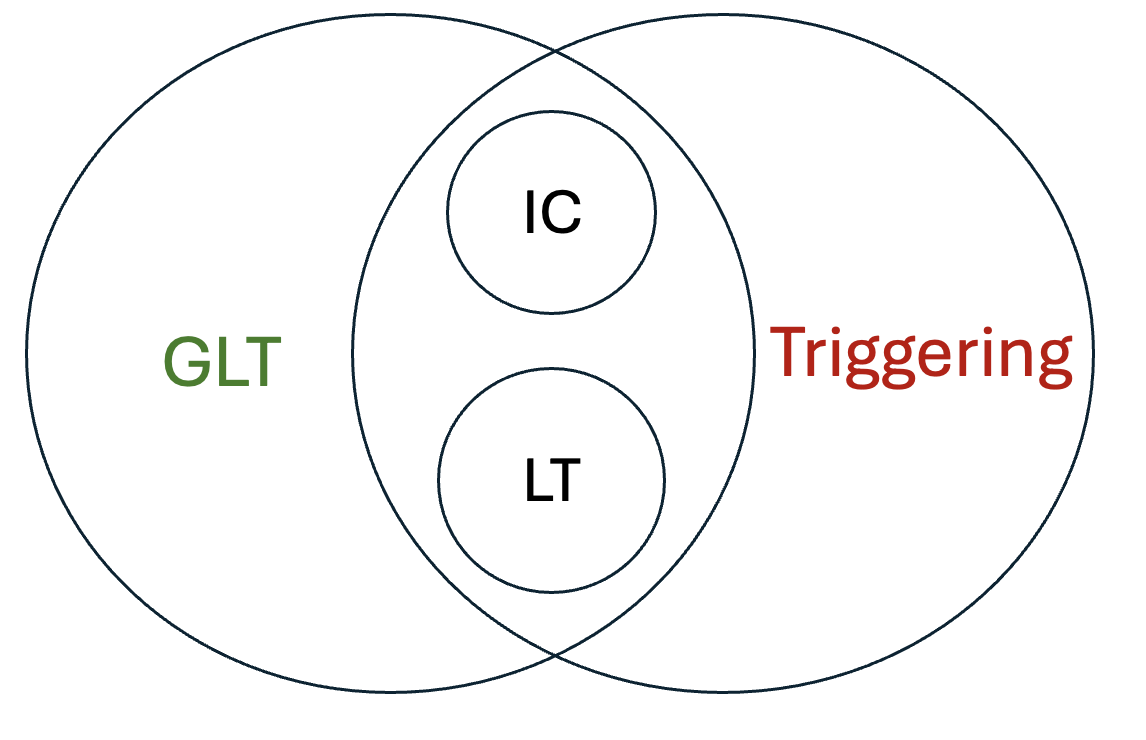}
    \caption{Relationship between different diffusion models.}
    \label{fig:model_relation}
\end{figure}

\section{Estimation and theoretical properties under the GLT model}
    \label{ch:methods2}
	
In this section, we study the GLT model on a given fixed network $G = (V, E)$. We present the necessary and sufficient conditions for identifiability of the weights, derive a constrained maximum likelihood estimator (MLE) of the weights from fully observed propagation traces and a finite sample bound for its error, and show that the estimator is asymptotically normal. We then extend the estimation procedure to the case of partially observed traces and unknown node threshold distributions. 

\subsection{Identifiability for the GLT model}\label{identif_section}
We begin by studying identifiability of the GLT edge weights with respect to the trace distribution family $\{\mathbb{P}_{\boldsymbol{\theta}}, \boldsymbol{\theta}\in \Theta\}$ induced by the GLT model. Plugging the GLT transition probability from \eqref{trans_prob_glt} into the general trace likelihood in \eqref{trace_distrib_def_eq}, we can conveniently express the trace distribution as follows:
\begin{align}\label{individ_trace_prob}
\mathbb{P}_{\boldsymbol{\theta}}(\mathcal{D})& = \mathbb{P}^0(D_0) \prod_{v \in C(A_{T})}  \left\{1 - F_v\left[B_{v}\left(A_T; \boldsymbol{\theta}_v\right)\right]\right\}
%\times
\prod_{t=0}^{T - 1} \prod_{v \in D_{t+1}}  \left\{F_v\left[B_{v}\left( A_t; \boldsymbol{\theta}_v\right)\right] - F_v\left[B_{v}\left( A_{t-1}; \boldsymbol{\theta}_v\right)\right]\right\}.
\end{align}
\noindent Here, the first term does not depend on $\boldsymbol{\theta}$ by definition of the seeded diffusion model. The second term represents nodes that were not activated but have at least one active parent in the trace. 
The third term captures activated nodes, i.e., nodes in $A_{T} \setminus D_{0}$.

For reasons that we will elaborate on later, it is hard to establish identifiability unless the probability of any feasible trace with $\mathbb{P}^0(D_0) > 0$ in \eqref{individ_trace_prob} is bounded away from zero. To guarantee that, we first need to assume that the threshold cdf $F_v$ is strictly monotone for the nodes $v$ that may appear in the trace likelihood, that is, the (child) nodes having at least one parent:
\begin{equation}\label{sink_nodes}
    V_c = \{v \in V: P(v) \ne \emptyset \}.
\end{equation}
\begin{assumption}[Invertible cdf]\label{assump1} The threshold cdf $F_v$ of every node $v\in V_c$ is strictly monotone (and thus invertible) on its support $[0, h_v]$.
\end{assumption}
We also need to ensure that edge weights are strictly positive and that the parent weights of every node $v\in V_c$ sum to less than $h_v$. Thus we  truncate the parameter space as follows: 
\begin{equation}\label{trunc_glt_param_space}
 \tilde{\Theta} = \left\{\boldsymbol{\theta} \in \mathbb{R}^{|E|}  \text{ s.t.\ for all } v\in V, \ \|\boldsymbol{\theta}_v\|_1 \le \gamma \text{ and } \boldsymbol{\theta}_v \ge \varepsilon\right\},
\end{equation}
where $\varepsilon > 0$ and $\max_{v\in V_c}|P(v)|\varepsilon <\gamma < \min_{v\in V_c}h_v$ are some universal constants. Together with Assumption \ref{assump1},  $\boldsymbol{\theta}_v \ge \varepsilon$ will ensure that any node has a positive chance to activate its child even if none of the other parents are activated, and $ \|\boldsymbol{\theta}_v\|_1 \le \gamma$ will ensure that even if all parents of a node are activated, there is a positive probability that the node will not be activated.  

 The following lemma states these restrictions on the parameter space and Assumption \ref{assump1} are sufficient for the likelihood in \eqref{individ_trace_prob} to be positive for every feasible trace starting with a seed set from the support of $\mathbb{P}^0$.  
\begin{lemma} \label{pos_trace_prob_lemma}
    Under Assumption \ref{assump1}, 
    any feasible trace $\mathcal{D} \in\mathcal{F}(G)$ with $\mathbb{P}^0(D_0) > 0$ satisfies $\mathbb{P}_{\boldsymbol{\theta}}(\mathcal D)> 0$ for each $\boldsymbol{\theta}\in \Tilde{\Theta}$.
\end{lemma}
\noindent The proof of this lemma and all other results in this section can be found in Section \ref{identif_section_proofs} of the Appendix.

Lemma \ref{pos_trace_prob_lemma} does not guarantee that any node can appear in a feasible trace with positive probability. Indeed, if the seed sets from the support of $\mathbb{P}^0$ are not sufficiently rich, some nodes may be not reachable by any trace, making their parent weights non-identifiable. We say a node $u \in V$ is {\it reachable} if either $\mathbb{P}^0(u\in D_0) > 0$ or there is a directed path to $u$ from at least one $v\in V$ with $\mathbb{P}^0(v\in D_0) > 0$. The following proposition formalizes this intuition as a necessary identifiability condition.
\begin{proposition}\label{all_reachable_nodes_propos}  If $\{\mathbb{P}_{\boldsymbol{\theta}}, \boldsymbol{\theta} \in \Tilde{\Theta}\}$ is identifiable, all non-isolated nodes in $V$ are reachable.  
\end{proposition}

 \begin{figure}[t!]
\centering  \includegraphics[scale=0.3]{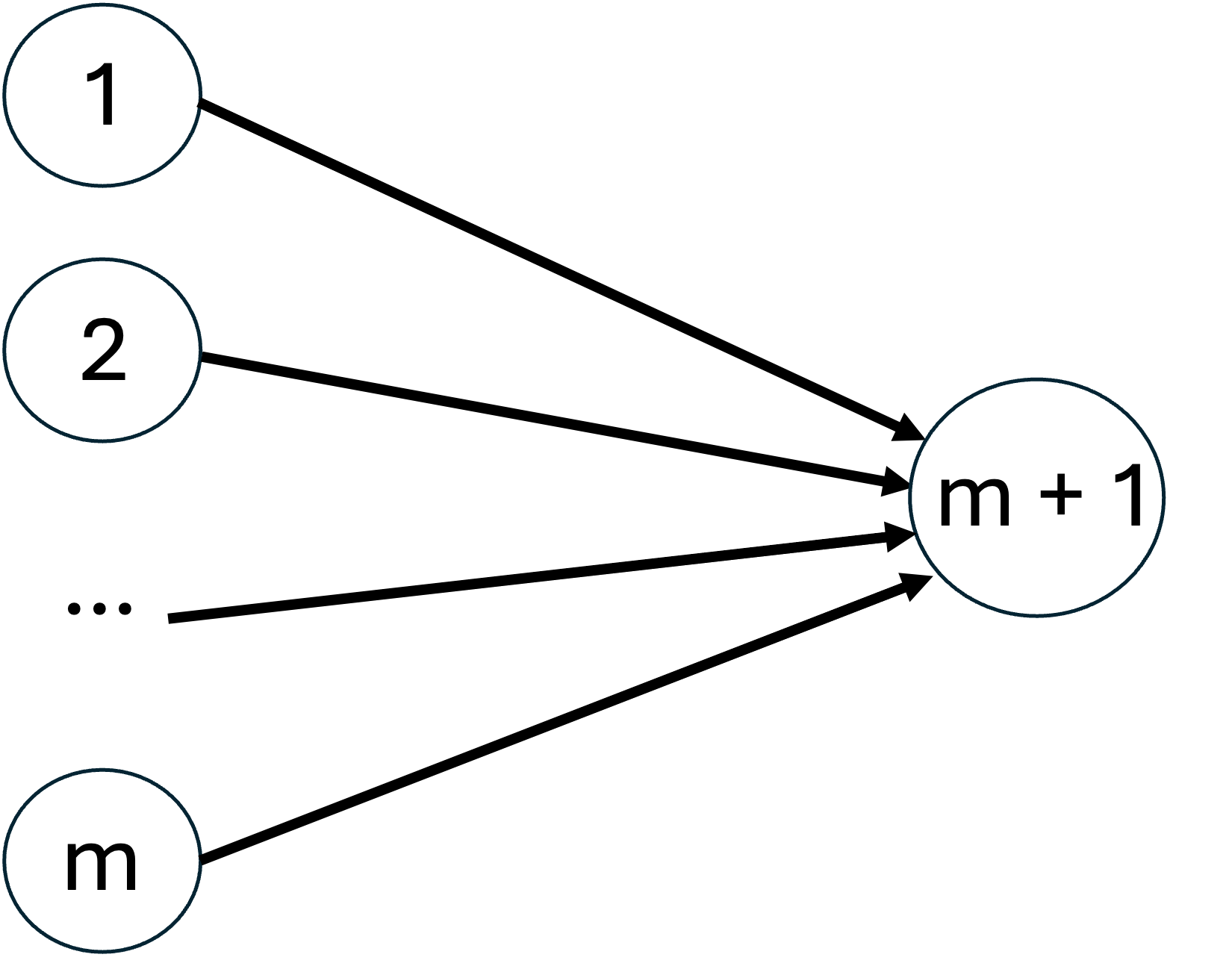}  \caption{A star graph of in-degree $m$:  $V=\{1, \ldots, m+1\}$,   $E=\{(1, m+1) \ldots, (m, m+1)\}$. }
   \label{fig:instar_graph}
\end{figure}

The following example demonstrates that reachability by itself is not sufficient for identifiability. 
\begin{example} Consider the star graph in Figure \ref{fig:instar_graph} with $m=2$.  Its edge set $E = \{(1,3), (2,3) \}$ has weights $b_{13}$ and $b_{23}$, respectively.  
 Fix the seed set $D_0 =\{1,2\}$, so that $\mathbb{P}^0(\{1, 2\}) = 1$.   Any seeded GLT model with this $\mathbb{P}^0$ induces a trace distribution that is a function of $b_{13} + b_{23}$:
$$\mathbb{P}(\mathcal{D}) =\begin{cases}
    F_3(b_{13} + b_{23}), & \text{ node } 3 \in D_1, \\
    1 - F_3(b_{13} + b_{23}), & \text{ node 3} \notin D_1, 
\end{cases} $$ 
and therefore $b_{13}$ and $b_{23}$ are not individually identifiable, only their sum is.   It is easy to verify that if the support of $\mathbb{P}^0$ includes at least two distinct subsets of $\{1, 2\}$, the weights are identifiable. More formally, there needs to exist $S_j \subseteq \{1, 2\}, j=1, 2$ with $\mathbb{P}^0(S_j) > 0$ such that the $2 \times 2$ matrix $X = [\mathbf{1}(i \in S_j)]_{i, j=1}^2$ has full rank.
\end{example}
The condition on $\mathbb{P}^0$ can be directly extended to a star graph of arbitrary in-degree $m$,  requiring the existence of $S_1, \ldots, S_m \subseteq \{1, \ldots, m\}$ such that an analogous $m\times m$ matrix $X$ is of full rank. For graphs with more than one child node, it turns out that it is necessary and sufficient to require a similar condition for each child node $v$, with the only difference that the corresponding matrix $X_v$ is now constructed using parent subsets $S_j$ that can appear within any active set $D_t$ preceding activation of $v$, not just the seed $D_0$. We formally state this condition in the following theorem.  % Apparently, it would be unrealistic to assume that all possible subsets of nodes can become seed sets, so the support of $\mathbb{P}^0$ is not restricted to $2^{V}$. In particular, this implies that not all of the possible traces will have a positive probability of occurrence and thus some subset of edges may never appear in the likelihood. Therefore, when speaking about consistency, we will only refer to the edges through which the activation attempt can be observed with positive probability, i.e. the subset 
% $$\Tilde{E} = \{(u,v) \in E: \mathbb{P}_{\mathcal{U}, \mathbb{P}^0}(u\in A_T) > 0\}.$$

\begin{theorem}\label{identif_theorem}
    % Consider the restriction $G'=(V', E')$ of graph $G$ as defined in Eq. \ref{graph_restriction}. 
    Under Assumption \ref{assump1}, $\{\mathbb{P}_{\boldsymbol{\theta}},\boldsymbol{\theta}\in \Tilde{\Theta}\}$ is identifiable if and only if for each child node $v\in V_c$ with  $P(v) = \{u_1, \ldots, u_m\}$, there exist $S_1,\ldots, S_{m} \subseteq P(v)$ such that 
    % 
    % $S_1,\ldots, S_m \subseteq P(v)$ such that
    \begin{enumerate}
        \item 
        For each $j=1,\ldots, m$, there is a feasible trace $(D_0^{(j)}, \ldots, D_{t_j}^{(j)}) \in \mathcal{F}(G)$ with $\mathbb{P}^0(D_0^{(j)}) > 0$, $v\notin A_{t_j}^{(j)}$, and $D_{t_j}^{(j)} \cap P(v) = S_j$.
        \item The matrix $\tilde{X}_v=[\mathbf{1}(u_i \in S_j)]_{i,j = 1}^m$ is invertible.
    \end{enumerate}  
\end{theorem}

 \begin{remark}
     The identifiability condition in Theorem \ref{identif_theorem} implies the necessary reachability condition from Proposition \ref{all_reachable_nodes_propos}: if a source node $u$ of an edge $(u,v)$ is unreachable, the matrix $\tilde{X}_v$ is not invertible, since it has a row of zeros $[\mathbf{1}(u\in S_j)]_{j=1}^m$ for any choice of $S_j$. 
 \end{remark}

To conclude this section, we would like to stress the novelty of the result in Theorem \ref{identif_theorem}.  Most previously proposed methods for estimating the parameters of a diffusion model \citep{learnability_influence,  learning_diffusion_in_cont_time, daneshmand2014estimatingdiffusionnetworkstructures} assume that any node can appear in the seed set with a positive probability, which is much more  restrictive, as in many applications not all nodes can be directly influenced.   To the best of our knowledge, this paper is the first to derive both necessary and sufficient identifiability conditions for the parameters of a diffusion model.

\subsection{Weight estimation under the GLT model}
\label{weight_estim_section}

Next, we derive a maximum likelihood estimator for the weights in the GLT model, given a  collection $\mathbb{D}$ of $N$ observed (and therefore feasible) propagation traces, 
\begin{equation}\label{trace_eq}
    \mathbb{D} = \{\mathcal{D}_n :=(D_1 ^{(n)}, \ldots, D_{T_n} ^{(n)}) \ | \ n=1, \ldots, N\},
\end{equation}
where $T_n$ is the number of time steps in trace $\mathcal{D}_n$. 
For now, we assume that all threshold distributions $F_v$ are known and postpone the discussion of estimating the threshold distribution to Section \ref{est_threshold_params}.  

We assume that the trace collection $\mathbb{D} = \{\mathcal{D}_n| \ n= 1,\ldots, N\}$ is i.i.d., by which we mean 
\begin{enumerate}
    \item[(a)] Seed sets $\{D_0^{(n)}| \ n=1, \ldots, N\}$ are generated independently from the seed distribution $\mathbb{P}^0$; 
    \item[(b)]  Node thresholds are generated independently for each trace and for each node, with 
    $$\mathcal{U}_n := \left(U_1^{(n)}, \ldots, U_{|V|}^{(n)}\right) \stackrel{\text{iid}} {\sim} (F_1,\ldots, F_{|V|}), \quad n=1,\ldots, N.$$
\end{enumerate}
For an i.i.d.\ trace collection, the parameters can be estimated by 
\begin{equation}\label{opt_prob}
   \hat{\boldsymbol{\theta}} =\arg\max_{\boldsymbol{\theta}\in \tilde{\Theta}} \sum_{n=1}^N L(\mathcal{D}_n|\boldsymbol{\theta}),   
\end{equation} 
where $L(\mathcal{D}_n|\boldsymbol{\theta})$, the log-likelihood of the trace $\mathcal{D}_{n}$, by \eqref{individ_trace_prob}, takes the form 
\begin{align}\label{individ_trace_lik}
L(\mathcal{D}_n| \boldsymbol{\theta})& =  \sum_{v \in C(A_{T_n}^{(n)})} \log \left\{1 - F_v\left[B_v\left(A_{T_n}^{(n)}; \boldsymbol{\theta}_v\right)\right]\right\}\\
&+\sum_{t=1}^{T_n} \sum_{v \in D_{t}^{(n)}} \log \left\{F_v\left[B_v\left(A_{t-1}^{(n)}; \boldsymbol{\theta}_v\right)\right] - F_v\left[B_v\left(A_{t-2}^{(n)}; \boldsymbol{\theta}_v\right)\right]\right\}. \nonumber
\end{align}
Here, we omitted the $\log \mathbb{P}^0(D_0)$ term as it does not depend on $\boldsymbol{\theta}$.
 Note that in \eqref{opt_prob}, we optimize over the truncated space $\Tilde{\Theta}$, for which we need to know the ``slack'' variables $\varepsilon$ and $\gamma$ from Assumption \ref{assump2}. As their values are inaccessible in practice, in our implementation, we set $\varepsilon = 10^{-6}$ and $\gamma = h_v - \varepsilon$ for distributions with $h_v < \infty$ and $\gamma = 10$, otherwise.
 These constraints empirically improved numerical stability compared to optimization over the untruncated space $\Theta$.
 
Examining \eqref{individ_trace_lik}, we see that the trace log-likelihood 
only involves weights $b_{u,v}$ for nodes $v$ that were either activated after time 0 or failed to become activated but had an active parent. %\liza{Don't these two sets combine to just the nodes that had at least one active parent at some point in time?  } \alex{Unfortunately, not. A node from a seed set may later obtain an active parent, but it does not mean that the edge weight leading into the seed node appears in the likelihood.} 
We will denote this set of ``informative'' nodes in trace $\mathcal{D}$ as
\begin{equation}\label{informative_node_set}
    V_c(\mathcal{D}) := \left[A(\mathcal{D}) \setminus D_{0}^{(n)}\right] \cup C(A(\mathcal{D})).
\end{equation}
Similarly, we can define the set of all informative nodes in the observed trace collection $\mathbb{D}$ as $V_c(\mathbb{D}) := \bigcup_{n=1}^N V_c(\mathcal{D}_n)$.
In principle, this set may still be smaller than $|V_c|$, meaning that the parent weights of some nodes cannot be estimated from the data. But even in this case, with a sufficiently rich trace collection, we expect the total number of parameters in \eqref{opt_prob} to be close to $|E|$, which creates a major computational problem for large networks. 

Fortunately, the optimization problem has a block structure we can use to speed up computations.  
By changing the order of summation, we can rewrite its objective as a sum of terms, each depending only on the parent edges of a single child node:
$$
\sum_{n=1}^N L(\mathcal{D}_n| \boldsymbol{\theta}) = \sum_{v\in V_c(\mathbb{D})} L_v(\boldsymbol{\theta}_v) $$
with 
\begin{align}\label{individ_node_likelihood}
      L_v(\boldsymbol{\theta}_v)  &= \sum_{n:\ v \in C(A_{T_n}^{(n)})} \log \left\{1 - F_v\left[B_v\left( A_{t(v, n)}^{(n)}; \boldsymbol{\theta}_v\right)\right]\right\} \nonumber
    \\
    + & \sum_{n:\ v \in A_{T_n}^{(n)}\setminus D_0^{(n)}} \log \left\{F_v\left[B_v\left( A_{t(v, n)}^{(n)}; \boldsymbol{\theta}_v\right)\right] - F_v\left[B_v\left(A_{t(v, n)-1}^{(n)}; \boldsymbol{\theta}_v\right)\right]\right\},
\end{align}
where we denote the last time that node $v$ is {\it not} active in trace $n$ by $t(v, n):= \arg\max\{t\le T_n: v \notin A_{t}^{(n)} \}$.
Importantly, it is not just the likelihood that conveniently separates into blocks with independent variables, but also the feasibility set $\tilde{\Theta}$ that can be rewritten as a Cartesian product of child node-specific individual parameter spaces 
 $\tilde{\Theta} =\mathop{\otimes}\limits_{v\in V_c}  \tilde{\Theta}_v$
 defined as
\begin{equation}\label{individ_param_space}
    \tilde{\Theta}_v  = \{\boldsymbol{\theta}_v \in \mathbb{R}^{|P(v)|}: \ \boldsymbol{\theta}_v \geq \varepsilon, \ \|\boldsymbol{\theta}_v\|_1 \le \gamma\}.
\end{equation}
Therefore, solving \eqref{opt_prob} is equivalent to maximizing $L_v(\boldsymbol{\theta}_v)$ over $\boldsymbol{\theta}_v \in \tilde{\Theta}_v$ for each $v \in V_c(\mathbb{D})$:
\begin{equation}\label{opt_prob_changed_sum}
    \hat{\boldsymbol{\theta}}_v = \arg\max_{\boldsymbol{\theta}_v \in \tilde{\Theta}_v}  L_v(\boldsymbol{\theta}_v), \quad v \in V_c(\mathbb{D}).
\end{equation}
Each optimization problem in \eqref{opt_prob_changed_sum} now has only $|P(v)|$ variables and $|P(v)| + 1$ affine constraints, allowing for efficient parallelized optimization.

The next natural question is whether a node-specific optimization problem in \eqref{opt_prob_changed_sum} is convex.   The feasible set $\tilde{\Theta}_v$ is a convex simplex, and the arguments of $F_v$ in \eqref{individ_node_likelihood} depend linearly on $\boldsymbol{\theta}_v$. Thus, $L_v$ is a concave function of $\boldsymbol{\theta}_v$ if $\log [F_v(x) - F_v(y)]$ is a concave function on $h_v \ge x > y \ge 0$. For example, if $U_v$ is uniformly distributed on $[0,1]$, as in the standard LT model, then $\log [F_v(x) - F_v(y)] = \log (x - y)$
is concave.  The following proposition demonstrates that it is true for all distributions with log-concave densities, and in particular when $F_v$ is the Beta distribution with parameters $\alpha_v \ge 1$ and $\beta_v \ge 1$.  The proof is given in Section \ref{weight_estim_section_proofs} of the Appendix.

\begin{proposition}\label{logconc}
    The function      $L_v(\boldsymbol{\theta}_v)$  in \eqref{individ_node_likelihood} is concave in $\boldsymbol{\theta}_v$ if $F_v$ has a log-concave density. 
\end{proposition}

Note that for a child node $v$ with just one parent $u$, \eqref{individ_node_likelihood} reduces to the log-likelihood of a $\operatorname{Bernoulli}\bigl(F_v(b_{u, v})\bigr)$ random variable, which is concave in $b_{u, v}$ only if $F_v$ is log-concave. Therefore, the log-concavity of the cdf $F_v$ is clearly necessary as well as sufficient for concavity of $L_v(\boldsymbol{\theta}_v)$ for a node $v$ with an arbitrary in-degree.

\begin{remark}  Some special cases of the GLT model may have a different natural parametrization; %rather than the edge weights $\boldsymbol{\theta}_v$. 
for example, for the IC model the natural parameters are the edge transmission probabilities $p_{u,v}$, which can be expressed in terms of the GLT weights as $p_{u,v} = 1 - \exp(-b_{u,v})$ according to Proposition \ref{ic_is_glt_propos}.   We can always estimate a reparametrized set of parameters by using the plug-in estimators, such as $\hat{p}_{u,v} = 1 - \exp(-\hat{b}_{u,v})$ for the IC model.
\end{remark}

% \begin{remark}\label{individ_node_lik_interpret_remark}
%     For further narration, it will be useful for us to keep in mind the interpretation of the two types of terms in \eqref{individ_node_likelihood}. The terms in the first sum correspond to the informative traces $n$ where $v$ was not activated and can be interpreted as the log-probabilities of parents in $A_{t(v, n)}^{(n)}$ being insufficient for its activation. The terms in the second sum correspond to traces, where $v$ was activated, and represent the log-probabilities of parents in $A_{t(v, n)- 1}^{(n)}$ being not enough for activation and parents in $A_{t(v, n)}$ being sufficient for that.
% \end{remark}

\subsection{Theoretical properties of the GLT weight estimator}\label{consist_section}
In this section, we derive a finite sample bound on the error of the MLE in \eqref{opt_prob_changed_sum} which holds with high probability,  and also the asymptotically normal distribution of the estimator. Since the weight estimation procedure separates into a collection of optimization problems, each involving only the parent weights $\boldsymbol{\theta}_v$ of a given node $v\in V_c$, we will establish the theoretical properties of the estimate $\hat{\boldsymbol{\theta}}_v$ for a given fixed node $v\in V_c$ with $P(v) = \{u_1, \ldots, u_m\}$ and the corresponding ground-truth parameter $\boldsymbol{\theta}_v^*$. The proofs of all the results in this section can be found in Section \ref{consist_section_proofs} of the Appendix.

We begin by establishing the finite sample result.  First, we introduce additional notation to conveniently encode the data used to fit the subproblem of node $v$ in \eqref{opt_prob_changed_sum}. Denote the trace indices where node $v$ is informative by $\mathcal{I}_v:=\{1\le n\le N: \ v\in V_c(\mathcal{D}_n)\}$ and for each $\mathcal{D}_n, n\in\mathcal{I}_v$, denote the time points when $v$ acquired at least one new active parent node by $\mathcal{T}_{v}^{(n)}:= \{0\le \tau\le T_n:  D_{\tau}^{(n)} \cap P(v)\ne \emptyset \}$. For each time point $t\in \mathcal{T}_{v}^{(n)}$, let $x^{(n)}_t := [\mathbf{1}(u_i \in D_t^{(n)}]_{i=1}^m$ be the indicator vector of $v$'s newly active parents at time $t$. Define also the matrix of $x^{(n)}_t$ stacked over $t\in \mathcal{T}_{v}^{(n)}$ as $X_v^{(n)}\in \{0, 1\}^{|\mathcal{T}_{v}^{(n)}|\times m}$ and the further stacked matrix of $X_v^{(n)}, n\in\mathcal{I}_v$ as $X_v \in \{0, 1\}^{N_v \times m}$ with $N_v:=\sum_{n\in \mathcal{I}_v} |\mathcal{T}_{v}^{(n)}|$. 

Importantly, the role of the ``sample size'' in our finite sample bound will be played not by the number of informative traces $|\mathcal{I}_v|$, but by $N_v$, the number of times $v$ had a non-empty newly active parent set across all traces. Note that $1\le |\mathcal{T}_v^{(n)}| \le m$ for any $n\in \mathcal{I}_v$, since $v$ can have no more than $m$ different newly active parent sets throughout a progressive propagation, but should have at least one such set, since the trace is informative. This means that in the best-case scenario, $N_v$ can be up to $m$ times larger than $|\mathcal{I}_v|$. This aligns with other work,  for example, with the finite sample bound for the parent weight estimator derived for the General Cascade model in Theorem 1 of \citep{InferringGraphs2015}. % also depends on the number of such observations but not the number of observed traces.
Similarly to their framework, we assume that within each informative trace $n\in\mathcal{I}_v$, we observe the propagation history up to the last time point $t(v, n)$ when $v$ is not active, and that the only randomness is in the activation event of the node $v$. Since the matrix $X_v$ essentially encodes all such histories, we condition on $X_v$ whenever we need to emphasize that these histories are observed.

% \begin{remark} 
% Given that the number of terms in the optimized likelihood in \eqref{individ_node_likelihood} is $|\mathcal{I}_v|$, it may seem surprising that the MLE error may improve as a function of $N_v$. However, by looking at the structure of each likelihood term, it becomes clear that the ``informativeness'' of each term increases with the number of time points at which $v$ got a new active parent. For example, assume that the trace, where $v$ did not activate, would stop earlier at some $\tau < t(v, n)$. Then, we would observe its threshold $U_v$ in an interval $\bigl(B_v(A_{\tau}; \boldsymbol{\theta}_v), +\infty\bigr)$, which clearly contains the smaller interval we would observe if the trace stopped at $t(v, n)$, and thus contains less information on $U_v$. For a more formal explanation, see Section \ref{consist_theorem_proof_section} of the Appendix.
% \end{remark}

Before we state the main result, we make a couple of additional mild regularity assumptions.  
First, we require that the threshold cdfs are not only invertible as in Assumption \ref{assump1}, but also sufficiently smooth, and, secondly, that the ground-truth parent weights of every node lie in the interior of the corresponding parameter space:
\begin{assumption}\label{assump2} For each $v\in V_c$, the threshold cdf $F_v$ is strictly monotone and three times continuously differentiable, and the ground-truth weights $\boldsymbol{\theta}_v^*$ satisfy
$\boldsymbol{\theta}^*_v > \varepsilon$ and $\|\boldsymbol{\theta}_v^*\|_1 < \gamma$.
\end{assumption}

\noindent We also require that the negative log-likelihood is almost surely non-strictly convex and strictly convex on average for any $\boldsymbol{\theta}_v$ in the parameter space:
\begin{assumption}[Convexity] \label{invertible_hessian_assump}
   The density of $F_v$ is log-concave to ensure that, per Proposition \ref{logconc}, $-L_v(\boldsymbol{\theta}_v)$ is non-strictly convex. Moreover, the expected Hessian of $-L_v(\boldsymbol{\theta}_v)$ 
   should be positive definite everywhere on $\tilde{\Theta}_v$. With compactness of $\tilde{\Theta}_v$ and continuity of the Hessian guaranteed by Assumption \ref{assump2}, this means that there exists $\lambda_{\min} > 0$ such that for every $\boldsymbol{\theta}_v \in \tilde{\Theta}_v$, it holds conditional on $X_v$:
    $$-\mathbb{E}\bigl[{1\over N_v}\nabla^2 L_v(\boldsymbol{\theta}_v)\bigr] \succeq \lambda_{\min}I_{|P(v)|}.
    $$
\end{assumption}
\noindent For the LT, IC, and Beta-GLT models, Assumption \ref{invertible_hessian_assump} can be replaced by a much more intuitive sufficient condition requiring non-degeneracy of $X_v$:
\begin{proposition}\label{ic_and_lt_satify_concavity}
    With $F_v \sim \operatorname{Exponential}(1)$ or $F_v \sim \operatorname{Beta}(\alpha, \beta)$ with $\alpha,\beta \ge 1$, Assumption \ref{invertible_hessian_assump} is satisfied if $X_v$ has a full column rank. 
    Moreover, $\lambda_{\min}$ can be set as the smallest eigenvalue of a Gram matrix ${c_\lambda} X_v^\top X_v / N_v$ 
    % for the IC and LT models and (2) ${c\over N_v}X_v^\top X_v$ for $\operatorname{Beta}(\alpha, \beta)$-GLT models with $\alpha, \beta \ge 1$ except for $\alpha=\beta = 1$.
    where $c_\lambda$ is a constant that depends only on $(F_v, \varepsilon, \gamma)$.
\end{proposition}
\noindent  Notice that the condition in Proposition \ref{ic_and_lt_satify_concavity} is guaranteed to hold asymptotically if the identifiability condition of Theorem \ref{identif_theorem} is satisfied. Indeed, with positive probability, $X_v$ includes each row of the $m\times m$ invertible identifiability matrix $\tilde{X}_v$ defined in Theorem \ref{identif_theorem}, and thus has the full column rank itself.

Now, we are ready to state the main theoretical result.

\begin{theorem}\label{mle_consistency_theorem}
    Consider the MLE $\hat{\boldsymbol{\theta}}_v$ obtained by solving the optimization problem \eqref{opt_prob_changed_sum} and fix an arbitrary $\delta \in (0,1)$. Under Assumptions \ref{assump2}, \ref{invertible_hessian_assump}, and the assumption of i.i.d.\ traces, $\hat{\boldsymbol{\theta}}_v$ satisfies  the following concentration bound conditional on $X_v$ as long as $N_v \ge {c_0m \over \lambda_{\min}}  \log{2m\over \delta}$:
    \begin{equation} \label{finite_sample_bound_mle_error} \mathbb{P}\Bigl[\|\hat{\boldsymbol{\theta}}_v - \boldsymbol{\theta}_v^*\|_2 \le {C_0\over \lambda_{\min}}\sqrt{{m\over N_v}\log {4m\over\delta}}  \Bigr] \ge 1-\delta.
    \end{equation}
    Here, $C_0, c_0 > 0$ are constants depending only on $(F_v, \varepsilon, \gamma)$.
\end{theorem}
\begin{remark}
    Since in our framework, the underlying network and thus the node's indegree $m=|P(v)|$ are fixed, the established result presents sufficient conditions for the estimator $\hat{\boldsymbol{\theta}}_v$ to be $\sqrt{N_v}$-consistent. However, in principle, we could consider a sequence of networks with associated GLT models on them and apply this result to each network to explore how the estimator's error changes when the sample size $N_v$ and the indegree $m$ grow simultaneously. For this scenario, the derived concentration bound suggests that consistency holds as long as $N_v$ asymptotically dominates $m\log m$. Unfortunately, Proposition \ref{ic_and_lt_satify_concavity} shows that in the final sample case, the necessary Assumption \ref{invertible_hessian_assump} may not be satisfied when $m$ is too close to $N_v$, as this would lead to the singularity of the log-likelihood Hessian.
\end{remark}

In addition to a finite sample result, we derive an asymptotic distribution for the MLE error. Contrary to a finite sample statement of Theorem \ref{mle_consistency_theorem}, establishing an asymptotic result requires considering the full trace-generating distribution in \eqref{individ_trace_prob} with the total number of traces $N$ being a more natural candidate for the sample size than $N_v$. Another difference is that we need much weaker conditions on the log-likelihood convexity -- now it is sufficient to require it only locally in the neighborhood of $\boldsymbol{\theta}^*$. We state this as a separate assumption:

\begin{assumption}[Local convexity] \label{local_convexity_assumop}
The negative expectation of the trace log-likelihood Hessian is positive definite at $\boldsymbol{\theta}^*$:
$$\mathbb{E}_{\mathcal{D}\sim\mathbb{P}_{\boldsymbol{\theta}^*}}\bigl[-\nabla^2\log \mathbb{P}_{\boldsymbol{\theta}}(\mathcal{D})\bigr] \rvert_{\boldsymbol{\theta} = \boldsymbol{\theta}^*}\succ 0.
$$
\end{assumption}

The following proposition characterizes the asymptotic behavior of the MLE.
\begin{proposition}\label{asympt_normality_propos}
    Consider the MLE $\hat{\boldsymbol{\theta}}_v$ obtained by solving the optimization problem \eqref{opt_prob_changed_sum}. Then, under Assumptions \ref{assump2}, \ref{local_convexity_assumop}, identifiability condition on the seed distribution $\mathbb{P}^0$ in Theorem \ref{identif_theorem}, and the assumption of i.i.d.\ traces, it holds
    \begin{equation}\label{asympt_normal_mle} 
       \hat{\Sigma}_v(\hat{\boldsymbol{\theta}}_v)^{-1/2}(\hat{\boldsymbol{\theta}}_v - \boldsymbol{\theta}^*_v)\stackrel{\mathbb{P}_{\boldsymbol{\theta^*}}}{\longrightarrow} \mathcal{N}(0, I_{m}) \quad \text{as } \ N\rightarrow \infty,
    \end{equation}
    where the limit is taken with respect to the ground-truth trace distribution $\mathbb{P}_{\boldsymbol{\theta^*}}$ defined in \eqref{trace_distrib_def_eq} and the estimated covariance matrix is $\hat{\Sigma}_v(\boldsymbol{\theta}_v) = [-\nabla^2 L_v(\boldsymbol{\theta}_v)]^{-1}$.
\end{proposition}

% \begin{remark}
%     Notice that instead of using Assumption \ref{invertible_hessian_assump}, which essentially required the traces to contain a sufficient number of node-wise observations $N_v$, here we used the identifiability condition of Theorem \ref{identif_theorem}, which guarantees that $N_v$ grows to infinity together with $N$. 
% \end{remark}

There are many downstream tasks for which the asymptotic distribution of the GLT weights may be useful. For example, \cite{chen2016robustinfluencemaximization} propose a method for solving the Influence Maximization problem (defined formally in Section \ref{IM_section}) assuming that the edge weights are only known to lie in some intervals. This method can be naturally coupled with estimation of these confidence intervals from the data using the normal approximation in \eqref{asympt_normal_mle}.

Uncertainty quantification for the estimated weights may also be of interest.  As a simple example, consider comparing the effects of two parent nodes $u$ and $w$ on a child node $v$.   We may then want to test the hypothesis of no difference between the corresponding edge weights, 
$$H_0: \ b_{u, v} - b_{w, v} = 0\quad \text{vs}\quad H_a: \ b_{u, v} -  b_{w, v} \ne 0,
$$
and use the difference of their estimates as the test statistic. Then, the reference distribution is normal with zero mean and variance that can be derived from the asymptotic covariance matrix of $\hat{\boldsymbol{\theta}}_v$ in \eqref{asympt_normal_mle}:
$$\operatorname{Var}[\hat{b}_{u, v} - \hat{b}_{w, v}] \ \approx \ \hat{\Sigma}_{v, uu} + \hat{\Sigma}_{v, ww} - 2\hat{\Sigma}_{v, uw},
$$
where $\hat{\Sigma}_{v, ij}$ denotes the $(i,j)$-th entry of $\hat{\Sigma}_{v}(\hat{\boldsymbol{\theta}}_v)$.

In other applications, knowing the asymptotic distribution of the weights can help quantify uncertainty in predicting various quantities of interest;  for example, in epidemiology we may be interested in predicting the probability of node activation (infection) in the next time step given the propagation history.  This is a complicated function of the weights but one can still compute a confidence interval for it using the delta method.   %$(D_0, \ldots, D_{t-1})$. Assuming that GLT is the diffusion model that governs the propagation, we can first estimate $\hat{\boldsymbol{\theta}}_v$ and then plug it into \eqref{trans_prob_glt} for prediction. The corresponding confidence interval can be built by applying Delta method to $g(\boldsymbol{\theta}_v):= \mathbb{P}^t_{\boldsymbol{\theta}_v}(v \in D_t |\ D_0, \ldots, D_{t-1})$, which gives the asymptotic variance $\nabla g(\hat{\boldsymbol{\theta}}_v)^\top  \hat{\Sigma} \ \nabla g(\hat{\boldsymbol{\theta}}_v)$.

\subsection{Extension to partially-observed traces.} \label{pseudo_trace_section}
In many applications, we do not observe a full propagation trace, but we know which of the node's parents were active  before it was activated.  % For example, suppose a person (node) gets infected given their other family members (parent nodes) are currently infected. In that case, we can infer that the virus transmission from one of the family members occurred while we may not know the full virus propagation path leading to their own infection. 
We write each such observation for node $v$ as a pair $(A_v, y)$, referred to as {\it pseudo-trace},  where $A_v \subset P(v)$ is a set of $v$'s active (infected) parents and $y \in \{0, 1\}$ is an indicator of the event that $A_v$ together activate (infect) $v$.  

Suppose that for each child node $v\in V_c$, we observe a possibly empty collection of pseudo-traces 
\begin{equation}\label{nodes_pseudotraces}
    \mathbb{D}_v = \{(A_v^{(n)}, y_v^{(n)}), \ n = 1, \ldots, N_v\}
\end{equation}
that we would like to use to estimate the GLT weights $\boldsymbol{\theta}_v$.  If we could specify a pseudo-trace generating distribution for $\mathbb{D}_v$, we could apply the likelihood approach.  One way to do this is to treat a pseudo-trace $(A_v, y)$ as a trace seeded at $A_v$ and propagating in a star graph $G_v$ attached to the child node $v$. Given that $A_v \subset P(v)$, the feasible traces on $G_v$ can only be of two types:  those that stopped immediately at the seed set $A_v$ (corresponding to the pseudo-trace $(A_v, 0)$) and those that activated $v$ at $t=1$ and then stopped (pseudo-trace $(A_v, 1)$). 
 Assuming that the sets $A_v^{(n)},\ n=1, \ldots, N_v$ are independently generated from a parameter-free seed distribution $\mathbb{P}^0_v$ supported on the subsets of $P(v)$, the pseudo-trace likelihood has the form
\begin{equation}\label{pseudo_trace_probability}
\mathbb{P}_{\boldsymbol{\theta}_v}(A_v, y) = 
\mathbb{P}^0_v(A_v)\left\{1 - F_v\left[B_v(A_v; \boldsymbol{\theta}_v)\right]\right\}^{1-y} F_v\left[B_v(A_v; \boldsymbol{\theta}_v)\right]^y.
\end{equation}
Aggregation of these terms across all pseudo-traces in $\mathbb{D}_v$ results in the log-likelihood 
\begin{equation}\label{individ_pseudo_trace_likelihood}
      L_v^{pt}(\boldsymbol{\theta}_v)  = \sum_{n:\ y^{(n)} = 0} \log \bigl\{1 - F_v\bigl[B_{v}\bigl(A_{v}^{(n)}; \boldsymbol{\theta}_v\bigr)\bigr]\bigr\} + \sum_{n:\ y^{(n)} = 1} \log F_v\bigl[B_{v}\bigl(A_{v}^{(n)}; \boldsymbol{\theta}_v\bigr)\bigr] , 
\end{equation}
where we omitted the parameter-free terms $\log \mathbb{P}^0_v(A_v)$.   This pseudo-trace likelihood coincides with the that of the General Cascade model (see Section 2.3 of \citep{InferringGraphs2015}), which assumes that the activation probability of a node is an increasing function $f$ mapping the sum of incoming edge weights from active parents to $[0, 1]$. This implies that the model is the analogue of the GLT model for pseudo-trace case with $f$ set as the cdf $F_v$.

The assumption that the seed distribution $\mathbb{P}_v^0$ does not depend on any diffusion model parameters may seem strong.  %, as a seed node in $A_v$ could have been activated by its own network parents, making the probability of $A_v$ indirectly dependent on the seed's incoming edge weights. 
However, by carefully comparing \eqref{individ_pseudo_trace_likelihood} with its counterpart for fully observed traces in \eqref{individ_node_likelihood}, we observe that the only difference is that, in the latter case, we always subtract $F_v\bigl[B_v(A^{(n)}_{t(v, n) - 2})\bigr]$ under the logarithm for traces where $v$ was activated. 
This term represents the probability that the active parent set preceding the one that eventually activated $v$ was \textit{not} enough for $v$'s activation. Thus, the only information lost in a pseudo-trace, compared to a fully observed trace, is which parent subset of an influenced node $v$ was \textit{not} sufficient to activate it. Importantly, since pseudo-traces are assumed as traces propagating on a star graph $G_v$, both consistency and asymptotic normality results of Theorem \ref{mle_consistency_theorem} and Proposition \ref{asympt_normality_propos} still hold in the pseudo-trace case.

\subsection{Estimation of threshold parameters}\label{est_threshold_params}

%Notice that both Problem \ref{opt_prob} and Problem \ref{indiv_opt_prob} 
So far, we have treated threshold distributions as known, which is unlikely in reality.   While it would be challenging to estimate these distributions fully nonparametrically given we typically only observe a limited number of traces concerning any given node, we could easily obtain an estimate if we model each $F_v, v\in V$ as a member of some parametric family with parameters $\boldsymbol{\varphi}_v \in \Phi_v \subset \mathbb{R}^{r_v}$. For example, if we model $F_v \sim \operatorname{Beta}(\alpha_v, \beta_v)$, we can define $\boldsymbol{\varphi}_v = (\alpha_v, \beta_v)$ with $\Phi_v = [1, + \infty)^2$ to satisfy the condition of Proposition \ref{logconc}. Then we can estimate $(\boldsymbol{\theta}_v, \boldsymbol{\varphi}_v)$ for each $v \in V_c(\mathbb{D})$ by solving the following optimization problem:
\begin{equation}\label{individ_opt_prob_two_sets}
    \max_{ \boldsymbol{\varphi}_v \in \Phi_v, \  \boldsymbol{\theta}_v \in \Theta_v}   L_v\left(\mathbb{D}| \boldsymbol{\theta}_v, \boldsymbol{\varphi}_v \right) \, , 
\end{equation}   
where the individual node likelihood $L_v$ is obtained from \eqref{individ_node_likelihood} with the Beta distribution cdfs plugged in.   Allowing $F_v$ to vary within the feasible set makes the optimization problem non-convex even in the simplest case of a one-parameter $\operatorname{Beta}$ family.   Thus finding even a local optimum of  \eqref{individ_opt_prob_two_sets} requires careful tuning of the gradient steps since $\boldsymbol{\theta}$ and $\boldsymbol{\varphi}$ might have very different magnitudes. A natural way to deal with that problem is to switch to coordinate gradient descent, alternating between fixing one set of variables ($\boldsymbol{\theta}_v$ or $\boldsymbol{\varphi}_v$) and optimizing over the other one.  However, our numerical experiments (available in the GitHub repository) showed that this type of coordinate gradient descent converges reliably only if the initial values are sufficiently close to the truth.     
Therefore, unless the dimension $r_v$ of $\Phi_v$ is very high, we choose $\Phi_v$ from a discrete grid, optimize \eqref{individ_opt_prob_two_sets} over $\boldsymbol{\theta}$ for each $\boldsymbol{\varphi} _v\in \Phi_v$  and choose the one resulting in the highest log-likelihood. 
% An example of such optimization when all node thresholds follow the same $\operatorname{Beta}(\alpha^*, \beta^*)$ distribution (with $\alpha^*, \beta^*$ unknown) can be found in the Section \ref{beta_cv_section} of the Appendix. 
While in further numerical experiments we do assume all nodes' thresholds follow the Beta distribution, it is important to note that due to the node-wise separability of the optimization problem in  \eqref{individ_opt_prob_two_sets}, the parametric family as well as the parameter grid $\Phi_v$  are not required to be the same across $v\in V_c(\mathbb{D})$.  

% Unfortunately, if we add more threshold parameters to the model, the grid search approach quickly becomes computationally infeasible. A possible alternative is the method of \cite{Turnbull1976}, who proposed a non-parametric estimator of the cdf from a set of intervals within which the corresponding (censored) random variable was independently observed. In the GLT model case, for each node $v$, we observe independent realizations $U_v^{(n)}, n=1, \ldots, N$ of its activation threshold, each within a weight-dependent interval. This interval is $ \left(B(v, A_{t(v, n)-2}^{(n)}), \right. \left. B(v, A_{t(v, n)-1}^{(n)}) \right]$ for traces where $v$ was activated, and $\left(B(v, A_{T_n}^{(n)})  \right. \left.\infty\right)$ for traces where it was not. This suggests a natural iterative procedure that alternates solving \eqref{opt_prob_changed_sum} separately for each $\boldsymbol{\theta}_v$ and estimating $F_v$ nonparametrically. We leave investigation of this method for future work.

\section{Influence maximization under the GLT model}\label{IM_section}
In this section, we study the GLT model in the context of the Influence Maximization (IM) problem, that is, the task of choosing a seed set of a given size that maximizes the expected spread of information through the network. We start with a brief review of the IM problem, which was introduced by \cite{richardson} and further formalized by \cite{Kempe}.  Formally, we define the {\it influence function}, 
%\begin{definition}[Influence function]
    for a given simple directed graph $G=(V, E)$ and a diffusion model $M$, as the function that maps any subset $S \subset V$ to the expected number of nodes influenced if $M$ is seeded by $D_0 = S$, 
$$\sigma_{ M}(S)  =  \mathbb{E}_{\mathcal{D}|D_0=S}|A(\mathcal{D})|, \qquad S\subset V. 
$$
%\end{definition}

%\begin{definition}[The Influence Maximization (IM) problem]

The IM problem is then to find a subset $S^* \subset V$ which maximize the influence function over all subsets of $V$ of a given size $k$, 
$$S^* = \operatorname*{argmax}_{S\subset V:\ |S |\le k}\sigma_{ M}(S).$$
% For the rest of the paper, we omit the subscripts $G$ and $M$ of $\sigma_{ M}$ whenever they are clear from the context.  
%\end{definition} 
\cite{Kempe} showed that the IM problem is NP-hard under the LT, IC, and Triggering models. However, if the influence function has certain properties, the optimal solution can be well approximated by a greedy strategy (see Algorithm \ref{algorithm:greedy}). These properties are monotonicity and submodularity.  
\begin{definition}[Monotonicity] An influence function $\sigma(\cdot)$ is monotone if $\sigma\left(S^{\prime}\right) \leq \sigma(S)$ for any $S^{\prime} \subset S \subseteq V$.  In words, increasing the size of the seed set cannot decrease the value of the influence function. 
\end{definition} 
\begin{definition}[Submodularity] \label{submod_def} An influence function $\sigma(\cdot)$ is submodular if $\sigma\left(\{v\} \cup S^{\prime}\right)-\sigma\left(S^{\prime}\right) \geq \sigma(\{v\} \cup S)-\sigma(S)$ for any $S^{\prime} \subset S \subseteq V$ and $v \in V \backslash S$.  In words, adding a node to a given seed set increases the influence function by at least as much as adding it to a superset of this seed set.  
\end{definition}
These are both reasonable and mild assumptions, reflecting the intuitive meaning of information propagation.     
Though \cite{Kempe} were the first to study the greedy algorithm behavior in the context of the IM problem, worst-case performance under monotonicity and submodularity assumptions dates back to the following classical result: 
\begin{theorem}[\cite{ga_proof1} and \cite{ga_proof2}]\label{ga_theorem}
Let $\sigma: 2^{V} 
 \rightarrow \mathbb{R}{+}$ be a a monotone and submodular influence function. Let $\hat{S} \subset V$ of size $k$ be the set obtained by selecting elements from $V$ one at a time, when at each step one chooses an element that provides the largest marginal increase in the value of $\sigma$. Let $S^*$ be the true maximizer of $\sigma$ over all $k$-element subsets of $V$. Then 
 \begin{equation}\label{opt_ineq}
\sigma(\hat{S}) \ge \left(1-\frac{1}{\mathrm{e}} \right)\sigma(S^*) \, , 
\end{equation}
that is, $\hat{S}$ provides a $1-\frac{1}{\mathrm{e}}$ approximation to the optimal $S^*$.      
\end{theorem}

 \cite{Kempe} showed that under the LT, IC, and Triggering models, the influence function has the monotonicity and submodularity properties,  thus showing that the greedy algorithm can solve the IM problem under these models with the optimality guarantee \eqref{opt_ineq}. A natural question about our new GLT model is whether it enjoys similar properties. Monotonicity is trivially satisfied for the GLT model class,  and the following theorem gives a sufficient condition for submodularity, which is also necessary in a particular sense. The proof of this theorem and all other results in this section can be found in Section \ref{IM_section_proofs} of the Appendix.

 \begin{theorem}\label{GLT_submodular_theorem}
     A GLT model on a graph $G=(V, E)$ has a submodular influence function if the threshold cdf $F_v$ is concave for each child node $v\in V_c$. Moreover, for any non-concave cdf $F$ with a non-negative support, there exists an instance of the GLT model on a star graph as in Figure \ref{fig:instar_graph} with $F$ as the threshold cdf of the only child node, such that the corresponding influence function is not submodular.
 \end{theorem}
Theorem \ref{GLT_submodular_theorem} implies submodularity of the IC and LT models since both uniform and exponential distributions have concave cdfs. For GLT models with beta-distributed thresholds, we can use this general statement to derive a simpler submodularity condition. \begin{corollary}\label{beta_submod_corol}
     A GLT model on a graph $G=(V, E)$ with $U_v \sim \operatorname{Beta}(\alpha_v, \beta_v)$ has a submodular influence function if $\alpha_v \le 1$ and $\beta_v \ge 1$ for all $v\in V_c$.
 \end{corollary}
\begin{proof}
      The second derivative of the cdf of the $\operatorname{Beta}(\alpha, \beta)$ distribution is given by 
    $$F''(x| \alpha, \beta)\ = C(\alpha, \beta) \ x^{\alpha-2}(1-x)^{\beta - 2}[(\alpha-1)(1-x) - (\beta-1)x] 
    $$
 where $C(\alpha, \beta) > 0$ does not depend on $x$, and thus $F''(x| \alpha, \beta)\ \le 0$ for all $x \in (0, 1)$  if $\alpha \le 1$ and $\beta \ge 1$.
\end{proof}

\begin{algorithm}[t!]
\caption{The greedy algorithm for the IM problem}
\begin{algorithmic}
%\Procedure{D2 Clustering}{}
\State {\bf Input:} graph $G = (V, E)$, diffusion model instance $M$,
and seed budget $k$ 
\State $S \leftarrow \emptyset$
\While{$|S| < k$} 
    \State $v \leftarrow \arg \max _{v \in V\setminus S}\left(\sigma_{ M}(S\cup \{v\})-\sigma_{ M}(S)\right)$
    \State $S \leftarrow S \cup \{v\}$
\EndWhile
\State \Return{S}

% \EndProcedure
\end{algorithmic}
\label{algorithm:greedy}
\end{algorithm}

Since we estimate the parameters of GLT model from data, the estimation error in $\hat{\boldsymbol{\theta}}$ can affect the IM problem solution. To study this, let $\sigma_{\boldsymbol{\theta}}$ with $\boldsymbol{\theta}\in \Theta$ denote the influence function of a diffusion model $M_{G, \boldsymbol{\theta}}$.  Let $S^*(\boldsymbol{\theta}) = \arg\max_{|S|\le k} \sigma_{\boldsymbol{\theta}}(S)$ denote the solution of the IM problem under $M_{G, \boldsymbol{\theta}}$, where we omit the dependency on $k$ and treat it as fixed throughout this section. Then the question is to relate the difference $|  \sigma_{\boldsymbol{\theta}}(S^*(\hat{\boldsymbol{\theta}})) - \sigma_{\boldsymbol{\theta}}(S^*(\boldsymbol{\theta}))|$ to the parameter estimation error 
$\|\hat{\boldsymbol{\theta}} -\boldsymbol{\theta}\|$. Unfortunately,  general results which hold for an arbitrary graph topology and any choice of the true parameters usually imply loose and impractical bounds (see, for example, Lemma 3 and subsequent discussion in \cite{chen2016robustinfluencemaximization}).
Therefore, we present a less general but more illustrative result for a family of directed bipartite graphs, that is, graphs with two disjoint sets of nodes and all edges going from a node in the first set to a node in the second. The bipartite graph structure can be thought of as the most general graph topology that ensures the propagation traces are at most of unit length. The following proposition essentially states that in this setting, the discrepancy between the spreads from the IM solution under the true and estimated GLT models is governed by the $\ell_1$ error of the weight estimates.
\begin{proposition}\label{spread_error_bound_propos}
    Consider a directed bipartite graph and a GLT model on it, and assume every child node $v\in V_c$ has a $L$-Lipschitz cdf, that is, $|F_v(x) - F_v(y)| \le L |x - y|$ for any $0\le x, y\le h_v$. Then with the notations above, it holds:
    $$ | \sigma_{\boldsymbol{\theta}}(S^*(\hat{\boldsymbol{\theta}})) - \sigma_{\boldsymbol{\theta}}(S^*(\boldsymbol{\theta}))| \le 2L\|\hat{\boldsymbol{\theta}} -\boldsymbol{\theta}\|_1.
    $$
    In particular, if all threshold cdfs of a GLT model are differentiable and concave, and therefore satisfy the submodularity condition in Theorem \ref{GLT_submodular_theorem}, the Lipschitz constant can be taken as $L=\max_{v\in V_c}F_v'(0)$.
\end{proposition}
Combining this result with the finite-sample error bound derived in \eqref{finite_sample_bound_mle_error}, we can conclude that the spread of the IM solution obtained under the estimated model converges in probability to the spread of the ground-truth model solution at a rate of $\min_{v\in V_c}\sqrt{N_v}$.

    \section{Experiments}
	
	\label{ch:experiments}
	
% In this section, we compare Algorithm \ref{algorithm:EMalg} and problem \eqref{opt_prob_changed_sum} to existing heuristics and EM-based method of \cite{Saito} for the IC model in terms of weight estimation accuracy, selection of the best seed set for influence maximization, and information spread prediction.  Additionally, we explore how the estimation accuracy of our methods depends on the density and size of the network.  

In this section, we present numerical results on both simulated and real-world data.  The code for these analyses is available at  
\url{https://github.com/AlexanderKagan/gltm_experiments}.  The Python package \texttt{InfluenceDiffusion}, available at \url{https://github.com/AlexanderKagan/InfluenceDifusion}, includes the convex optimization method and the greedy algorithm to fit the GLT model, as well as code for trace sampling and spread estimation.  Whenever we fit the GLT model by solving the optimization problem in equation \eqref{opt_prob_changed_sum}, we use the SciPy implementation of the SLSQP solver. 

\subsection{Simulation settings}\label{simulation_setting_section}
We generate synthetic networks from the connected Watts-Strogatz model \citep{WattsStrogatz1998}. As the original model is for an undirected graph, we double each sampled undirected edge to go both ways.  We denote the distribution of a directed graph $G$ generated this way by $G \sim  \operatorname{CWS}(n, k, p)$, where $n$ is the number of nodes, $k$ is the initial degree of each node, and $p\in [0, 1]$ is the probability of edge rewiring, controlling the randomness of the graph.   By construction, the number of edges in $G$ is fixed to $kn$ and therefore the edge density is fixed at $k / n$.    The imposed connectedness of $G$ is not strictly necessary for modeling information diffusion, but it is convenient in simulations, ensuring that enough nodes are reachable.

To generate edge eights for $G$, we independently and uniformly sample parent edge weights of each child node from a simplex scaled by a given positive constant $d_{\max}$, which upper bounds the node's weighted in-degree:
\begin{equation}\label{weight_generating_interior_simplex}
\boldsymbol{\theta}_v  \sim \operatorname{Unif}\{\boldsymbol{w} \in \mathbb{R}^{|P(v)|}: \boldsymbol{w} \ge 0, \ \|\boldsymbol{w}\|_1 \le d_{\max}\}.
\end{equation}
To generate seed sets for the traces, we independently and uniformly sample them from node sets of sizes between 1 and $s_{\max}$:
\begin{equation}\label{seed_set_generator}
    D_0 \sim \operatorname{Unif}\{S\subset V: 1\le |S|\le s_{\max}\}.
\end{equation}
%\alex{I realized that in some experiments where "I need more data" I use 10-node seet sets, so I decided to change this size between experiments.}
Unless otherwise stated, we use $d_{\max}=1$ and $s_{\max} = 5$ as default values.
% To satisfy the in-degree constraints of the GLT model, $d_{\max} \le \min_{v\in V} h_v$. 

%\begin{equation}\label{seed_sampling}
%    \mathbb{P}^0 =\operatorname{Unif}\{D_0 \subset V: 1 \le |D_0| \le 5\}.
%\end{equation}
%For brevity, we will further omit the prefix ``seeded'' when referring to a seeded diffusion model and assume that its seed distribution is as above.

The difference between two vectors $\boldsymbol{y}$ and $\hat{\boldsymbol{y}}$ (the truth and the estimator) will be measured by Relative Mean Absolute Error (RMAE), defined as ${\|\boldsymbol{y} - \hat{\boldsymbol{y}}\|_1  /  \|\boldsymbol{y}\|_1} $.

\subsection{Estimation of edge weights} \label{lt_weight_estim_exper}
In this section, we study how the quality of weight estimation depends on key parameters of the underlying graph and the ground-truth GLT model. For simplicity, we use the original LT model here, with uniform thresholds, as the observed trends in weight estimation are very similar across different threshold distributions.  

In the first experiment, we fix the number of traces at $N=2000$ while varying the number of graph nodes $n$ and the Watts-Strogatz model average node in-degree $k$.  The weights are sampled as in \eqref{weight_generating_interior_simplex} with $d_{\max}=1$. In the second experiment, we fix the graph size to $n=100$ nodes and the in-degree to $k=10$, and vary the number of traces $N$ and the maximum weighted in-degree of the nodes $d_{\max}$, sampling the weights from \eqref{weight_generating_interior_simplex} with $d_{\max} \in \{0.2, 0.4, 0.6, 0.8, 1\}$. In both scenarios, we use $p=0.2$ as the probability of edge rewiring. 

The results are presented in Figure \ref{fig:LT_model_variation}.
 In the left panel of Figure \ref{fig:LT_model_variation}, we observe that the estimation error increases as the density $k / n$ or the network size $n$ grow. This is expected, as both higher density and larger size increase the number of edge weights to estimate, thus requiring more traces for accurate estimation.  The right panel of Figure \ref{fig:LT_model_variation} shows that lower weights lead to higher estimation errors.  This is because larger weights result in higher node activation probabilities, producing longer traces with more data that can be used in estimation.

 %In this experiment, we use a synthetic graph with $100$ nodes generated from the directed Erdös-Renyi (ER) model \citep{erdos}, where each directed edge has a probability of $p = 0.1$. 

\begin{figure}[tbh!]
    \centering
    \includegraphics[width=0.48\linewidth]{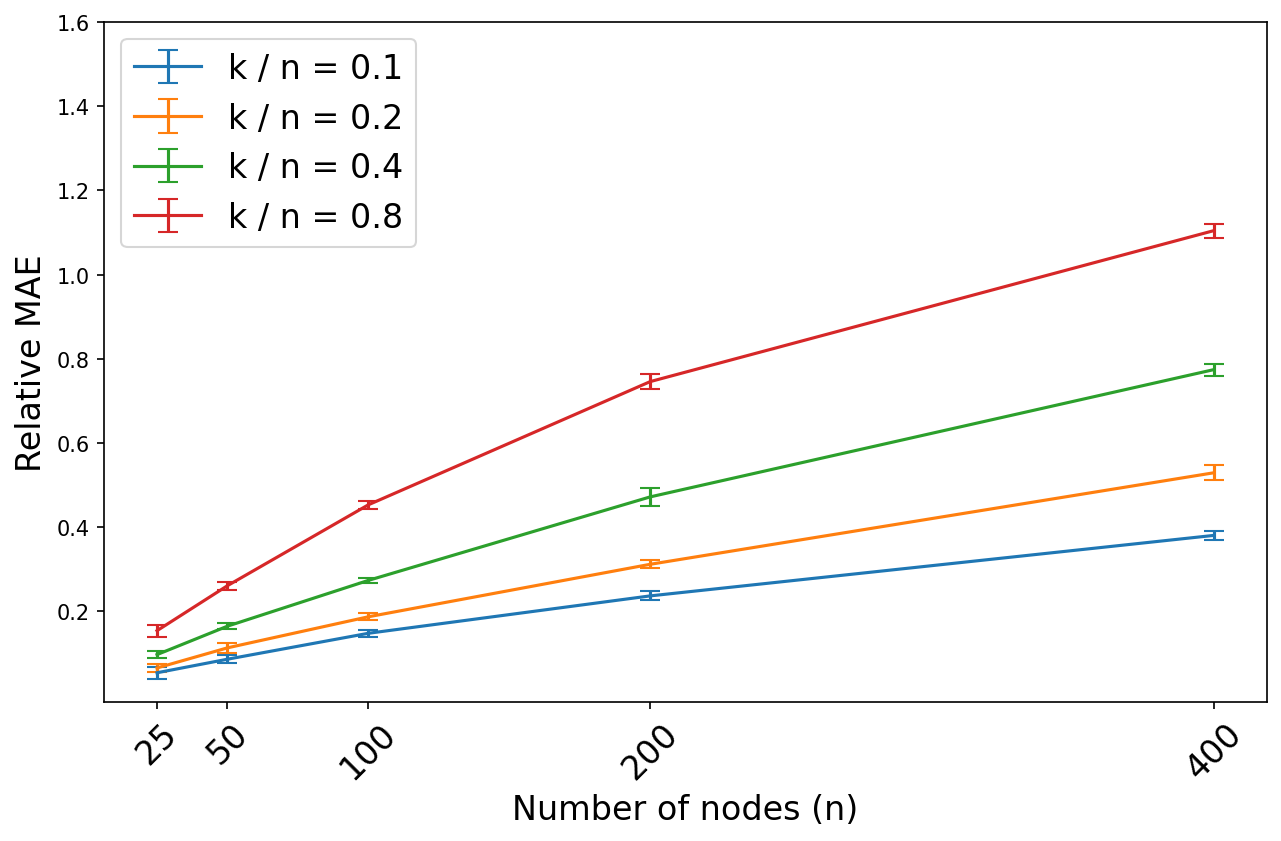} 
     \includegraphics[width=0.47\linewidth]{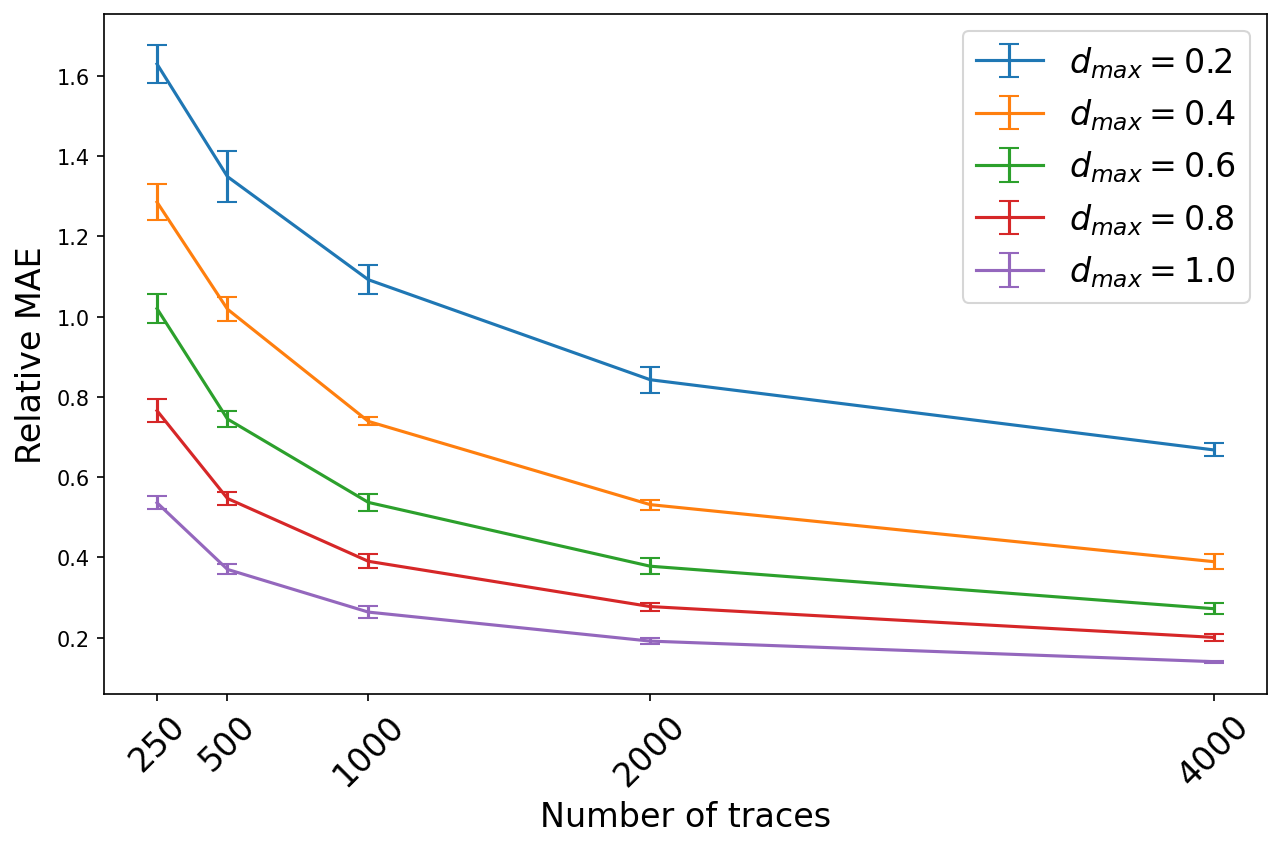} 

    \caption{Left: Relative MAE of the LT estimator as a function of the number of nodes $n$ for different densities $k / n$, with $N=2000$ traces and $d_{\max} = 1$. Right: Relative MAE as a function of the number of traces $N$ for different maximum node in-degrees $d_{\max}$, with $n = 100$ and density $k/n = 0.1$.  The error bars represent two standard errors and are calculated from 10 repetitions of each experiment.} 
\label{fig:LT_model_variation}
\end{figure}

% \begin{center}
%     \includegraphics[width=11cm, height=7cm]{figures/weight_estimation_compar.png}
% \end{center}
% Figure \ref{fig:weight_model_relationship} shows that the {\bf LT} approach consistently outperforms both heuristics approaches as the number of traces grows. Interestingly, {\bf WC} dominates  {\bf PTP} in both scenarios even though {\bf WC} contrary to {\bf PTP} completely ignores the training data. A possible explanation may be that the error of models wrongly assuming the true node in-degrees are equal to one, is smaller when the parent edges of each node have equal weights. 

% However, it may be caused by the fact that the {\bf Opt} model is the only one that estimates the ground truth model directly, so further experiments should be conducted here to support this observation.   \liza{I am not sure what the last sentence is saying;  maybe we can just delete it. } 

% \textcolor{red}{(We may need more simulation settings, such as SBM(2) and settings in favor of the benchmarks, e.g. benchmark (1)), and we can also vary the parameter settings, such as the number of nodes and the density of the network}

\subsection{Uncertainty quantification and robustness to model misspecification}\label{node_activation_experiment_section}
Here, we focus on node activation probabilities as the main object of interest and present experiments that use asymptotic theory developed in Section \ref{consist_section} to quantify the uncertainty in their estimation.  
We also study how node activation probabilities behave under misspecification of the threshold distribution.

For this experiment, we sample a $100$-node CWS network with $p=0.2$ and $k=5$. The GLT model thresholds all have $\operatorname{Beta}(2, 1)$ distribution, and the weights are sampled as in \eqref{weight_generating_interior_simplex} with $d_{\max} = 1$. With seed sets generated as in \eqref{seed_set_generator} with $s_{\max} = 10$, we sample $N=1500$ traces from this model for training and generate additional 500 traces for testing. Then, for each candidate distribution $\operatorname{Beta}(2, 1)$, $\operatorname{Beta}(3, 1)$, $\operatorname{Unif}[0, 1]$ (LT model), and $\operatorname{Exponential(1)}$ (IC model), we use the training set to estimate the weights of the GLT model under the assumption of the candidate distribution for the node thresholds. Then,  we run through each test trace $\mathcal{D}_n, \ n=1,\ldots, 500$ and do the following: 
\begin{enumerate}
    \item For each ``informative'' node $v \in V_c(\mathcal{D}_n)$, as defined in \eqref{informative_node_set}, extract the last time step it is not activated, that is, $t(v, n) = \arg\max\{t\le T_n:\ v\notin A_{t}^{(n)}\}$,
    \item Under each of the estimated GLT models, compute the probability $v$ is activated at $t(v, n) + 1$ conditional on the history of $\mathcal{D}_n$ as defined in \eqref{trans_prob_glt}. 
\end{enumerate}
We also do this with the ground-truth GLT model to obtain the true conditional activation probabilities. Finally, for each of the four candidate models, we plot the true probabilities against their predicted values along with the corresponding $95\%$ asymptotic confidence intervals computed using the Delta method, as described at the end of Section \ref{consist_section}, shown in Figure \ref{fig:node-activation}.  We also evaluate RMAE between the predicted and true probabilities, as well as the confidence interval average length and coverage, that is, the proportion of times it contains the ground truth probability.
According to RMAE and coverage metrics, the best performance is clearly obtained by using the true model, Beta$(2,1)$, while the average confidence interval length roughly equals 0.04 for all four models.
When misspecified as the "hard to influence" threshold model $\operatorname{Beta}(3, 1)$  (shown in Figure \ref{fig:community_beta}), activation probabilities close to 0 get underestimated, and those close to 1 overestimated.   With the "easy to influence" threshold models given by the uniform and the exponential, which correspond to LT and IC models, all activation probabilities tend to be underestimated. Confidence interval coverage meets the nominal target of 95\% under the true model, and is considerably lower under misspecified models.    %This is likely related to the fact that with a conservative ground-truth model, the node activation probabilities are low in most of the training traces. To fit these low probabilities well, ``conservative'' models need to significantly overestimate the edge weights, which leads to overestimated activation probabilities in the smaller number of traces where the true probabilities are high. Similarly, more ``liberal'' models need to significantly underestimate weights to fit the low probabilities well, leading to underestimation for high true probabilities. Importantly, based on the much better RMAE metric and visual dominance of the $\operatorname{Beta}(2, 1)$ model over the misspecified candidate models, we can conclude that the choice of threshold distribution for the GLT model fitting has a significant impact on the prediction quality of node activation probabilities.

\begin{figure}[tbh!]
    \centering
    \includegraphics[width=0.99\linewidth]{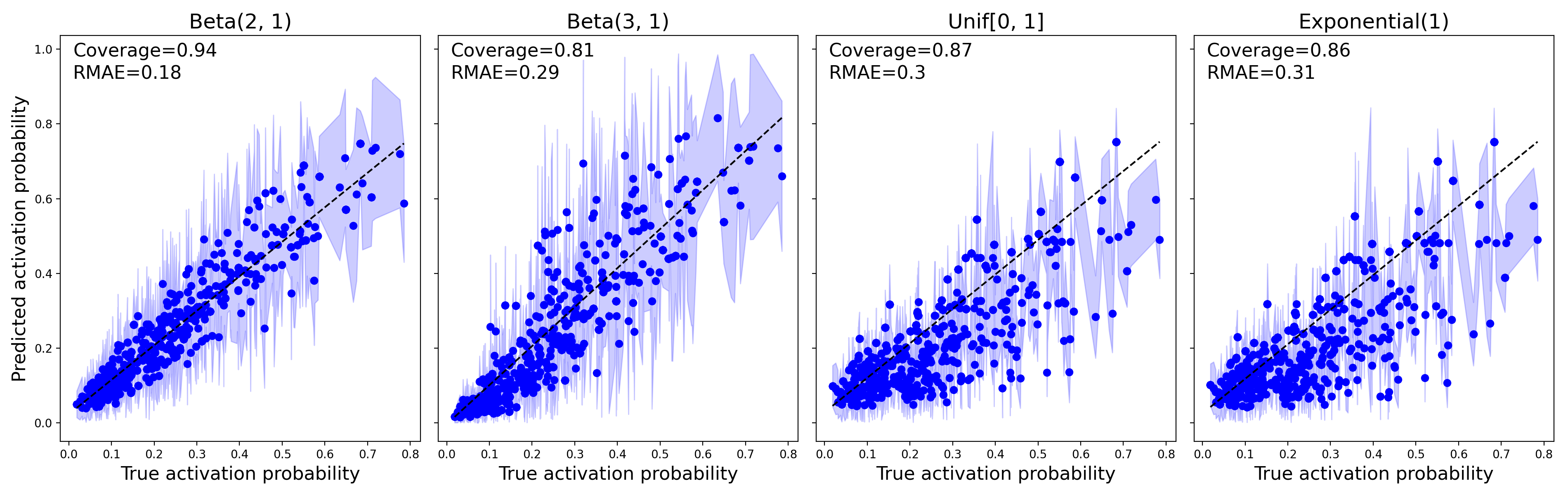} 

    \caption{Estimated node activation probabilities together with the corresponding Delta method confidence intervals computed under different GLT models with $\operatorname{Beta}(2, 1)$-GLT as the ground truth.   %Each model is fitted with 1500 traces by solving \eqref{opt_prob_changed_sum} with the corresponding threshold distribution for all nodes. Each point corresponds to the activation probability of a node within a trace at the last time point it was not activated; this is computed as in \eqref{trans_prob_glt} under the true (x-axis) and estimated (y-axis) GLT model.
    }
\label{fig:node-activation}
\end{figure}

\subsection{The GLT model in the Influence Maximization problem}\label{IM_experiments_subsection}

In this section, we demonstrate that using an appropriate diffusion model, which can be learned from trace data, can significantly improve the quality of the seed set obtained by Algorithm \ref{algorithm:greedy}.
To explore the behavior of IM solutions across different network instances, we sample 10 networks $G_\ell, \ell=1, \ldots, 10$ from the CWS model with $n=100$, $p=0.2$, and $k=10$. For each network, we generate $N=2000$ traces from the ground-truth GLT model with weights sampled according to \eqref{weight_generating_interior_simplex} with $d_{max}=1$, and the threshold distribution is set to $  F_v \sim \operatorname{Beta}(1, \beta_v)$, with  $\beta_v \sim \operatorname{Uniform}\{1, 2, 3, 4, 5\}.$

To examine how misspecification of the diffusion model impacts IM solutions, we compare the following methods: the LT model and the IC model fitted by solving \eqref{opt_prob_changed_sum}, and the GLT model with both weights and threshold distributions estimated by solving problem \eqref{individ_opt_prob_two_sets}, where the threshold distribution for each node $v$ is assumed to be $\operatorname{Beta}(1, \beta_v)$, with $\beta_v$ estimated from data.  %via a grid search over $\Phi_v = \{1, 2, \dots, 10\}$.  
As benchmarks, we also include the LT model with weights assigned via the following heuristics from  \cite{goyal2011databased}:
	\begin{itemize}
		\item Weighted Cascade ({\bf WC}): The weight of an edge $(u, v)\in E$ is estimated as the inverse of the in-degree of $v$, i.e., $\hat{b}_{u,v} = 1 / |P(v)|$. 
		\item Propagated Trace Proportion ({\bf PTP}):  The weight of edge $(u, v)$ is estimated based on the ratio between the number of traces where $u$ is activated before $v$ and the number of traces where $u$ is activated. Normalization is used to ensure that the in-degree of each node equals 1:
		$$\hat{b}_{u,v} \propto {|\{n: u \in D_{t_u}^{(n)}, v\in D_{t_v}^{(n)},\ t_u < t_v\} | \over |\{n: u \in A(\mathcal{D}_n)\}|}, \qquad \sum_{u\in P(v)} \hat{b}_{u, v} = 1.
		$$
	\end{itemize}

\begin{figure}[t!]
    \centering     \includegraphics[width=0.7\linewidth]{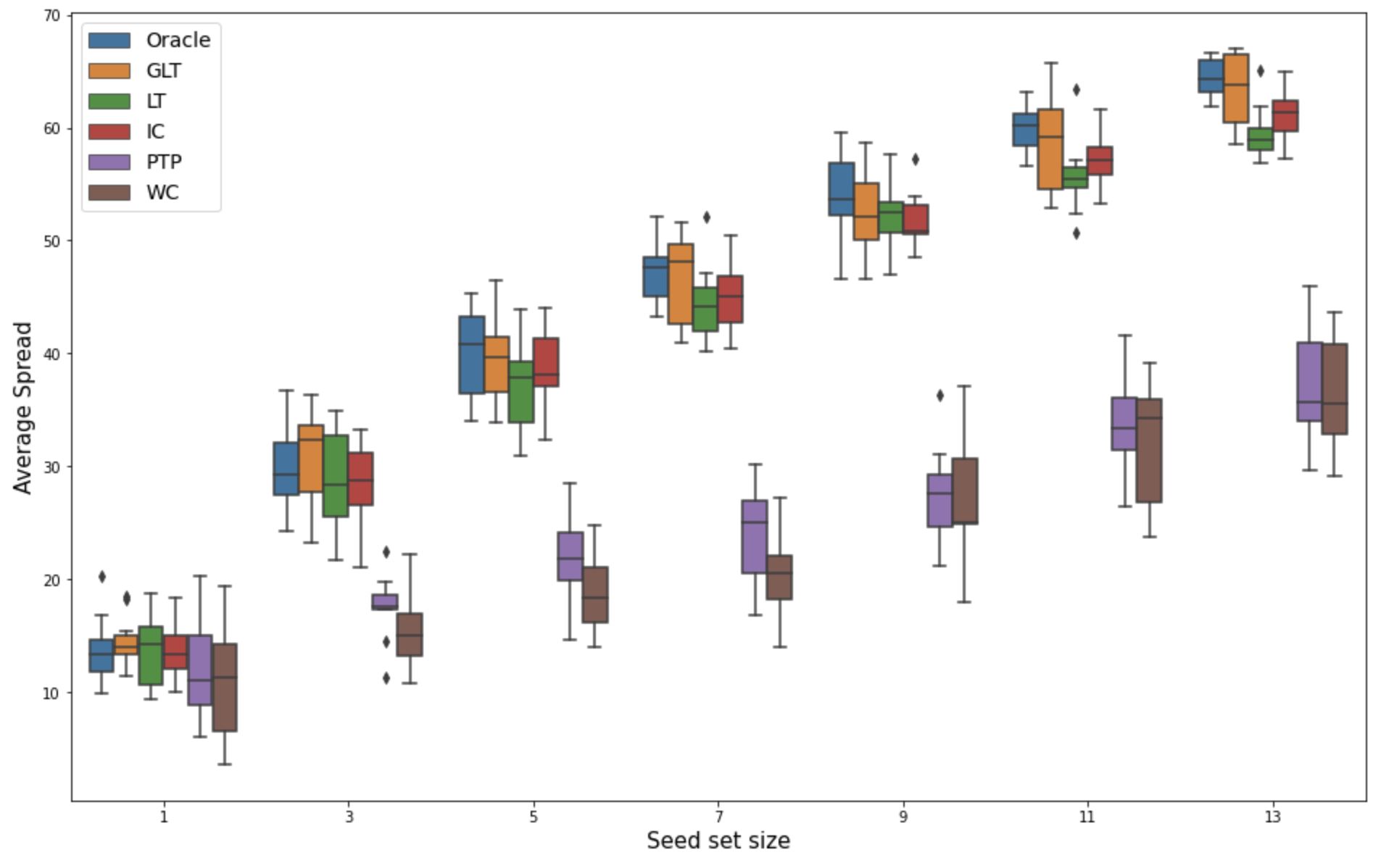}
    \caption{Comparison of the average spread across different seed set sizes, where the seeds are selected by a greedy algorithm under five candidate diffusion models and the ground truth. 
    The ground truth threshold distributions are $F_v \sim \operatorname{Beta}(1, \beta_v)$, where $\beta_v \sim \operatorname{Unif}\{1, 2, 3, 4, 5\}$. 
    Each box plot represents the estimated spread across 10 networks, where the spread is averaged over five fitted models (each trained on a separate set of 2000 traces). All values of the spread function $\sigma(\cdot)$ are estimated using $1000$ Monte Carlo simulations.}
    \label{fig:im_problem_res}
\end{figure}

\noindent
For each of the diffusion models described above and each generated network, we follow these steps.  First, we estimate the model weights and, in the case of GLT, also the $\beta_v$ values via a grid search over $\{1, 2, \dots, 10\}$. Then we 
\begin{enumerate}
%\item For each network $G_\ell$, we estimate the model weights (and in the case of GLT, the $\beta_v$ values as well) using the corresponding estimation procedure. The estimated diffusion model is denoted $\hat{M}_\ell$, $\ell=1, \ldots, 10$.
%diffusion model $\hat{M}_{r, i}$ for each network $G_r$ and every train set $\mathbb{D}^{(r)}_i$ with the corresponding parameter estimation procedure.
\item  Run Algorithm \ref{algorithm:greedy} with seed set size $k = \{1, 4, 7, 10, 13\}$ under the fitted diffusion model to obtain a seed set $\hat{S}_{k}$. The influence function $\sigma$ in Algorithm \ref{algorithm:greedy} is approximated using $1000$ Monte Carlo simulations.
%Run Algorithm \ref{algorithm:greedy} with $k=1, 4, 7, 10, 13$ on each network $G_r$ under the estimated GLT model $\hat{M}_{r, i}$ with parameters $\hat{\boldsymbol{\theta}}^{(r, i)}$ to obtain the seed set $\hat{S}_{k}^{(r, i)}$. For each $(k, r, i)$ triple, we note $\sigma_{G_r, \hat{M}_{r, i}}(\hat{S}_{k}^{(r, i)})$. All values of $\sigma(\cdot)$ are approximated using $1000$ MC simulations.
\item   Run the diffusion 1000 times from the seed set obtained in Step 1 and obtain the average spread $\hat{\sigma}_{k}$.

%For each network $G_r$ and every seed size $k$, compute the average spread over 5 trace sets 
%$$\hat{\sigma}_{k, r} = {1\over 5} \sum_{i=1}^{5} \sigma_{G_r, \hat{M}_{r, i}}(\hat{S}_{k}^{(r, i)})$$
\end{enumerate}
As a benchmark, we also compute the average spread under the ground truth GLT model, referred to as the {\it oracle}.

Figure \ref{fig:im_problem_res} presents boxplots (across the 10 networks) of $\hat{\sigma}_{k}$ for different values of $k$, ranging from 1 to 13.  As the seed set size $k$ increases, the choice of model has a greater impact, likely because for very small seed set size the greedy algorithm tends to select the most connected nodes under any model.  However, significant differences in spread emerge as the seed set size increases.  As expected, LT and IC are inferior to the ground truth and the GLT model, and the heuristics work poorly.   
%despite estimating the weights from the observed propagation traces. 
The GLT model with estimated threshold distributions achieves performance comparable to that of the oracle.

\subsection{Spread estimation}
In the previous experiment, we only evaluated the spreads from seed sets that were selected by the IM algorithm.   In some applications, we may also be interested in assessing spread, for example, of fake news or a virus, initiated from a given seed set that has not been optimally selected. Here we show that, for propagation under the GLT model, selecting an accurate threshold distribution can significantly improve the accuracy of spread estimation from any seed set.

As in Section \ref{node_activation_experiment_section}, we consider a CWS network with $n=100$, $p=0.2$, and $k=10$, and the associated GLT model with weights sampled according to \eqref{weight_generating_interior_simplex}. This time, we set the ground-truth threshold distributions as $F_v\sim\operatorname{Beta}(2, 2)$ for all child nodes $v\in V_c$. With seed sets generated as in \eqref{seed_set_generator} with $s_{\max}=20$, we generate a set of 1000 train traces and  additional 500 test seed sets. We pick a more "easy-to-influence" model and higher than usual $s_{\max}$ to make the observed traces spread farther, and so allow studying the effect of model misspecification on the full range of trace lengths.
%(denoted as $\mathbb{D}$), sort them by the final number of activated nodes $|A(\mathcal{D})|$,  and split to $\mathbb{D}_{train}$  and $\mathbb{D}_{test}$ by putting every fifth trace into the test set, so that $|\mathbb{D}_{train}| = 2000$ and $|\mathbb{D}_{test}| = 500$. 
Then, for each candidate distribution $\operatorname{Beta}(2, 2)$, $\operatorname{Beta}(1, 2)$, $\operatorname{Beta}(2, 1)$, and $\operatorname{Unif}[0, 1]$ (LT model), we use the training set to estimate the weights of the GLT model under the assumption of the candidate distribution for the node thresholds. We then compute the predicted spread by running the estimated GLT models 1000 times from each test seed set. We also do this with the ground-truth GLT model to obtain the true spreads. In Figure \ref{fig:spread_compar}, we plot the estimated spreads against the ground truth for each candidate GLT model and report the RMAE between them. Similarly to the node activation probabilities, the spread from a given seed set is also a function of the GLT model weights; however, it is a complicated implicit function that does not lend itself to a delta method calculation, so instead we estimated the spread empirically, by repeating the simulation 10 times.  The variability was negligible on the scale of the plots, and is thus not shown; this is expected since spread estimators are known to be robust to small perturbations in diffusion model weights (see, for example, \cite{goyal2011databased}). 

\begin{figure}[h!]
    \centering    \includegraphics[width=0.99\linewidth]{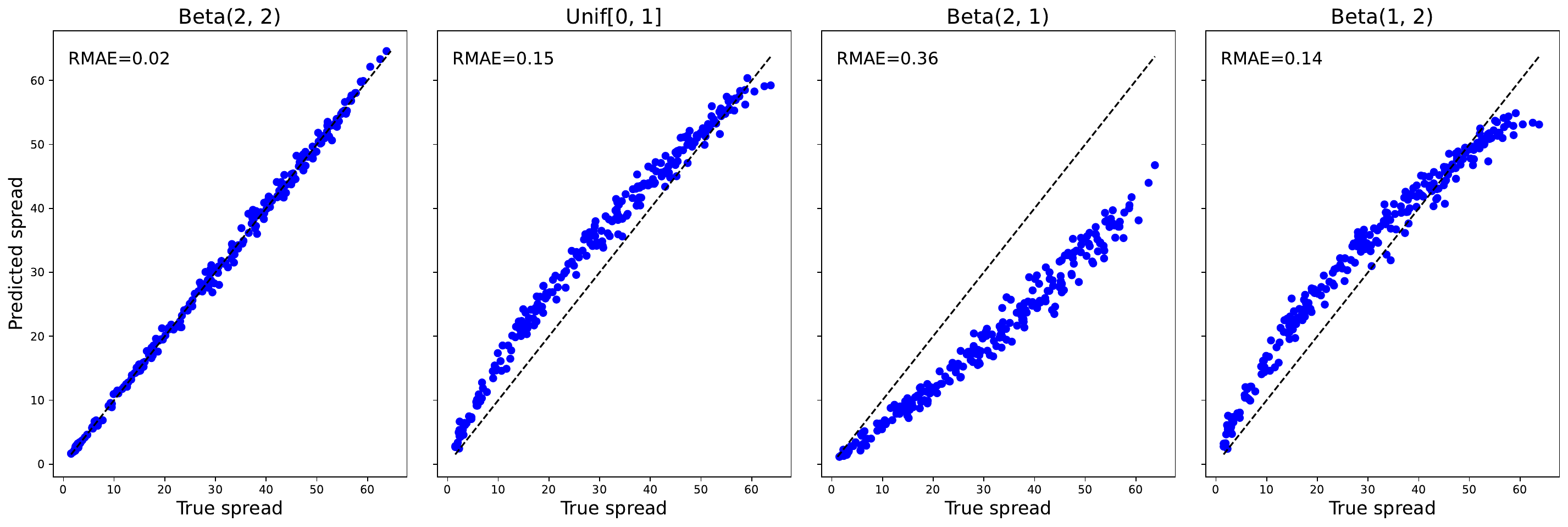}
    \caption{Comparison between the ground-truth spreads from 500 test trace sets and the estimated spreads for several candidate GLT models with a fixed $\operatorname{Beta}$ threshold distribution for all nodes.  The trace-generating model is the GLT with $F_v\sim \operatorname{Beta(2, 2)}$ for all nodes. Weights for each candidate GLT model are estimated using 1000 training traces.}
    \label{fig:spread_compar}
\end{figure}
In Figure \ref{fig:spread_compar}, we observe that the ground-truth model clearly outperforms the other candidates, even the LT model with the same mean of $1/2$, suggesting that getting the shape of the threshold distribution right is important. As a more general conclusion, we can notice that the more "easy-to-influence" models, such as $\operatorname{Beta}(1, 2)$, tend to overestimate the spread and the more "hard-to-influence" models, such as $\operatorname{Beta}(2, 1)$, tend to underestimate it.
% Interestingly, between the $\operatorname{Beta}(2, 2)$ and $\operatorname{Unif}[0, 1] = \operatorname{Beta}(1, 1)$ models, the unit step in $\beta$ from the ground truth leads to a considerably higher error than the unit step in $\alpha$. Although the two candidate distributions have equal expectations at 1/2, we explain this by the fact that similarity to ground truth is more important in low-threshold regions near zero than in high-threshold regions near one, since the nodes with large threshold values are hard to activate and an inaccurate estimation of their activation does not significantly impact the spread prediction.
% To conclude, we can confidently state that the misspecification of the GLT model thresholds significantly impacts the accuracy of the spread estimation.

\subsection{The movie ratings example}
In this section, we apply the proposed weight estimation procedures to the Flixster dataset, collected from www.flixster.com, a popular social media platform for movie ratings. The dataset contains an undirected, unweighted social network of approximately 1 million users and over 8 million time-stamped ratings of movies by the users, in the time period from 2005 to 2009. Following  \cite{goyal2011databased}, we represent these ratings as {\it action logs}, which is a collection of triples $(u, a, t)$, where $u$ represents the user ID, $a$ the movie ID, and $t$ the time when the user rated the movie; the value of the rating is disregarded.     \cite{goyal2011databased} analyzed this dataset under the assumption that if user $u$ rated movie $a$ before user $v$, and $u$ and $v$ are connected, then $v$ was influenced by $u$.  There is of course the possibility that $v$ rated the movie independently of $u$, but they showed that even without accounting for this possibility, learning from the action logs can significantly improve the estimation of influence propagation. Here, our goal is to compare the accuracy of different diffusion models for predicting node activation events.

%Unfortunately, the original Flixster data is clearly too large to work with. 
To process the data, we first removed all users who rated fewer than 20 movies, 
%as it is difficult to estimate the influence of their social network neighbors on them with limited activity.  Next, we 
then applied the algorithm \citep{core_extraction} to extract the core sub-graph of the remaining users, and finally extracted the largest connected component of the core. This resulted in a network of 8174 nodes, approximately 50K undirected edges (which we doubled to create directed edges), and approximately 2.1M action logs.  

The next step is to transform the action logs into trace or pseudo-trace data. 
Inferring full propagation traces from action logs is challenging, even in simple scenarios. For example, consider a graph of three connected users who rated the same movie at distinct times $t_1 < t_2 < t_3$. There are already five possible ways to construct the corresponding trace -- $(\{1, 2,3\}), (\{1\}, \{2, 3\}), (\{1, 2\}, \{3\}),(\{1\}, \{2, 3\})$, or $(\{1\}, \{2\}, \{3\})$ -- because of the inherent ambiguity of converting the continuous time stamps into discrete propagation events.    While one could pick a threshold to decide whether two time stamps should be considered the same or not in the discrete time space, it would be arbitrary, and the resulting analysis can be sensitive to this arbitrary choice.  Instead, we use the pseudo-trace framework described in Section \ref{pseudo_trace_section}. 

\begin{figure}[h!]
    \centering
\includegraphics[width=0.8\linewidth]{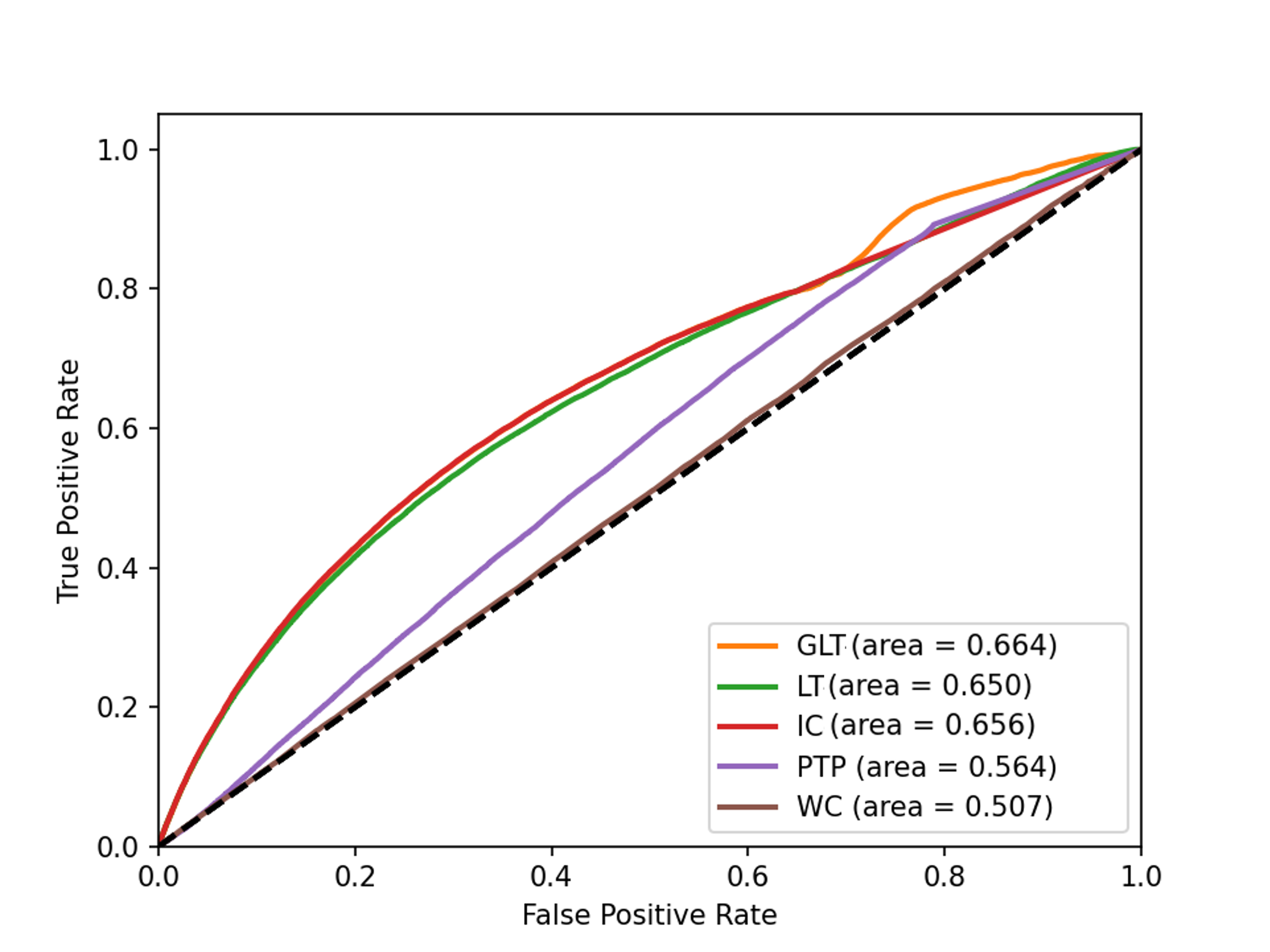}
    \caption{ROC curves and AUC scores for the activation probabilities computed on the test pseudo-trace set, from candidate diffusion models estimated on the training pseudo-trace set.  }
    \label{fig:real_data_prob_predict_compar}
\end{figure}

For each user (node) $v$, we extract the pseudo-traces where $v$ was activated by identifying all movies $a$ they rated and noting the set $A_v^{(a)} \subset P(v)$ of $v$'s parents who rated $a$ before $v$ did. For pseudo-traces where $v$ was not activated, we consider all movies $a$ that $v$ did not rate but at least one of their parents did, noting the set $A_v^{(a)}$ of all its active parents at the last recorded time point in the action log.  We then randomly split the trace data into training (80\%) and test (20\%) sets, stratified by the activation status of the node. 
Using the same candidate diffusion models as in Section \ref{IM_experiments_subsection}, we fit the models on the training pseudo-traces and compute predicted node activation probabilities on the test pseudo-traces. The only difference from the settings in Section \ref{IM_experiments_subsection} is a larger parameter grid for estimation of the Beta distribution parameters, where we allows both $\alpha_v$ and$\beta_v$ to range from 1 to 10.  %$ for the Beta threshold distribution:  $$F_v \sim \operatorname{Beta}(\alpha_v, \beta_v), \ \text{ where }\ (\alpha_v, \beta_v) \in \{1, 2, \ldots, 10\}^2. $$
Note that for the fitting problem \eqref{individ_opt_prob_two_sets}, we only require $F_v$ to be log-concave, which is satisfied when both $\alpha \ge 1$ and $\beta \ge 1$ for the Beta distribution; it is only for solving the IM problem that $F_v$ is required to be concave.  

Figure \ref{fig:real_data_prob_predict_compar} presents the resulting ROC curves and AUC scores for the estimated diffusion models.   These results are for one random split into training and test;  variability in AUC scores under different random splits was less than 0.001.  The results show that GLT performs the best, followed fairly closely by IC and LT, and these three models significantly  outperform the heuristic PTP and WC.   While this dataset does not represent an ideal test case for our model, since the traces are not observed directly, it demonstrates that flexibility in modeling the thresholds can help even when neither the data collection mechanism nor, presumably, the true propagation model exactly match the GLT framework.

\section{Discussion}
	
	\label{ch:conclusion}
	
In this paper, we have proposed a new flexible framework for information diffusion on networks, the general linear threshold (GLT) model.  We derived identifiability conditions for edge weights which are weaker than previously available, developed a statistically principled likelihood-based method to estimate the edge weights from fully or partially observed traces, and proved that these estimates are $\sqrt{n}$-consistent and asymptotically normal when threshold distributions are known. %In \eqref{individ_opt_prob_two_sets}, we also proposed to use the grid search to estimate the parameters of threshold distributions. 
We also proposed a parametric approach to estimating threshold distributions, which saves computational time but is relatively inflexible compared to nonparametric distribution estimators;  we leave that for future work.   In the parametric setting, it would also be of interest to establish identifiability and consistency conditions for the GLT model where both the parent edge weights $\boldsymbol{\theta}_v$ and the threshold distribution parameters $\boldsymbol{\varphi}_v$ vary with $v$.  

In Section \ref{IM_section}, we illustrated the application of the GLT to the IM problem, and established the relationship between the quality of the IM solution under the GLT model errors for the case of bipartite graphs;  establishing this relationship for more general graph classes is a topic for future work.  Another  important question to study is stability of the IM solution to misspecification of node threshold distributions.  

To allow for easy use with arbitrary distribution of thresholds, our implementation of the likelihood optimization problem \eqref{opt_prob_changed_sum} uses the SLSQP solver, which accommodates both convex and non-convex problems. This choice was made to allow for fitting the GLT model with non-log-concave threshold densities, which may violate the convexity condition in Proposition \ref{logconc} but be desirable in practice. When convexity is guaranteed, however, optimization efficiency can be significantly improved by using convex solvers, such as Gurobi or Mosek. Improving optimization speed and quality is another goal for future work.

% \printbibliography
\bibliographystyle{chicago}
\bibliography{references}

\appendix
 \section{Appendix}
	
	\label{ch:appendix}
	 % \renewcommand{\thesection}{\Alph{section}}
% \subsection{Proofs}

% \subsection{Proof of Proposition \ref{prop_model_realtionship}}

%   \begin{enumerate}
%      \item See Theorem 2.5 in \cite{Kempe}. 
%      \item It is enough to check that the following threshold function
%      $$f_v(S):=\mathbb{P}(S\cap \Gamma_v \ne \emptyset) = \sum_{ A \subseteq S} \mathbb{P}(\Gamma_v = A)$$ 
%      is monotone and submodular. Monotonicity is obvious,  and submodularity can be verified by noticing that by adding a node $u$ to a set $S$ we increase the influence function by a quantity depending only on the node $v$ itself: 
%      $$f_v(S\cup \{u\}) - f_v(S)  = \sum_{A\supset \{u\}} \mathbb{P}(\Gamma_v = A) 
%      $$
%      \liza{First, $\sum_{A\supset \{u\}} \mathbb{P}(\Gamma_v = A)$ seems to depend on both $u$ (through the summation index) and $v$.   Second, doesn't showing submodularity just require this difference to be nonnegative?   Why does the dependency on $v$ only matter then?}
%      \item    If we set $f_v:=F_v$ and apply the inverse cdf to both sides of the activation rule for the GT model, we obtain
%  $$f_v(S) := F_v\left(\sum_{u\in S} b_{u,v}\right) \ge U_v \sim \operatorname{Unif}[0, 1] \quad \Leftrightarrow \quad \sum_{u\in S}  b_{u,v} \ge F^{-1}_v (U_v) \sim F_v.$$
%      \end{enumerate}

\subsection{Section \ref{gltm_introduce_section} Proofs}\label{gltm_introduce_section_proofs}
\begin{proof}[Proof of Proposition \ref{glt_vs_trig}]
Consider a GLT model on the star graph of in-degree 3 in Figure \ref{fig:instar_graph} with equal weights $b_{u,v} = 1/3$. Let node $4$ have a threshold cdf $F$ satisfying
$$F(0) = 0, \ F(1/3) = 0.5, \ F(2/3) = 0.85, \ F(1) = 1.$$ % An example of such $F$ can be constructed using Lagrange interpolation:
%$$F(x) = \begin{cases}
%    0, & x < 0\\
%    -0.225x^3 - 0.45 x^2 + 1.675x, & x \in [0, 1]\\
%    1, & x > 1
%\end{cases}
%$$
% Note that due to Theorem \ref{GLT_submodular_theorem}, the influence function of this GLT model is submodular since $F$ is concave on $[0, 1]$:
% $$F''(x) = -1.35x - 0.9 < 0, \quad x \in [0, 1]
% $$
Threshold distributions of nodes 1, 2, and 3 do not affect the diffusion model as other nodes cannot activate them.  Any triggering model on this graph will have only 8 relevant parameters,  representing probabilities of each possible subset of $\{1,2,3\}$ to be the triggering set for node 4. We will denote this probability distribution as
$$\mathcal{P} = \{P_\emptyset, P_1, P_2, P_3, P_{12}, P_{13}, P_{23}, P_{123}\}.$$
If the GLT model was a special case of the Triggering model, a distribution  $\mathcal{P}$  would exist such  that the activation probability of node 4 is the same for the two models given any seed set, i.e., the following linear system should have a solution:
$$
\begin{cases}
    P_1 + P_{12} + P_{13} + P_{123} &=  F(b_{1, 4}) = 0.5 \\
    P_2 + P_{12} + P_{23} + P_{123} &=  F(b_{2, 4}) = 0.5 \\
    P_3 + P_{13} + P_{23} + P_{123} &=  F(b_{3, 4}) =0.5\\
    P_1 + P_2 + P_{12} + P_{13} + P_{23} + P_{123} &=  F(b_{1, 3} + b_{2,4}) = 0.85\\
    P_1 + P_3 + P_{12} + P_{13} + P_{23} + P_{123} &=  F(b_{1, 3} + b_{3,4}) = 0.85\\
    P_2 + P_3 + P_{12} + P_{13} + P_{23} + P_{123} &=  F(b_{1, 3} + b_{2,4}) = 0.85\\
    P_1 + P_2 + P_3 + P_{12} + P_{13} + P_{23} + P_{123} &=  F(b_{1, 4} + b_{2,4} + b_{3, 4}) = 1
\end{cases}
$$
Solving this system, we can verify that a solution exists, but it has $P_{123} = - 0.05$, and thus is not a valid probability distribution. 
\end{proof}

%$$\begin{pmatrix}
%    P_1\\ P_2 \\P_3\\ P_{12}\\ P_{13}\\ P_{23}\\ P_{123}
%\end{pmatrix} = 
%\left(\begin{array}{ccccccc}
%0 & 0 & 0 & 0 & 0 & -1 & 1 \\
%0 & 0 & 0 & 0 & -1 & 0 & 1 \\
%0 & 0 & 0 & -1 & 0 & 0 & 1 \\
%0 & 0 & -1 & 0 & 1 & 1 & -1 \\
%0 & -1 & 0 & 1 & 0 & 1 & -1 \\
%-1 & 0 & 0 & 1 & 1 & 0 & -1 \\
%1 & 1 & 1 & -1 & -1 & -1 & 1
%\end{array}\right) 
%\left(\begin{array}{l}
%0.5 \\
%0.5 \\
%0.5 \\
%0.85 \\
%0.85 \\
%0.85 \\
%1
%\end{array}\right)
%$$

\subsection{Section \ref{identif_section} Proofs}\label{identif_section_proofs}

\begin{proof}[Proof of Lemma \ref{pos_trace_prob_lemma}]
    Note that the trace probability in \eqref{individ_trace_prob} consists of at most $|V_c|$ terms (excluding the positive seed set probability), as only child nodes can be activated at $t \ge 1$ and each node in a trace can be activated at most once. In turn, for each node $v \in V_c$, the corresponding term is either $F_v(B_v(A_t; \boldsymbol{\theta}_v)) - F(B_v(A_{t-1}; \boldsymbol{\theta}_v))$ or $1- F_v(B_v(A_{T}; \boldsymbol{\theta}_v)$. By trace feasibility,
    $D_{t} \cap P(v)$ and $A_T \cap P(v)$ are non-empty. Coupled with the strict monotonicity of $F_v$,  it implies that for any $\boldsymbol{\theta}_v \in \tilde{\Theta}_v$ that
    $$F_v\left[B_v(A_{t-1}; \boldsymbol{\theta}_v)\right] < F_v\left[B_v(A_{t}; \boldsymbol{\theta}_v)\right] \quad \text{and} \quad F\left[B_v(A_T; \boldsymbol{\theta}_v)\right] \le F_v\left[B_v(P(v); \boldsymbol{\theta}_v)\right] < 1.$$
   Thus, both types of terms are positive, which completes the proof.
\end{proof}

 Before we proceed to prove Proposition \ref{all_reachable_nodes_propos}, we establish a useful lemma, showing the equivalence between the node reachability and its appearance in a trace.

\begin{lemma}\label{reachability_equiv_lemma}
     The following conditions on node $u\in V$ are equivalent:
    \begin{enumerate}   
        \item[(a)] For all $\boldsymbol{\theta} \in \Tilde{\Theta}$, it holds $\mathbb{P}_{\boldsymbol{\theta}}(u\in A(\mathcal{D})) > 0$; 
        \item[(b)] There exists $\boldsymbol{\theta} \in \Tilde{\Theta}$ such that $\mathbb{P}_{\boldsymbol{\theta}}(u\in A(\mathcal{D})) > 0$;  
        \item[(c)] Node $u$ is reachable.
    \end{enumerate}    
\end{lemma}

\begin{proof}[Proof of Lemma \ref{reachability_equiv_lemma}]

    Statement (a) trivially implies (b).  To show (b) implies (c), note that (b) implies that there is a feasible trace $\mathcal{D}=(D_0, \ldots, D_t, \ldots, D_T)$ with $\mathbb{P}^0(D_0) > 0$ and $u\in D_t$.   We prove by induction over $\tau \ge 1$ that there is a path to $u$ from a node $w_{t-\tau} \in D_{t-\tau}$.    First,  for $\tau = 1$, since $\mathcal{D}$ is feasible, $u$ should have a parent $w_{t-1}$ in $D_{t-1}$.   Now suppose the induction hypothesis holds for $\tau$.  If there is a path to $u$ from  $w_{t-\tau}\in D_{t-\tau}$, there is also a path from  $D_{t-\tau-1}$ since $w_{t-\tau}$ should have a parent in $D_{t-\tau-1}$ by feasibility.
     This results in a path $(w_0, \ldots, w_{t-1}, u)$ connecting $D_0$ and $u$ which implies that $u\in R$.
     
     To prove (c) implies (a), consider an arbitrary $u\in R$. If $\mathbb{P}^0(u\in D_0) > 0$, then (a) holds.  Assume now there is a sequence of nodes $(w_0, w_1, \ldots, w_T=u)$ such that $(w_{t-1}, w_{t}) \in E$ for all $t=1, \ldots, T$ and $w_0\in D_0$ with $\mathbb{P}^0(D_0) > 0$. Then,  $\mathcal{D} = (D_0, \{w_1\}, \ldots, \{w_T\})$ is a feasible trace that has a positive probability for any $\boldsymbol{\theta}\in \Tilde{\Theta}$ according to Lemma \ref{pos_trace_prob_lemma}.
\end{proof}

\begin{proof}[Proof of Proposition \ref{all_reachable_nodes_propos}]
Suppose there is an unreachable non-isolated node $u \in V$. Then, by Lemma \ref{reachability_equiv_lemma}, there is $\boldsymbol{\theta}\in \Tilde{\Theta}$ such that $\mathbb{P}_{\boldsymbol{\theta}}(u \in A(\mathcal{D})) = 0$. If $u$ has a child $v$, the edge $(u,v)$ will never participate in a trace and thus changing $b_{u,v}$, while keeping all other weights fixed, preserves $\mathbb{P}_{\boldsymbol{\theta}}$. If $u$ has a parent $v$, then $v$ should also be unreachable, as otherwise, there would be a directed path from some positive probability seed set to $u$ passing through $v$. But then changing $b_{v,u}$ alone again does not change $\mathbb{P}_{\boldsymbol{\theta}}$. In both cases, we get a contradiction with identifiability.
\end{proof}

We now conclude this section with the proof of the identifiability theorem.
\begin{proof}[Proof of Theorem \ref{identif_theorem}]
     (Sufficiency)  Suppose there are distinct vectors of parameters $\boldsymbol{\theta}=\{b_{u,v}: (u,v)\in E\}$ and $\Tilde{\boldsymbol{\theta}}=\{\Tilde{b}_{u,v}:  (u,v)\in E\}$ in $\Tilde{\Theta}$ for which $\mathbb{P}_{\boldsymbol{\theta}} = \mathbb{P}_{\Tilde{\boldsymbol{\theta}}}$. Since $\boldsymbol{\theta} \ne \Tilde{\boldsymbol{\theta}}$, there is an edge $(u,v)\in E$ such that $b_{u,v} \ne \Tilde{b}_{u,v}$. 
    Consider the subsets $S_j, j=1,\ldots, m$ of $P(v)$ together with the corresponding traces $\mathcal{D}_j=(D_0^{(j)}, \ldots, D_{t_j}^{(j)})$ satisfying conditions of the theorem. For each $\mathcal{D}_j$, let $t_{jk}, k=1, \ldots, r_j$ with $t_{jr_j} = t_j$ denote all time points before $t_j$ when $D_{t_{jk}} \cap P(v) \ne \emptyset$.  As the equality of distributions $\mathbb{P}_{\boldsymbol{\theta}}$ and $\mathbb{P}_{\Tilde{\boldsymbol{\theta}}}$ implies equality of the corresponding seeded diffusion model instances, we have for any time point $t_{jk}$:
    $$\mathbb{P}_{\boldsymbol{\theta}}(v \notin D_{t_{jk} + 1}|D_0, \ldots, D_{t_{jk}}) = \mathbb{P}_{\Tilde{\boldsymbol{\theta}}}(v \notin D_{t_{jk} + 1}|D_0, \ldots, D_{t_{jk}}),$$
    which implies by \eqref{trans_prob_glt} that 
    \begin{equation}\label{eq_marginal_ps}
        {1 - F_v\left(B_v(A_{t_{jk}}^{(j)}; \boldsymbol{\theta}_v\right) \over 1 - F_v\left(B_v(A_{t_{jk} - 1}^{(j)}; \boldsymbol{\theta}_v))\right)}=  {1 - F_v\left(B_v(A_{t_{jk}}^{(j)}; \tilde{\boldsymbol{\theta}}_v)\right) \over 1 - F_v\left(B_v(A_{t_{jk} - 1}^{(j)}; \tilde{\boldsymbol{\theta}}_v))\right)}.
     \end{equation}
    Taking the product of these equalities across $t_{jk}$ for $k \le r_j$ and $k < r_j$, we obtain by telescoping 
    $$
       F_v\left(B_v(A_{t_{j}}^{(j)}; \boldsymbol{\theta}_v) \right)=  F_v\left(B_v(A_{t_{j}}^{(j)}; \tilde{\boldsymbol{\theta}}_v)\right) \quad \text{and} \quad  F_v\left(B_v(A_{t_{j} - 1}^{(j)}; \boldsymbol{\theta}_v) \right)=  F_v\left(B_v(A_{t_{j} - 1}^{(j)}; \tilde{\boldsymbol{\theta}}_v)\right),
    $$
    where for the second equality we used $A_{t_{j, r_j -1}}^{(j)} \cap P(v) = A_{t_j - 1}^{(j)} \cap P(v)$.
    By the monotonicity of $F_v$, we can apply its inverse to both sides of the equations above to deduce 
    $$ B_v(A_{t_{j}}^{(j)}; \boldsymbol{\theta}_v )=  B_v(A_{t_{j}}^{(j)}; \tilde{\boldsymbol{\theta}}_v) \quad \text{and} \quad  B_v(A_{t_{j} - 1}^{(j)}; \boldsymbol{\theta}_v)=  B_v(A_{t_{j} - 1}^{(j)}; \tilde{\boldsymbol{\theta}}_v).
    $$
   By definition of $\mathcal{D}_j$, we have $D_{t_{j}}^{(j)} \cap P(v) = S_j$, so, by subtracting the above equations we obtain  $B_v(S_j; \boldsymbol{\theta}_v)=  B_v(S_j; \tilde{\boldsymbol{\theta}}_v)$.  Aggregating the resulting equalities across $j=1, \ldots, m$, we obtain
    $\tilde{X}_v^\top\boldsymbol{\theta}_{v} = \tilde{X}_v^\top\Tilde{\boldsymbol{\theta}}_{v}.$
 But according to our assumption, $\tilde{X}_v$ is invertible, thus $\boldsymbol{\theta}_{v} = \Tilde{\boldsymbol{\theta}}_{v}$  and $b_{u,v}=\Tilde{b}_{u,v}$ in particular. Contradiction. 
    \\
    \\
    (Necessity) Suppose $\{\mathbb{P}_{\boldsymbol{\theta}},\boldsymbol{\theta}\in \Tilde{\Theta}\}$ is identifiable but there is $v\in V_c$ with $P(v) = \{u_1, \ldots, u_m\}$ for which conditions of the theorem do not hold.
    % Note that if $m=1$, we have $\tilde{X}_v = 1$ which is always invertible, therefore $m$ can be only $\ge 2$. 
    Take an arbitrary $\boldsymbol{\theta}_v$ in the interior of $ \Tilde{\Theta}_v$, which is non-empty by definition of $\varepsilon_0$.
    Our further goal is to obtain a contradiction by constructing $\Tilde{\boldsymbol{\theta}} \in \Tilde{\Theta}$ coinciding with $\boldsymbol{\theta}$ everywhere except for $\boldsymbol{\theta}_v$ so that $\mathbb{P}_{\Tilde{\boldsymbol{\theta}}} = \mathbb{P}_{\boldsymbol{\theta}}$. Consider all possible subsets $S_j \subset P(v), j=1, \ldots, k$  which satisfy condition 1 of the theorem. Then, our assumption implies that the matrix $\tilde{X}_v = [\mathbf{1}(u_i \in S_j)] \in \{0, 1\}^{m \times k}$ has $\operatorname{rank}(\tilde{X}_v) < m$. In other words, there is a non-zero vector $z\in \mathbb{R}^m$ such that $\tilde{X}_v^\top z = \boldsymbol{0}_k$. As $\boldsymbol{\theta}_v$ lies in the interior of $\tilde{\Theta}_v$ by definition, we can find a sufficiently small scalar $\delta > 0$, such that $\Tilde{\boldsymbol{\theta}}_{v} = \boldsymbol{\theta}_{v} + \delta z $ still lies in $\Tilde{\Theta}_v$ while preserving \begin{equation}\label{equal_theta_tilde}
       \tilde{X}_v^\top\boldsymbol{\theta}_{v} = \tilde{X}_v^\top\Tilde{\boldsymbol{\theta}}_{v}.
    \end{equation}
    Note that by changing $\boldsymbol{\theta}_v$ alone, we could only change the probability of a feasible trace $\mathcal{D}=(D_0,\ldots, D_T)$ with $\mathbb{P}^0 (D_0) > 0$ and $P(v) \cap D_t \ne \emptyset$ for some $t\le t_v:=\arg\max_{\tau \le T}\{ \tau: v\notin A_\tau\}$ where $t_v$ is the last time $v$ is not activated. Indeed, if $\mathbb{P}^0 (D_0) = 0$ then $\mathbb{P}_{\boldsymbol{\theta}}(\mathcal{D}) = \mathbb{P}_{\Tilde{\boldsymbol{\theta}}}(\mathcal{D}) = 0$ and if $P(v) \cap D_t = \emptyset$ for any  $t\le t_v$, then by \eqref{individ_trace_prob}, trace probability does not depend on $\boldsymbol{\theta}_v$.
    Take such a trace and consider all times $s_j \le t_v, j=1\ldots,r$  for which $P(v) \cap D_{s_j}\ne \emptyset$. 
    Note that for each time $s_j$, the trace $(D_0, \ldots, D_{s_j})$ is also feasible and satisfies condition 1 of the theorem. Therefore, there is a corresponding column $x_j:= [\mathbf{1}(u_i \in D_{s_j})]_{i=1}^m$ in matrix $\tilde{X}_v$ which according to \eqref{equal_theta_tilde} satisfies
    $$B_v(D_{s_j}; \boldsymbol{\theta}_v) = \langle x_j, \boldsymbol{\theta}_v\rangle = \langle x_j,\Tilde{\boldsymbol{\theta}}_v \rangle = B_v(D_{s_j}; \tilde{\boldsymbol{\theta}}_v).
    $$
    Summing these equations over $j \le r$ and $j<r$, we obtain, respectively,
    $$B_v(A_{t_v}; \boldsymbol{\theta}_v) =B_v(A_{t_v}; \tilde{\boldsymbol{\theta}}_v)\ \quad \text {and} \quad B_v(A_{t_v - 1}; \boldsymbol{\theta}_v)=  B_v(A_{t_v-1}; \tilde{\boldsymbol{\theta}}_v).
    $$
    But from \eqref{individ_trace_prob}, trace probability is either a function of $B_v(A_{T}; \boldsymbol{\theta}_v)=B_v(A_{t_v}; \boldsymbol{\theta}_v)$ if $v\notin A_T$ or of  
    $B_v(A_{t_v}; \boldsymbol{\theta}_v)$ and $B_v(A_{t_v - 1}; \boldsymbol{\theta}_v)$ if $v\in A_T$. Therefore, 
    we deduce $\mathbb{P}_{\boldsymbol{\theta}}(\mathcal{D}) = \mathbb{P}_{\Tilde{\boldsymbol{\theta}}}(\mathcal{D})$ and arrive at a contradiction with identifiability of $\Tilde{\Theta}$.
\end{proof}

\subsection{Section \ref{weight_estim_section} Proofs}\label{weight_estim_section_proofs}
\begin{proof}[Proof of Proposition \ref{logconc}]
    Let $F$ be an arbitrary cumulative distribution function with density $f$. We need to show that \begin{equation}\label{integral_form}
        F(x) - F(y) = \int_\mathbb{R} \mathbf{1}(y < t \le x) f(t) dt
    \end{equation}
    is concave on $\{(x, y): x > y\}$.
    Note that $g(x,y,u) := \mathbf{1}(y < t \le x)$ is a log-concave function since for any $\lambda \in [0, 1]$ and points $A_1 = (x_1, y_1, u_1), A_2 = (x_2, y_2, u_2)$ with $x_1 > y_1, x_2 > y_2$, it holds:
\begin{align*}
    g(\lambda A_1 + (1-\lambda) A_2) &= \mathbf{1}\left[\lambda x_1 + (1-\lambda)x_2 < \lambda u_1 + (1-\lambda) u_2 \le \lambda y_1 + (1-\lambda) y_2\right] \\
    &\ge \mathbf{1}\left[x_1 < u_1 \le y_1\right] \mathbf{1}\left[x_2 < u_2 \le y_2\right] \\
    &= \mathbf{1}\left[x_1 < u_1 \le y_1\right]^\lambda \mathbf{1}\left[x_2 < u_2 \le y_2\right]^{1-\lambda} \\
    &= g(A_1)^\lambda g(A_2)^{1-\lambda}.
\end{align*}
Thus, the expression under the integral in \eqref{integral_form} is log-concave as a product of log-concave functions. Finally, by Theorem 6 in \cite{Prkopa}, the integral of a multivariate log-concave function w.r.t. any of its arguments is also log-concave. 
This completes the proof.
\end{proof}

\subsection{Section \ref{consist_section} Proofs} \label{consist_section_proofs}
     Along with the data matrix $X_v$, it will be convenient for this section to define its cumulative row-sum version. More formally, for
     $n\in\mathcal{I}_v$ and $t\in\mathcal{T}_v^{(n)}$, denote the indicator vector of all the parent nodes active at time $t$ within trace $n$ as 
    $z_{t}^{(n)} = [\mathbf{1}(u_i \in A_{t}^{(n)})]_{i=1}^m$, its concatenation over $t\in\mathcal{T}_v^{(n)}$ as $Z_v^{(n)}$, and the further concatenation over $n\in\mathcal{I}_v$ as $Z_{v}\in\{0, 1\}^{N_v \times m}$. 
    
    We start with a technical lemma relating the minimum eigenvalues of $X_v^{(n)}$ and $Z_v^{(n)}$. Consider arbitrary $X\in\mathbb{R}^{k\times m}$ and let $R\in\{0,1\}^{k \times k}$  denote a lower triangular matrix with $R_{ij} = \mathbf{1}[j \le i]$. Then the rows of $Z = RX$ are cumulative sums of the rows of $X$.  Let $\lambda_{\min}(\cdot)$ denote the minimal eigenvalue of a symmetric matrix.  
    \begin{lemma}\label{lemma_delta_X_eigval_ralt}
        With $X$ and $Z$ defined as above, 
        $$\lambda_{\min}(XX^\top) \le 4\lambda_{\min}(ZZ^\top).
        $$
    \end{lemma}
    \begin{proof}[Proof of Lemma \ref{lemma_delta_X_eigval_ralt}]
 Since$[RR^\top]^{-1}$ is a tri-diagonal Toeplitz matrix with conveniently computed characteristic polynomial, one can show that $\lambda_{\min}(RR^\top) = 1 /[4\cos^2(\pi/(4k + 2))] \ge 1/4$. Therefore, we obtain
        \begin{align*}
            \lambda_{\min}(ZZ^\top)&= \lambda_{\min}(RXXR^\top) \ge \lambda_{\min}(RR^{\top})\lambda_{\min}(XX^\top) \ge {1\over 4} \lambda_{\min}(XX^\top).
        \end{align*}
    \end{proof}
    
With Lemma \ref{lemma_delta_X_eigval_ralt} at hand, we are ready to prove Proposition \ref{ic_and_lt_satify_concavity}.
    \begin{proof}[Proof of Proposition \ref{ic_and_lt_satify_concavity}]
    With $n \in\mathcal{I}_v$ and $t\in\mathcal{T}_v^{(n)}$, denote
    $d_{n, t+1} (\boldsymbol{\theta}_v) := F_v(\boldsymbol{\theta}^\top_vz^{(n)}_{t}) - F_v(\boldsymbol{\theta}^\top_vz^{(n)}_{t-1})$. Additionally, let $d_{n, T_n + 1} (\boldsymbol{\theta}_v) := 1-F_v(\boldsymbol{\theta}^\top_vz^{(n)}_{T_n}) = 1-F_v(\boldsymbol{\theta}^\top_vz^{(n)}_{t(v, n)})$. Observe that for any $t \in\mathcal{T}_v^{(n)}$, we have
    \begin{align*}
        &\mathbb{P}[y^{(n)}_{t+1}=1 \mid X_v] = \mathbb{P}[B_v(A_{t-1}^{(n)}; \boldsymbol{\theta}_v^*)< U_v \le B_v(A_t^{(n)}; \boldsymbol{\theta}_v^*)\mid X_v] = d_{n, t+1}(\boldsymbol{\theta}_v^*) \quad \text{and}\\
        &\mathbb{P}[y^{(n)}_{t+1}=0 \text{ for all } t\in \mathcal{T}_v^{(n)} \mid X_v] = \mathbb{P}[B_v(A_{T_n}^{(n)}; \boldsymbol{\theta}_v^*)< U_v \mid X_v] = d_{n, T_n + 1}(\boldsymbol{\theta}_v^*).
    \end{align*}
   Note that by strict monotonicity and continuity of $F_v$, we can lower bound $d_{n, t+1}(\boldsymbol{\theta}_v), t\in \mathcal{T}_v^{(n)} \cup \{T_n\}$ uniformly for all $\boldsymbol{\theta}_v \in \tilde{\Theta}_v$ and for $\boldsymbol{\theta}^*_v$ in particular as
    $$d_{\min} = \min_{\substack{z_1, z_2 \in [0, \gamma], \\ z_1 \ge z_2 + \varepsilon}} \bigl[F_v(z_1) - F_v(z_2)] > 0.
    $$
    Moreover, with all considered threshold distributions having a log-concave density, it holds $\nabla^2 \log d_{n, t+1} \succeq 0$ for any $n\in\mathcal{I}_v, t\in\mathcal{T}_v^{(n)}\cup \{T_n\}$ per Proposition \ref{logconc}. So, we can lower-bound the expected Hessian as:
    % Moreover, due to strict monotonicity of $F_v$, we have for any $z_1 > z_2$ that 
    % $${\partial^2\over \partial z_1\partial z_2}\log (F_v(z_1) - F_v(z_2)) = -{F'_v(z_1) F_v'(z_2) \over (F_v(z_1) - F_v(z_2))^2} \ne 0.$$
    % Therefore, minimum non-zero eigenvalue of $-H_v(z_1, z_2)$ is positive everywhere on $z_1 > z_2$. Since the operation of taking the maximal eigenvalue of a matrix and $H_v$ are continuous, we can define
    % $$c'=\min_{\substack{z_1, z_2\in [0, \gamma], \\ z_1 \ge z_2 + \varepsilon}} \lambda_{\max}(-H_v(z_1, z_2)) > 0
    % $$
    % using the compactness of the minimization set.
    \begin{equation}\label{lower_bound_expected_hessian}
    \begin{aligned}
        \mathbb{E}[-\nabla^2L_v(\boldsymbol{\theta}_v)\mid X_v] &=-\sum_{n\in \mathcal{I}_v}\sum_{t\in\mathcal{T}_v^{(n)}\cup \{T_n \}}d_{n, t+1} (\boldsymbol{\theta}^*_v) \nabla^2 \log d_{n, t+1} (\boldsymbol{\theta}_v)  \\
        &\succeq -d_{\min} \sum_{n\in \mathcal{I}_v}\sum_{t\in\mathcal{T}_v^{(n)}}\nabla^2 \log d_{n, t+1}(\boldsymbol{\theta}_v),
    \end{aligned} 
    \end{equation}
    where each term in the final expression can be expressed using the chain rule in terms of $H_v(z_1, z_2) =\nabla^2_{z_1, z_2} \log (F_v(z_1) - F_v(z_2)) \in \mathbb{R}^{2 \times 2}$ and $L_{n, t} = [z^{(n)}_{t}; z^{(n)}_{t-1}] \in \mathbb{R}^{m \times 2}$ as:
    $$\nabla^2\log d_{n, t+1}(\boldsymbol{\theta}_v)  = L_{n, t}\cdot H_v(\boldsymbol{\theta}^\top_vz^{(n)}_{t}, \boldsymbol{\theta}^\top_vz^{(n)}_{t-1})\cdot L_{n, t}^\top.
    $$
    
    One can show by checking $\det(H_v) > 0$ and $H_{v, 11} < 0$ that for $F_v\sim \operatorname{Beta}(\alpha, \beta)$ with $\alpha \ge 1, \beta \ge 1$ except for the case $\alpha=\beta=1$ (uniform distribution), the Hessian $H_v(z_1, z_2)$ is strictly negative definite for any $z_1 > z_2$. In this case, by continuity of the Hessian and the mapping $A\mapsto \lambda_{\min}(A)$ taking the minimum eigenvalue of a negative definite matrix, we can define
    $$c_1 = \min_{\substack{z_1, z_2 \in [0, \gamma], \\ z_1 \ge z_2 + \varepsilon}} \lambda_{\min}\bigl(H_{v}(z_1, z_2)\bigr) > 0.
    $$
    With that, the last line of \eqref{lower_bound_expected_hessian} can be further lower bounded by 
    $$d_{\min}c_1 \sum_{n\in\mathcal{I}_v} \sum_ {t\in \mathcal{T}_v^{(n)}}L_{n, t}L_{n, t}^\top \succeq d_{\min}c_1Z_v^\top Z_v,
    $$
    where by Lemma \ref{lemma_delta_X_eigval_ralt}, the smallest eigenvalue of the RHS is lower bounded  by ${d_{\min}c_1\over 4}\lambda_{\min}(X_v^\top X_v).$
    The result follows by substituting $c_\lambda=d_{\min}c_1 / 4$.
    
    For $F_v \sim \operatorname{Unif}[0, 1]$ and $F_v\sim \operatorname{Exponential}(1)$,
    the Hessian can be computed with $g(z_1, z_2)=1/ (z_1 - z_2)^2$ and $g(z_1, z_2)= e^{-(z_1+z_2)} / (e^{-z_2} - e^{-z_1})^2$, respectively,  as
    $$H_v(z_1, z_2) = g(z_1, z_2)\begin{pmatrix}
        -1  & 1 \\
        1 & -1
    \end{pmatrix}.$$
   Note that $H_v(z_1, z_2)$ is a negative semi-definite matrix of rank 1 with eigenvector $h=(1, -1)$ and corresponding eigenvalue $\lambda_h(z_1, z_2) < 0$, satisfying $H_v(z_1, z_2)=\lambda_h(z_1, z_2) hh^\top$. By continuity of $H_v$ and the largest eigenvalue of a matrix, we can define 
   $$c_2 = \min_{\substack{z_1, z_2\in [0, \gamma], \\ z_1 \ge z_2 + \varepsilon}} -\lambda_h(z_1, z_2) > 0$$
    and further lower bound the last line of \eqref{lower_bound_expected_hessian} as follows:
    $$ d_{\min}c_2\sum_{n\in\mathcal{I}_v} \sum_ {t\in \mathcal{T}_v^{(n)}}(L_{n, t} h) (L_{n, t} h)^\top = d_{\min}c_2X_v X_v^\top.
    $$
    Letting $c_\lambda=c_2d_{\min}$ completes the proof.
    \end{proof}
    
    Before we proceed to the proof of Theorem \ref{mle_consistency_theorem}, we present its main idea.  The key step is to reformulate \eqref{opt_prob_changed_sum} as a binary classification problem, for which a finite sample result can be conveniently derived by treating the design $X_v$ as fixed and the response as the only source of randomness. The response here is the activation indicator of node $v$, defined for each trace $n\in\mathcal{I}_v$ and $t\in \mathcal{T}_{v}^{(n)}$ as
    $$y^{(n)}_{t + 1} =\mathbf{1}(v\in D_{t + 1}^{(n)}).$$ 
    Then, the log-likelihood of node $v$ in \eqref{individ_node_likelihood} can be rewritten as  \begin{equation}\label{individ_node_loglik_as_classifier}
        L_v(\boldsymbol{\theta}_v) = \sum_{n\in \mathcal{I}_v}\sum_{t\in\mathcal{T}_{v}^{(n)}} \mathbf{1}[t\le t(v, n)]\Bigl\{y_{t + 1}^{(n)}\log \bigl[p_{t+1}^{(n)}(\boldsymbol{\theta}_v)\bigr] + (1-y_{t + 1}^{(n)})\log \bigl[1 - p_{t+1}^{(n)}(\boldsymbol{\theta}_v)\bigr]\Bigr\},
    \end{equation}
    where $p_{t + 1}^{(n)}(\boldsymbol{\theta}_v)$ is the GLT transition probability defined in \eqref{trans_prob_glt} and written with new notation as 
    \begin{equation}\label{tran_prob_glt_classif}
    p_{t + 1}^{(n)}(\boldsymbol{\theta}_v) = {F_v(\boldsymbol{\theta}_v^\top z_{t}^{(n)}) - F_v(\boldsymbol{\theta}_v^\top z_{t - 1}^{(n)}) \over 1 - F_v(\boldsymbol{\theta}_v^\top z_{t - 1}^{(n)})}.
    \end{equation}
    Note that the activation time $t(v, n)$ in this framework is random, and $y_{t+1}^{(n)} = 0$ for any $t < t(v, n)$ and $y_{t+1}^{(n)} = 1$ implies $t=t(v,n)$. 
    In particular, this means that the responses $\mathbf{y}^{(n)} = \{y_{t + 1}^{(n)}: t\in \mathcal{T}_{v}^{(n)}, t\le t(v, n)\}$ are dependent conditionally on $X_v^{(n)}$ and their number is random, since we do not receive more observations after $v$ is activated. This creates additional technical complexity relative to a binary classification model. Fortunately, $\mathbf{y}^{(n)}$ are independent across $n \in \mathcal{I}_v$ due to the assumed trace independence.

Given the log-likelihood of the binary classification model in \eqref{individ_node_loglik_as_classifier}, our next goal is to establish a high probability (with respect to the randomness in $\{\mathbf{y}^{(n)}, n\in \mathcal{I}_v\}$) error bound for the model's MLE conditional on the design $X_v$. 

\begin{proof}[Proof of Theorem \ref{mle_consistency_theorem}]
    Fix a node $v\in V_c$ with $m = |P(v)|$ and $\delta \in (0, 1)$. Denote the negative log-likelihood term corresponding to the trace $n\in\mathcal{I}_v$ and time point $t\in \mathcal{T}_v^{(n)}$ in   \eqref{individ_node_loglik_as_classifier}by 
    $$\ell_{n, t}(\boldsymbol{\theta}_v) =  -\mathbf{1}[t\le t(v, n)]\Bigl\{y_{t + 1}^{(n)}\log \bigl[p_{t+1}^{(n)}(\boldsymbol{\theta}_v)\bigr] + (1-y_{t + 1}^{(n)})\log \bigl[1 - p_{t+1}^{(n)}(\boldsymbol{\theta}_v)\bigr]\Bigr\}.
    $$
    % where $p_{t+1}^{(n)} = 1- {1-F(\boldsymbol{\theta}_v^\top z_{t}^{(n)}) \over 1-F(\boldsymbol{\theta}_v^\top z_{t-1}^{(n)})}$.
    Then, \eqref{opt_prob_changed_sum} essentially solves the following problem:
    $$\hat{\boldsymbol{\theta}}_v = \arg\min_{\boldsymbol{\theta}_v \in \tilde{\Theta}_v} \ell_{N_v}(\boldsymbol{\theta}_v), \quad \text{where} \quad \ell_{N_v}(\boldsymbol{\theta}_v) ={1\over N_v} \sum_{n\in \mathcal{I}_v} \sum_{t\in\mathcal{T}_v^{(n)}}\ell_{n, t}(\boldsymbol{\theta}_v).
    $$

    Our first step is to establish a high probability bound for $\|\nabla\ell_{N_v}(\boldsymbol{\theta}_v^*)\|_2$. 
    With $f_v = F'_v$ denoting the threshold density, we can express the log-likelihood gradient of each sample as 
    \begin{equation}\label{one_sample_loglik_grad}
        \nabla\ell_{n, t}(\boldsymbol{\theta}_v) =  c_{n, t}(\boldsymbol{\theta}_v) \nabla p_{t+1}^{(n)}(\boldsymbol{\theta}_v)\mathbf{1}[t\le t(v, n)] , 
    \end{equation}
    where $c_{n, t}(\boldsymbol{\theta}_v) = { p_{t+1}^{(n)} - y_{t+1}^{(n)} \over p_{t+1}^{(n)}(1-p_{t+1}^{(n)})}$ and 
    $$\nabla p_{t+1}^{(n)}(\boldsymbol{\theta}_v) ={[1-F_v(\boldsymbol{\theta}_v^\top z_{t-1}^{(n)})]  f_v(\boldsymbol{\theta}_v^\top z_{t}^{(n)})z_{t}^{(n)} - [1-F_v(\boldsymbol{\theta}_v^\top z_{t}^{(n)})]f_v(\boldsymbol{\theta}_v^\top z_{t-1}^{(n)})z_{t-1}^{(n)} \over [1-F_v(\boldsymbol{\theta}_v^\top z_{t-1}^{(n)})]^2}.
    $$
    Since $f_v$ is continuous, it achieves its maximum $f_{\max}$ on the interval $[0, \gamma]$ and minimum $f_{\min}$ on $[\varepsilon, \gamma]$. Note that both quantities are positive as both intervals lie in the support of $f_v$. Thus, by definition of $\Tilde{\Theta}$ and feasibility, we have $f_v(\boldsymbol{\theta}_v^\top z_{t-1}^{(n)}) \in [0, f_{\max}]$ and 
    $f_v(\boldsymbol{\theta}_v^\top z_{t}^{(n)}) \in [f_{\min}, f_{\max}]$. Then, for each coordinate $j=1, \ldots, m$ of $\nabla p_{t+1}^{(n)}(\boldsymbol{\theta}_v)$ the triangular inequality implies
    \begin{equation}      \label{nabla_p_bound_coord}
        |[\nabla p_{t+1}^{(n)}(\boldsymbol{\theta}_v)]_j| \
        \le  {2 f_{\max}  \over [1 - F_v(\gamma)]^2} .
    \end{equation}
    To bound $c(\boldsymbol{\theta}_v)$, we first show that $p_{t+1}^{(n)}$ is bounded away from 0 and 1.
    Note that $p_{t+1}^{(n)}(\boldsymbol{\theta}_v) \le F_v(\gamma)$. For the lower bound, if $z_{t-1}^{(n)} = \boldsymbol{0}_m$,
    then $p_{t+1}^{(n)}(\boldsymbol{\theta}_v) \ge F_v(\varepsilon)$ and otherwise, by the mean-value theorem and feasibility there is $\xi \in  [\varepsilon, \gamma]$ such that $p_{t+1}^{(n)}(\boldsymbol{\theta}_v) = f_v(\xi) \boldsymbol{\theta}_v^\top(z_{t}^{(n)} - z_{t-1}^{(n)}) \ge f_{\min}\varepsilon$.
    To summarize,
    \begin{equation*}
        F_v(\gamma) \ge p_{t+1}^{(n)}(\boldsymbol{\theta}_v) \ge \min\{f_{\min}\varepsilon, F_v(\varepsilon)\},
    \end{equation*}
   which implies together with $|p_{t+1}^{(n)}(\boldsymbol{\theta}_v)  - y_{t+1}^{(n)}| \le 1$ that 
    \begin{equation}\label{grad_const_bound}
        |c(\boldsymbol{\theta}_v)| = \Bigl|{p_{t+1}^{(n)}(\boldsymbol{\theta}_v) - y_n\over p_{t+1}^{(n)}(\boldsymbol{\theta}_v)(1-p_{t+1}^{(n)}(\boldsymbol{\theta}_v))}\Bigr|  \le {1 \over \min\{f_{\min}\varepsilon, F_v(\varepsilon)\}(1-F_v(\gamma))}.
    \end{equation}
Combining \eqref{nabla_p_bound_coord} and \eqref{grad_const_bound}, we obtain
\begin{equation}\label{grad_loglik_bound_coord}
|[\nabla\ell_{n, t}(\boldsymbol{\theta}_v^*)]_j | \le {2f_{\max} \over \min\{f_{\min}\varepsilon, F_v(\varepsilon)\}(1-F_v(\gamma))^3} = C_0'(\varepsilon, \gamma, F_v).
\end{equation}
To bound $|\sum_{n\in \mathcal{I}_v}\sum_{t\in\mathcal{T}^{(n)}_v}[\nabla\ell_{n, t}(\boldsymbol{\theta}_v)]_j|$, we will use Azuma-Hoeffding inequality exploiting the convenient conditional structure of the subsequent terms in the sum \footnote{The following derivation is inspired by the proof of Lemma 3 in \citep{InferringGraphs2015}.}. Denote $Y_{n, t} := [\nabla\ell_{n, t}(\boldsymbol{\theta}_v^*)]_j$ and note that
it is identically zero if $[\nabla p_{t +1}^{(n)}(\boldsymbol{\theta}_v^*)]_j = 0$ and otherwise satisfies
\begin{align*}
        \mathbb{E}[Y_{n, t} \ |\ \{Y_{n, \tau}\}_{\tau < t}, X_v^{(n)}]  
        &= \mathbb{E}[Y_{n, t} \ |\ \{y^{(n)}_{\tau + 1} = 0\}_{\tau < t},  X_v^{(n)}]  \\
        &\propto \mathbb{E}[y_{t + 1}^{(n)} - p_{t +1}^{(n)}(\boldsymbol{\theta}_v^*) \ | \ \{y^{(n)}_{\tau + 1} = 0\}_{\tau < t}, X_v^{(n)}] = 0.
\end{align*}
On the other hand, $Y_{n, t}$'s with different $n$ are independent conditional on $Z_v$, so partial sums of $\sum_n\sum_t Y_{n, t}$ form a martingale. Since each term in the sum is almost surely bounded by \eqref{grad_loglik_bound_coord}, we have by Azuma-Hoeffding for any $z > 0$:
$$\mathbb{P}\bigl({1\over N_v}\bigl|\sum_{n\in \mathcal{I}_v}\sum_{t\in\mathcal{T}^{(n)}_v} Y_{n,t} \bigr|\ge  z\bigr) \le 2\exp(-z^2N_v / 2C_0').
$$
Using the union bound over the coordinates $j=1,\ldots, m$ and inequality relating the $\ell_2$ and $\ell_\infty$ norms, we have 
$$\mathbb{P}\bigl(\|\nabla\ell_{N_v}\|_2\ge  z\bigr) \le \mathbb{P}\bigl(\|\nabla\ell_{N_v}\|_\infty \ge  z / \sqrt{m}\bigr) \le 2m\exp(-z^2N_v / 2mC_0') 
$$
Denoting the RHS by $\delta / 2$ and solving for $z$, we obtain:
\begin{equation}\label{grad_concetration}
    \mathbb{P}\bigl(\|\nabla\ell_{N_v}\|_2\ge  \sqrt{ {2C_0'm \over N_v}\log{4m\over \delta}}\bigr) \le \delta / 2
\end{equation}

% Therefore, conditonal on Thus, by Popoviciu's inequality and \eqref{grad_loglik_bound}, we have $\mathbb{E}[\|\nabla\ell_n(\boldsymbol{\theta}_v^*) \|_2^2 \ \mid\ z_{t}^{(n)}, z_{t-1}^{(n)}] \le C_1^2m$. With that, we can apply the vector Bernstein inequality (Lemma 18, \citep{kohler2017vector_bernstein}) to obtain for any $0 < z < C_1\sqrt{m}$:
% $$\mathbb{P}\Bigl[\|\nabla\ell_{N_v}(\boldsymbol{\theta}_v^*)\|_2 \ge z\Bigr] \le \exp\Bigl(-{N_vz^2 \over 8C_1^2m} + {1\over 4}\Bigr)
% $$
% Take $z :=C_1 \sqrt{{8m\over N_v} ({1 \over  4} + \log {2\over \delta})}$, which satisfies the constraint since we assumed $N_v > 2 +8\log 2/\delta$. Then we can rewrite the concentration as follows
% \begin{equation}\label{gradient_concentration}
%     \mathbb{P}\Bigl[\|\nabla\ell_{N_v}(\boldsymbol{\theta}_v^*)\|_2 \le C_1 \sqrt{{m\over N_v} \Bigl(2 + 8\log {2\over \delta}\Bigr)}\Bigr] \ge 1 - \delta / 2.
% \end{equation}

Our next step is to demonstrate that the Hessian of $\ell_{N_v}(\boldsymbol{\theta}_v)$ is positive definite on the whole $\Tilde{\Theta}_v$ with high probability. 
% By Assumption \ref{invertible_hessian_assump}, we have $\mathbb{E}[\nabla^2\ell_{N_v}(\boldsymbol{\theta}_v)] \succeq \lambda_{\min}I_m$.
% so by Weil's inequality, we are left to show that $\|\nabla^2\ell_{N_v}(\boldsymbol{\theta}_v) - \mathbb{E}[\nabla^2\ell_{N_v}(\boldsymbol{\theta}_v)]\|_2 < \lambda_{\min}$ with high probability.
Denote $H_{n} = \nabla^2[\sum_{t\in \mathcal{T}_v^{(n)}} \ell_{n, t}(\boldsymbol{\theta}_v)]$ so that by Assumption \ref{invertible_hessian_assump}, it holds $\mathbb{E}[{1\over N_v}\sum_{n\in\mathcal{I}_v} H_n] = \mathbb{E}[\nabla^2\ell_{N_v}(\boldsymbol{\theta}_v)] \succeq \lambda_{\min}I$. By telescoping, each $H_n$ can be conveniently rewritten with $t(v, n)$  denoted by $t^*$ for brevity:
% First, we compute the Hessian for one sample. By differentiating \eqref{one_sample_loglik_grad}, we have
\begin{equation}  \label{init_decomp_hessian}
     H_{n} = -y_{t^* + 1}^{(n)} \nabla^2\log [F_v(\boldsymbol{\theta}_v^\top z_{t^*}^{(n)}) - F_v(\boldsymbol{\theta}_v^\top z_{t^*-1}^{(n)})] -(1-y_{t^* + 1}^{(n)}) \nabla^2\log [1-F_v(\boldsymbol{\theta}_v^\top z_{t^*}^{(n)})] 
\end{equation}  
By Assumption \ref{invertible_hessian_assump}, the density $f_v$ is log-concave, implying, per Proposition \ref{logconc}, that $H_n$ is positive semidefinite. 
Our goal now is to apply a Chernoff-type concentration on the sum of independent positive semidefinite matrices $H_n$ to show that it is close to its positive definite expectation, and thus, with high probability, is positive definite itself. For that, we also need the largest eigenvalue of each $H_n$ to be bounded. Denote the two additive terms in the panel above by $H_{n, 1}$ and $H_{n, 2}$ , so that $H_n = H_{n, 1} + H_{n, 2}$. By the chain rule, the spectral norm of the second one is dominated  by
$$\|H_{n, 2}\|_2 \le \|z_{t^*}^{(n)}z_{t^*}^{(n)\top}\|_2\sup_{z \in [\varepsilon, \gamma]} |{d\over dz^2}\log [1 - F_v(z)]| \le C_2(\varepsilon, \gamma, F_v)m,
$$
where the supremum is achieved since the function inside is continuous on the interval by Assumption \ref{assump2}.
Again, by the chain rule and submultiplicativity of the norm, we also have
$$\|H_{n, 1}\|_2 \le \| [z_{t^*}^{(n)};z_{t^*-1}^{(n)}]\|^2_2 \sup_{ \substack{z_1, z_2\in[0, \gamma]\\ \ z_1 \ge z_2 + \varepsilon}} \|\nabla^2_{z_1, z_2} \log[F_v(z_1) - F_v(z_2)]\|_2 \le C_1(\varepsilon, \gamma, F_v)m
$$
where we again used the fact that a continuous function on a compact set achieves its maximum. Therefore, with $c_0' := C_1 + C_2 $, we have $\|H_n\|_2\le c_0'm$. So, by Corollary 5.2 in \citep{Tropp_2011}, we have the following concentration for any $z \in [0, 1]$:
$$\mathbb{P}\bigl( \bigl\|{1\over N_v}\sum_{n\in\mathcal{I}_v} H_n(\boldsymbol{\theta}_v) \|_2 \le \lambda_{\min} z \bigr) \le m\exp[-(1-z)^2 \lambda_{\min} N_v/ 2c_0'm].
$$
Set $z := 1/2$ and $c_0 = 8c_0'$. Then, the corresponding upper bound on the probability above is dominated by $\delta / 2$ when the assumed condition $N_v \ge {c_0m \over \lambda_{\min}}  \log{2m\over \delta}$ is satisfied:
\begin{equation}\label{hessian_concentration}
   \mathbb{P}\bigl(\bigl\|\nabla^2\ell_{N_v}(\boldsymbol{\theta}_v) \bigr\|_2 \le {\lambda_{\min} \over 2} \bigr) \le m\exp[-\lambda_{\min} N_v/ c_0m] \le \delta / 2.
\end{equation}
With that, the union bound implies that the gradient concentration in \eqref{grad_concetration} and the Hessian concentration in \eqref{hessian_concentration} hold together with probability at least $1-\delta$. Assume that both of these events hold. Expanding $\ell_{N_v}$ at $\boldsymbol{\theta}_v^*$ gives for some $\tilde{\boldsymbol{\theta}}_v = \boldsymbol{\theta}_v^* + z(\hat{\boldsymbol{\theta}}_v - \boldsymbol{\theta}_v^*)$  with $z\in [0, 1]$:
% Then by $\lambda_{\min} / 2$-convexity of $\ell_{N_v}(\boldsymbol{\theta}_v)$ on the whole parameter space, we can write the first-order condition for $\boldsymbol{\theta}_v^*$ and $\hat{\boldsymbol{\theta}}_v$, which both lie in $\tilde{\Theta}_v$ by Assumption \ref{assump2}:
$$ \ell_{N_v}(\hat{\boldsymbol{\theta}}_v) = \ell_{N_v}(\boldsymbol{\theta}_v^*) + \nabla\ell_{N_v}(\boldsymbol{\theta}_v^*)^\top (\hat{\boldsymbol{\theta}}_v - \boldsymbol{\theta}_v^*) + {1\over 2}(\hat{\boldsymbol{\theta}}_v - \boldsymbol{\theta}_v^*)^\top\nabla^2\ell_{N_v}(\tilde{\boldsymbol{\theta}}_v)(\hat{\boldsymbol{\theta}}_v - \boldsymbol{\theta}_v^*).
$$
Since $\hat{\boldsymbol{\theta}}_v, \boldsymbol{\theta}_v^* \in \tilde{\Theta}_v$ by Assumption \ref{assump2} and $\tilde{\Theta}_v$ is a convex set by definition, we have  $\tilde{\boldsymbol{\theta}}_v \in \tilde{\Theta}_v$, so  the quadratic term above is lower bounded by ${\lambda_{\min}  \over 4}\|\hat{\boldsymbol{\theta}}_v - \boldsymbol{\theta}_v^*\|_2$.
On the other hand, $\ell_{N_v}(\boldsymbol{\theta}_v^*) \ge \ell_{N_v}(\hat{\boldsymbol{\theta}}_v)$ by optimality, which implies together with the Cauchy-Schwarz inequality 
$$\|\nabla\ell_{N_v}(\boldsymbol{\theta}_v^*)\|_2\|\hat{\boldsymbol{\theta}}_v - \boldsymbol{\theta}_v^*\|_2  \ge -\nabla\ell_{N_v}(\boldsymbol{\theta}_v^*)^\top (\hat{\boldsymbol{\theta}}_v - \boldsymbol{\theta}_v^*) \ge {\lambda_{\min} \over 4}  \|\hat{\boldsymbol{\theta}}_v - \boldsymbol{\theta}_v^*\|_2^2.
$$
The needed bound is obtained by dividing through $\|\hat{\boldsymbol{\theta}}_v - \boldsymbol{\theta}_v^*\|_2$ and plugging in the gradient concentration from \eqref{grad_concetration} with $C_0 = 4\sqrt{2C_0'}$.
\end{proof}

We conclude this section with a proof of the asymptotic normality result:
\begin{proof}[Proof of Proposition \ref{asympt_normality_propos}]
    By Theorem 5.1 in \citep{lehmann1998theory}, the MLE $\hat{\boldsymbol{\theta}}$ of $\boldsymbol{\theta}^*$ is consistent and satisfies
    \begin{equation}\label{normal_all_params}
        \sqrt{N}(\hat{\boldsymbol{\theta}} - \boldsymbol{\theta}^*) \stackrel{\mathbb{P}_{\boldsymbol{\theta}^*}}{\longrightarrow} \mathcal{N}(0, \Sigma(\boldsymbol{\theta}^*)) \quad \text{as} \quad N\rightarrow \infty,
    \end{equation}
    where $\Sigma^{-1}(\boldsymbol{\theta}) = \mathbb{E}_{\mathcal{D}\sim\mathbb{P}_{\boldsymbol{\theta}^*}}[-\nabla^2 \log\mathbb{P}_{\boldsymbol{\theta}}(\mathcal{D})]$,
    if (a) the distribution family $\{\mathbb{P}_{\boldsymbol{\theta}}, \boldsymbol{\theta}\in \tilde{\Theta}\}$ is identifiable, (b) the traces are independent and have common support, (c)  $\boldsymbol{\theta}^*$ lies in the interior of $\tilde{\Theta}$, (d) $\mathbb{P}_{\boldsymbol{\theta}}$ is three times differentiable wrt $\boldsymbol{\theta}$ in an open neighborhood $\omega$ around $\boldsymbol{\theta}^*$, (e) the first and second order partial derivatives of $\log \mathbb{P}_{\boldsymbol{\theta}}(\mathcal{D})$ wrt $\boldsymbol{\theta}$ are dominated by an integrable $g(\mathcal{D})$, (f) third-order partial derivatives of $\log \mathbb{P}_{\boldsymbol{\theta}}(\mathcal{D})$ wrt $\boldsymbol{\theta}$ are bounded on $\omega$, and (g) $\Sigma$ is positive definite on $\omega$. We can verify all of these conditions:  (a) holds due to the assumed identifiability condition of Theorem \ref{identif_theorem};  (b) holds because traces are assumed independent and they have a common support of all feasible traces $\mathcal{F}(G)$ with $\mathbb{P}^0(D_0) > 0$ due to Lemma \ref{pos_trace_prob_lemma}; (c) holds by Assumption \ref{assump2}; (d) holds for any $\boldsymbol{\theta} \in \tilde{\Theta}$ since, according to \eqref{individ_trace_prob}, the trace likelihood is a composition of a linear transformation of $\boldsymbol{\theta}$ and $\{F_v\}_{v\in V_c}$, which are three times differentiable by Assumption \ref{assump2}; (e) holds because first- and second-order partial derivatives of $\log \mathbb{P}_{\boldsymbol{\theta}}$ are continuous on the compact $\tilde{\Theta}$ by Assumption \ref{assump2}
        and thus achieve their maximum and minimum values; 
    (f) holds because third-order partial derivatives are also continuous by Assumption  \ref{assump2} and thus  bounded on the compact $\tilde{\Theta}$; and finally (g) 
  holds by Assumption \ref{local_convexity_assumop} and continuity of the Hessian for the $\omega$ chosen as a ball of a sufficiently small radius.

Thus  $\hat{\boldsymbol{\theta}}$ is consistent and by continuity of $\Sigma^{-1}(\boldsymbol{\theta})$ and the continuous mapping theorem we have that 
$$\Sigma^{-1}(\hat{\boldsymbol{\theta}}) \rightarrow \Sigma^{-1}(\boldsymbol{\theta}^*) \quad \text{in probability as} \ N\rightarrow\infty.$$
On the other hand, by the law of large numbers, we have for any $\boldsymbol{\theta}\in\tilde{\Theta}$ that
$$ \hat{\Sigma}^{-1}(\boldsymbol{\theta}) := -{1\over N}\sum_{n=1}^N\nabla^2\log \mathbb{P}_{\boldsymbol{\theta}}(\mathcal{D}_n) \rightarrow \Sigma^{-1}(\boldsymbol{\theta}) \quad \text{in probability as} \ N\rightarrow\infty.
$$
Combining these two observations, we have $\hat{\Sigma}^{-1}(\hat{\boldsymbol{\theta}})  \rightarrow \Sigma^{-1}(\boldsymbol{\theta}^*)$ in probability. So, by Assumption \ref{local_convexity_assumop}, $\hat{\Sigma}^{-1}(\hat{\boldsymbol{\theta}})$ is positive definite with probability tending to one.  Slutsky's theorem and \eqref{normal_all_params} implies
$$\sqrt{N}\hat{\Sigma}^{-1/2}(\hat{\boldsymbol{\theta}}) (\hat{\boldsymbol{\theta}} - \boldsymbol{\theta}^*) \stackrel{\mathbb{P}_{\boldsymbol{\theta}^*}}{\longrightarrow} \mathcal{N}(0, I_{|E|}) \quad \text{as} \quad N\rightarrow \infty.
$$
Finally, the node-wise asymptotic result follows from the block-diagonal structure of $\hat{\Sigma}^{-1}(\hat{\boldsymbol{\theta}})$, which consists of the blocks  $-{1\over N}\nabla^2L_v(\hat{\boldsymbol{\theta}}_v), v\in V_c$.
\end{proof}

\subsection{Section \ref{IM_section} Proofs} \label{IM_section_proofs}

We start with the following technical lemma formulating a convenient equivalent definition of concavity:
\begin{lemma}\label{GLT_main_lemma}
   Consider the cdf $F$ of the distribution supported on $[0, h]$. Then the condition
   \begin{equation}\label{F_submod}
       F(x + b) - F(x) \ge F(y + b) - F(y)
   \end{equation}
holds for all  triples $(x, y, b)$ with $0\le x \le y \le y + b \le h$
if and only if $F$ is concave on $[0, h]$. 
\end{lemma}
\begin{proof}[Proof of Lemma \ref{GLT_main_lemma}]
 (Necessity) It is enough to verify the ``midpoint'' concavity condition  
    $$F\left({x' + y' \over 2}\right) \ge {F(x') + F(y') \over 2}, \quad 0 \le x' \le y' \le C
    $$
    since for bounded functions (cdf is bounded between 0 and 1), it is known to be equivalent to concavity.
    Plugging $x = x', b = (y'-x') / 2, y=(y' + x') /2$ into \eqref{F_submod} and rearranging the terms implies the needed inequality.
    % $$F\left({y' + x'\over 2}\right) - F(x') \ge  F(y') - F\left({y' + x'\over 2}\right)
    % $$\lim_{b \downarrow 0}{F(x + b) - F(x)\over b} \ge \lim_{b \downarrow 0}{F(y + b) - F(y)\over b} \quad \Leftrightarrow \quad q(x) = F'(x) \ge F'(y) = q(y).
    % $$
    \\
    \\
    (Sufficiency) Without loss of generality, assume $x + b \le y$ (otherwise, repeat the proof for $x' = x, y' = x + b, b' = y - x$) .
    Consider the equivalent definition of a concave function (Lemma 2.1 in \cite{conv_equiv}), which states that for any $x_1 < x_2 < x_3$ it should hold
    $${F(x_2) - F(x_1)\over x_2 - x_1} \ge {F(x_3) - F(x_2)\over x_3 - x_2},
    $$
    and use this inequality with $(x, x+b, y)$ and $(x+b, y, y+b)$ to obtain what was needed:
     $${F(x + b) - F(x)\over b} \ge {F(y) - F(x + b)\over y - x - b} \ge {F(y + b) - F(y)\over b}.
    $$

\end{proof}

\begin{proof}[Proof of Theorem \ref{GLT_submodular_theorem}] 
      (Sufficiency) Consider a graph $G=(V, E)$ and an arbitrary GLT model on it. By Theorem 1 in \cite{mossel2009submodularity}, it is enough to show that all threshold functions $f_v(S) = F_v\left(\sum_{u\in S} b_{u,v}\right)$ are monotone and submodular. Monotonicity holds trivially since all edge weights are nonnegative and $F_v$ is non-decreasing. To establish submodularity, we need to check for $S^{\prime} \subset S \subset P(v)$ and $w \notin S$ that
    $$F_v\Bigl(\sum_{u\in S \cup \{w\}} b_{u,v}\Bigr) - F_v\Bigl(\sum_{u\in S} b_{u,v}\Bigr) \le F_v\Bigl(\sum_{u\in S' \cup \{w\}} b_{u,v}\Bigr) - F_v\Bigl(\sum_{u\in S'} b_{u,v}\Bigr).$$
This follows by applying Lemma \ref{GLT_main_lemma} to $b:=b_{w, v}, x := \sum_{u\in S'} b_{u,v}, y:= \sum_{u\in S} b_{u,v}$.  The condition  $0 \le x\le y \le y + b\le h_v$ follows from weights' non-negativity:
$$0\le \sum_{u \in S'} b_{u, v} \le \sum_{u \in S} b_{u, v} \le \sum_{u \in S \cup \{w\}} b_{u, v} \le \sum_{u \in P(v)} b_{u, v}  \le h_v.$$
(Necessity) Let $F$ be the cdf of an arbitrary distribution supported on $[0, h]$. By Lemma \ref{GLT_main_lemma}, there exist $(x, y, b)$ with $0\le x \le y\le y + b \le h$, such that $F(x + b) - F(x) < F(y+b) - F(y)$. Let $G$ be a star graph of in-degree 3 as defined in Figure \ref{fig:instar_graph}. Consider an instance of the GLT model on $G$ with weights $\boldsymbol{\theta} =(x, y-x, b)$ and $F$ as the cdf of node 4. Then with notations of Definition \ref{submod_def}, submodularity is violated for $S=\{1, 2\}$, $S' = \{1\}$,  and $v=3$:
$$F(y + b) - F(y) = \sigma(S \cup \{v\}) - \sigma(S) >\sigma(S'\cup \{v\}) - \sigma(S') = F(x + b) - F(x).
$$
\end{proof}

Before proceeding to prove Proposition \ref{spread_error_bound_propos}, we introduce a preliminary lemma that establishes a general bound on the discrepancy between the spreads from the IM solutions obtained under the ground truth and estimated models, which holds for an arbitrary graph. 
\begin{lemma}\label{spread_diff_bound_lemma}
    For an arbitrary graph $G=(V, E)$, it holds with the notations of Proposition \ref{spread_error_bound_propos}:
    $$|\sigma_{\boldsymbol{\theta}}(S^*(\boldsymbol{\theta})) - \sigma_{\boldsymbol{\theta}}(S^*(\hat{\boldsymbol{\theta}}))| \le 2\max_{|S| \le k} |\sigma_{\boldsymbol{\theta}}(S) - \sigma_{\hat{\boldsymbol{\theta}}}(S)|.
    $$
\end{lemma}
\begin{proof}[Proof of Lemma \ref{spread_diff_bound_lemma}]
Denote for brevity $S^* = S^*(\boldsymbol{\theta})$ and $\hat{S} = S^*(\hat{\boldsymbol{\theta}})$.
Then, by the triangle inequality, we have
\begin{equation}
     |\sigma_{\boldsymbol{\theta}}(S^*) - \sigma_{\boldsymbol{\theta}}(\hat{S})| \le |\sigma_{\boldsymbol{\theta}}(S^*) - \sigma_{\hat{\boldsymbol{\theta}}}(\hat{S})| +  |\sigma_{\hat{\boldsymbol{\theta}}}(\hat{S}) - \sigma_{\boldsymbol{\theta}}(\hat{S})| . 
\end{equation}    
To bound the first term, we use the definition of $S^*$ and $\hat{S}$, followed by a standard maximum inequality 
\begin{align*}
    |\sigma_{\boldsymbol{\theta}}(S^*) - \sigma_{\hat{\boldsymbol{\theta}}}(\hat{S})| &\le |\max_{|S|\le k} \sigma_{\boldsymbol{\theta}}(S) - \max_{|S|\le k} \sigma_{\hat{\boldsymbol{\theta}}}(S)|\le \max_{|S|\le k} |\sigma_{\boldsymbol{\theta}}(S)  - \sigma_{\hat{\boldsymbol{\theta}}}(S)|.
\end{align*}
The second one can be bounded by observing that the difference evaluated at a given $\hat{S}$ with $|\hat{S}|\le k$ is dominated by its maximum across all possible sets $S$ of size not exceeding $k$:
$$|\sigma_{\hat{\boldsymbol{\theta}}}(\hat{S}) - \sigma_{\boldsymbol{\theta}}(\hat{S})| \le \max_{|S| \le k} |\sigma_{\hat{\boldsymbol{\theta}}}(S) - \sigma_{\boldsymbol{\theta}}(S)|.
$$
Combining the two bounds completes the proof.
\end{proof}

\begin{proof}[Proof of Proposition \ref{spread_error_bound_propos}]
    Consider a directed bipartite graph $G = (V, E)$ with node set $V$ consisting of child and parent node subsets, denoted respectively as $V_c$ and $V_p$, so that $V = V_c \sqcup V_p$ and $E \subset \{(u, v): u\in V_p, \ v \in V_c\}$. By Lemma \ref{spread_diff_bound_lemma}, we only need to show that for any node subset $S \subset V$ with $|S|\le k$, it holds
    $$|\sigma_{\hat{\boldsymbol{\theta}}}(S) - \sigma_{\boldsymbol{\theta}}(S)| \le L \|\hat{\boldsymbol{\theta}} - \boldsymbol{\theta}\|_1.
    $$
    Conveniently, we can explicitly compute the influence function of an arbitrary parameter $\boldsymbol{\theta}$ as 
    $$\sigma_{\boldsymbol{\theta}}(S) = |S| + \sum_{v\in V_c \setminus S} F_v(B_v(S; \boldsymbol{\theta}_v)).
    $$
    Indeed, nodes in $V_c \cap S$ cannot activate anyone else, and nodes in $V_p \cap S$ can also propagate the influence to their children. From that, the needed bound follows by sequentially applying the triangular inequality together with the Lipschitz property of $F_v$:
    \begin{align*}
    |\sigma_{\hat{\boldsymbol{\theta}}}(S) - \sigma_{\boldsymbol{\theta}}(S)| &\le \sum_{v\in V_c \setminus S}|F_v(B_v(S; \boldsymbol{\theta}_v)) - F_v(B_v(S; \hat{\boldsymbol{\theta}}_v))| \\
    &\le L \sum_{v\in V_c \setminus S}| B_v(S; \boldsymbol{\theta}_v) - B_v(S; \hat{\boldsymbol{\theta}}_v)| \\
    &\le L\sum_{v\in V_c \setminus S} \sum_{u\in S \cap P(v)} |b_{u, v} - \hat{b}_{u, v}|\\ 
    &\le L \|\hat{\boldsymbol{\theta}} - \boldsymbol{\theta}\|_1.
    \end{align*}
\end{proof}

\end{document}